%% file: main.tex
\documentclass[11pt]{article}
\usepackage[T1]{fontenc}
\usepackage{geometry} 
\usepackage{graphicx}
\usepackage{amsmath,amssymb,amsthm,amsfonts,dsfont,mathtools}
\usepackage{relsize}
\usepackage[margin=1cm]{caption}
\usepackage{wrapfig}
\usepackage{frame,color}
\usepackage{environ,subfig}
\usepackage{hyperref}
\usepackage{xspace}
\usepackage{framed}
\usepackage{comment}
\usepackage{changepage}
\usepackage[linesnumbered,boxed]{algorithm2e}
\usepackage{algorithmic}
\usepackage{mathrsfs}
\usepackage{xcolor}
\makeatletter
\renewenvironment{framed}{%
 \def\FrameCommand##1{\hskip\@totalleftmargin
 \fboxsep=\FrameSep\fbox{##1}
     \hskip-\linewidth \hskip-\@totalleftmargin \hskip\columnwidth}%
 \MakeFramed {\advance\hsize-\width
   \@totalleftmargin\z@ \linewidth\hsize
   \@setminipage}}%
 {\par\unskip\endMakeFramed}
\makeatother

\newtheorem{theorem}{Theorem}
\newtheorem{lemma}{Lemma}

\newtheorem{claim}{Claim}
\newtheorem{definition}{Definition}

\input{commands.tex}

\NewEnviron{fcodebox}[1]%
{\fbox{%
\begin{minipage}{#1\linewidth}
\begin{codebox}
\BODY
\end{codebox}
\end{minipage}
}}

\title{Distributed Edge Coloring with Small Palettes and\\
a Special Case of the Constructive \Lovasz{} Local Lemma\thanks{A preliminary version of this paper was presented at 29th Annual ACM-SIAM Symposium on Discrete Algorithms, SODA 2018, New Orleans, LA, USA, January 7--10, 2018.
Supported by NSF grants CCF-1514383 and CCF-1637546
and ERC Grant No. 336495 (ACDC).}}

\author{Yi-Jun Chang\\ University of Michigan
\and
Qizheng He\\
IIIS, Tsinghua University
\and
Wenzheng Li\\
IIIS, Tsinghua University
\and
Seth Pettie\\
University of Michigan
\and
Jara Uitto\\
ETH Z\"{u}rich \& University of Freiburg
}

\begin{document}
\date{}
\maketitle
\thispagestyle{empty}
\setcounter{page}{0}
\begin{abstract}
The complexity of distributed edge coloring depends heavily on the
\emph{palette size} as a function of the maximum degree $\Delta$.  In this paper we explore
the complexity of edge coloring in the $\LOCAL$ model in different palette size regimes.
Our results are as follows.

\begin{description}
\item[Lower Bounds:] First, we simplify the \emph{round elimination} technique of Brandt et al.~\cite{BrandtEtal16}
and prove that $(2\Delta-2)$-edge coloring requires $\Omega(\log_\Delta \log n)$ time w.h.p.~and
$\Omega(\log_\Delta n)$ time deterministically, \emph{even on trees}.  
Second, we show that a natural approach to computing $(\Delta+1)$-edge colorings (Vizing's theorem)
via extending partial colorings by iteratively re-coloring parts of the graph 
in the worst case requires recoloring
subgraphs of diameter $\Omega(\Delta\log n)$.

\item[Upper Bounds on General Graphs:]  We give a randomized edge coloring algorithm that can use palette sizes as small as
$\Delta + \tilde{O}(\sqrt{\Delta})$, which is a natural barrier for randomized approaches.
Our algorithm employs a permissive version of the constructive \Lovasz{} local lemma as a black box. The runtime of algorithm varies for different choices of $\Delta$ and palette size. For example, our algorithm computes a $(1+\epsilon)\Delta$-edge coloring in  $O(\log n)$ time when  $\epsilon \geq (\log^3 \Delta) / \sqrt{\Delta}$, or $O(\log_{\Delta} n) + (\log \log n)^{3 + o(1)}$ time when $\epsilon = \Omega(1)$.

\item[Upper Bounds on Trees:] 
We show that the $\Omega(\log_\Delta \log n)$ lower bound can be nearly matched on trees.
To establish this result,
we develop a new distributed \Lovasz{} local lemma algorithm for \emph{tree-structured
dependency graphs}. Specifically, our $(1+\epsilon)\Delta$-edge coloring algorithm for
trees takes
$O(\log(1 / \epsilon)) \cdot \max\{\fr{\log\log n}{\log\log\log n},\, \log_{\log \Delta} \log n\}$ time when  $\epsilon \geq (\log^3 \Delta) / \sqrt{\Delta}$, or
$O\left( \max\{\fr{\log\log n}{\log\log\log n},\, \log_{\Delta} \log n\}\right)$ time  when $\epsilon = \Omega(1)$.
\end{description}
\end{abstract}

\noindent\rule{\textwidth}{0.4pt}

{\small \color{purple}
\paragraph{Erratum.}
In the previous version of the paper, Theorem~\ref{thm:tree-decomp2} claimed a strong-diameter $(1,O(\log_{\lambda/k} s + d/k),O(\lambda^2),0)$-network decomposition of $T^k$. We thank Sebastian Brandt and Ananth Narayanan for identifying errors in the proof and bringing them to our attention. We have corrected the proof, at the cost of weakening the guarantee to a weak-diameter $(1,O(\log_{\lambda/k} s + d/k),O(\lambda^2),1)$-network decomposition of $T^k$. This change does not affect any application of the theorem in the paper. In particular, all claimed bounds for edge coloring and the distributed \Lovasz{} local lemma remain unchanged.
}

\newpage

\tableofcontents
\setcounter{page}{0}
\thispagestyle{empty}
\newpage

\input{intro}

\input{lowerBound}

\input{augment}

\input{randAlg}

\input{estimate}

\input{treeLLL}

\input{decomposition}

\input{edgeColorAlg}

\input{conclusion}

\bibliographystyle{abbrv}
\bibliography{references}

\newpage

\appendix

\input{concentration}

\end{document}

%% file: commands.tex
\newcommand{\ignore}[1]{}

\newcommand{\Expect}{\operatorname{E}}
\newcommand{\Var}{\operatorname{Var}}
\newcommand{\Cov}{\operatorname{Cov}}
\newcommand{\Prob}{\operatorname{Pr}}
\newcommand{\bydef}{\stackrel{\rm def}{=}}

\newcommand{\sqbrack}[1]{\left[ #1 \right]}

\newcommand{\ceil}[1]{\left\lceil #1 \right\rceil}
\newcommand{\floor}[1]{\left\lfloor #1 \right\rfloor}
\newcommand{\f}[2]{\frac{#1}{#2}}
\newcommand{\fr}[2]{\mbox{$\frac{#1}{#2}$}}

\newcommand{\poly}{{\operatorname{poly}}}

\newcommand{\dist}{\operatorname{dist}}
\newcommand{\Color}{\operatorname{Color}^\star}

\newcommand{\degree}{\operatorname{deg}}
\newcommand{\size}{\operatorname{size}}

\newcommand{\rb}[2]{\raisebox{#1 mm}[0mm][0mm]{#2}}
\newcommand{\istrut}[2][0]{\rule[- #1 mm]{0mm}{#1 mm}\rule{0mm}{#2 mm}}

\newcommand{\hcm}[1][1]{\hspace*{#1 cm}}

\newcommand{\R}{$\mathsf{R}$}
\newcommand{\Rand}{\mathsf{Rand}}
\newcommand{\Det}{\mathsf{Det}}
\newcommand{\RandLOCAL}{\mathsf{RandLOCAL}}
\newcommand{\DetLOCAL}{\mathsf{DetLOCAL}}
\newcommand{\ID}{\operatorname{ID}}
\newcommand{\LOCAL}{\mathsf{LOCAL}}

\newcommand{\diam}{\mathsf{diameter}}
\newcommand{\parent}{\operatorname{parent}}
\newcommand{\PSLOCAL}{\mathsf{PSLOCAL}}

\newcommand{\Lovasz}{Lov\'{a}sz}

\newcommand{\rake}{{\sf Rake}}
\newcommand{\compress}{{\sf Compress}}

\newcommand{\vbl}{\operatorname{vbl}}
\newcommand{\variable}{\mathcal{V}}
\newcommand{\eventA}{\tilde{F}}
\newcommand{\eventB}{F}
\newcommand{\eventC}{H}
\newcommand{\degA}[1]{{\degree_{#1}'}}
\newcommand{\degB}{\widehat{\degree}}
\newcommand{\TLLL}{T_{\operatorname{LLL}}}

\newcommand{\mindeg}{{\Delta}_{\min}}
\newcommand{\maxdeg}{{\Delta}_{\max}}
\newcommand{\oneshot}{\textsf{One-Shot-Coloring}}

%% file: intro.tex
\section{Introduction}\label{sect:intro}

In this paper, we consider the complexity of the edge coloring problem in the well-known
$\LOCAL$ model of distributed computation~\cite{Linial92,Peleg00}.
A {\em $k$-edge coloring} of a graph $G=(V,E)$
is a function $\phi : E \rightarrow \{1,\ldots,k\}$ such that edges sharing
an endpoint are colored differently; the parameter $k$ is called the \emph{palette size}.
The distributed complexity of computing a $k$-edge coloring depends heavily on the value of $k$ relative
to the maximum degree $\Delta$, 
and whether vertices can generate random bits.

\paragraph{The $\LOCAL$ Model.}
In the $\LOCAL$ model,
the input graph $G=(V,E)$ is identical to the underlying distributed network;
vertices are identified with processors and edges with bi-directional communication links;
\emph{time} is divided into synchronized rounds, and in each round each processor can perform
unlimited computation and communicate an unbounded-length message to each of its neighbors,
which is delivered before the next round.  Depending on the problem the vertices may carry
additional input labels.  The output of a $\LOCAL$ algorithm is typically a labeling of $V$ or $E$ satisfying some constraints.

For clarity, we bifurcate the $\LOCAL$ model
into $\RandLOCAL$ and $\DetLOCAL$ depending on whether random bits are
available.
In the $\RandLOCAL$ model the output labeling is correct w.h.p. (i.e., $1-1/\poly(n)$).
In the $\DetLOCAL$ model each vertex is assigned a unique $O(\log n)$-bit ID; the output labeling must always
be correct.

We assume each $v\in V$ initially knows $\deg(v)$, a port-numbering of its incident edges,
and global parameters such as $n=|V|$ and $\Delta = \max_{v \in V}  \deg(v)$, or upper bounds on them if   the exact parameters are not common knowledge. The assumption that global parameters are common knowledge can sometimes be removed; see Korman, Sereni, and Viennot~\cite{KormanSV13}.

\paragraph{Distributed Coloring.}
The two primary design objectives for distributed coloring algorithms are (i) minimizing the number of colors (palette size) and (ii) minimizing the number of rounds (time complexity). 
A modest standard for \emph{efficient algorithm} in the $\LOCAL$ model is $O(\poly \log n)$ time.  
However, there are now many examples
of locally checkable labeling problems with $O(\poly(\log\log n))$ randomized 
complexity~\cite{ChangKP19,PettieS15,GhaffariS17,GhaffariHKM18,GhaffariHK17,harris2018distributed}, 
and in some circumstances, 
$O(\log^* n)$ complexity~\cite{Linial92,SchneiderW10,ElkinPS15,ChangLP18}.

For the case of \emph{vertex coloring}, it is well-known that a $(\Delta+1)$-vertex coloring can be found in $O(\log n)$ time~\cite{Luby86,alon1986fast} in $\RandLOCAL$ or $2^{O(\sqrt{\log n})}$ time in $\DetLOCAL$~\cite{PanconesiS96}. The randomized complexity was recently improved to $O(\sqrt{\log \Delta}) + 2^{O(\sqrt{\log \log n})}$~\cite{HarrisSS16}, and then to $2^{O(\sqrt{\log \log n})}$~\cite{ChangLP18}.

These randomized upper bounds imply that a vertex coloring with palette size $\Delta+1$ can be computed efficiently. In general, the palette size of  $\Delta+1$ cannot be further reduced since there exists a graph (a complete graph with $\Delta+1$ vertices) that cannot be $\Delta$-colored.

\paragraph{Edge Coloring.} The case of  \emph{edge coloring} is more complicated.
Edge coloring can be interpreted as a \emph{vertex} coloring problem on the \emph{line graph} $L(G)$, in which
edges becomes vertices and two edges are adjacent if they share an endpoint; the line graph has maximum degree
$\hat{\Delta} = 2\Delta-2$.
Therefore, an edge coloring with palette size $\hat{\Delta}+1 = 2\Delta-1$ can be computed efficiently. The current state-of-the-art for 
$(2\Delta-1)$-edge coloring is 
$\tilde{O}(\log^3 \log n)$ time\footnote{Here $\tilde{O}(f(n)) = O(f(n)\cdot \poly(\log f(n)))$.} for all $\Delta$~\cite{harris2018distributed},
and $O(\log^* n)$ time when $\Delta > \log^{1+o(1)} n$~\cite{ElkinPS15}.
Vizing's theorem~\cite{Vizing64} guarantees the existence of a $(\Delta+1)$-edge coloring for all graphs; but it is unknown
whether such a coloring can be efficiently computed in $\LOCAL$.

The number ``$2\Delta-1$''    is the \emph{smallest} palette size with the property
that any partial edge coloring can be extended to a total coloring, by the trivial greedy algorithm.
Below the {\em greedy threshold} $2\Delta-1$, iterative coloring algorithms must be more careful in how they proceed.
In particular, at intermediate stages in the algorithm, edges must keep their available
palettes relatively large compared to the size of their uncolored neighborhood.

Using the \emph{R\"{o}dl nibble} technique, Dubhashi, Grable, and Panconesi~\cite{DubhashiGP98}
gave a $\RandLOCAL$ algorithm for $(1+\epsilon)\Delta$-edge coloring
in $O(\log n)$ time, provided that $\Delta$ is sufficiently large,
e.g., even when $\epsilon$ is constant, $\Delta > (\log n)^{1+\gamma}$.
Elkin, Pettie, and Su~\cite{ElkinPS15} gave $\RandLOCAL$ algorithms for $(1+\epsilon)\Delta$-edge coloring
that are faster when $\Delta$ is large and work for \emph{all} $\Delta$ via a reduction to the distributed
\Lovasz{} local lemma (LLL).
The $(1+\epsilon)\Delta$-edge coloring problem is solved in $O(\log^* n) \cdot \ceil{\frac{\log n}{\Delta^{1-o(1)}}}$ time.
The running time of the Dubhashi-Grable-Panconesi and Elkin-Pettie-Su algorithms depend \emph{polynomially}
on $\epsilon^{-1}$.  In both algorithms it is clear that $\epsilon$ need not be constant, but it is not self-evident
how small it can be made \emph{as a function of $\Delta$}.

The $\ceil{\frac{\log n}{\Delta^{1-o(1)}}}$-factor in the time complexity of \cite{ElkinPS15} is due to the Chung-Pettie-Su LLL
algorithm~\cite{ChungPS17}, which holds for all $\Delta$.  
The Ghaffari-Harris-Kuhn~\cite{GhaffariHK17}
and Fischer-Ghaffari~\cite{FischerG17} LLL algorithms
are faster when $\Delta = (\log n)^{o(1)}$; see Section~\ref{sect:previouswork}
and Table~\ref{table:LLL}.

\paragraph{New Results.} In this paper, we present new upper and lower bounds on the complexity of edge coloring in the regimes between palette
size $\Delta+1$ and $2\Delta-2$, i.e., strictly below the ``greedy'' threshold $2\Delta-1$.

From the lower bound side, 
we prove that  $(2\Delta-2)$-edge coloring requires $\Omega(\log_\Delta \log n)$ time w.h.p.~and
$\Omega(\log_\Delta n)$ time deterministically, \emph{even on trees}.  This result is attained via the  \emph{round elimination} technique of Brandt et al.~\cite{BrandtEtal16}.
Second, we consider a natural approach to computing $(\Delta+1)$-edge colorings (Vizing's theorem)
via extending partial colorings by iteratively re-coloring parts of the graph via ``alternating paths.''
We prove that this approach may be viable, but in the worst case requires recoloring
subgraphs of diameter $\Omega(\Delta\log n)$.
This stands in contrast to distributed algorithms for Brooks' theorem~\cite{PanconesiS95},
which exploit the existence of $O(\log_\Delta n)$-length alternating paths.

From the upper bound side,
we give an efficient randomized edge coloring algorithm that can use palette sizes as small as
$\Delta + \tilde{O}(\sqrt{\Delta})$, which is a 
natural barrier for randomized approaches.
Notice that with a palette of size $\Delta + \Theta(\sqrt{\Delta})$,
we have a constant probability
of being able to color an arbitrary edge $e$, 
given a \emph{random} feasible coloring of its neighborhood.
Edge coloring with this palette size was achieved in 1987 by Karloff and Shmoys~\cite{KarloffS87}
in the context of parallel (PRAM) algorithms, but has not been achieved in the $\LOCAL$ model before.
We also show that the $\Omega(\log_\Delta \log n)$ lower bound can be nearly matched on trees by developing a new distributed LLL algorithm for \emph{tree-structured
dependency graphs}. 

\subsection{Tools}

Randomized distributed algorithms in the $\LOCAL$ model are often composed of iterations of $O(1)$-round routines
that commit to a partial labeling~\cite{BEPS16,DubhashiGP98,ElkinPS15,PettieS15}.
A vertex may proceed to the next iteration only 
if it satisfies some property or invariant, 
which typically holds with probability $1/\poly(\Delta)$.

\paragraph{Graph Shattering.} In the  graph shattering framework~\cite{BEPS16,Beck91} of algorithm design, the \emph{bad vertices} that violate the require property are temporarily \emph{removed from consideration} in the subsequent iterations of the randomized algorithm. 
If it can be shown that at the end of the randomized algorithm, the connected components induced by the bad vertices have size at most $O(\poly(\log n))$,
one can revert to the best available {\em deterministic} algorithm and solve the problem on each component of the ``shattered'' graph in parallel. The randomized part is called the {\em pre-shattering} phase; the deterministic part is  called the {\em post-shattering} phase.

In some applications we cannot tolerate the existence of a bad vertex. 
For example, when the palette size is below the greedy threshold $2\Delta - 1$, not all partial edge coloring can be extended to a total edge coloring.
In this case, we need to resort to a distributed
\Lovasz{} local lemma (LLL) algorithm, 
which can guarantee a global success 
(i.e., there is no bad vertex) with probability
$1-1/\poly(n)$ (using a randomized LLL algorithm) or even 1 (using a deterministic LLL algorithm).\footnote{However, applying an 
LLL algorithm does not mean we have circumvented the graph 
shattering method! All known 
distributed LLL algorithms with 
a sublogarithmic dependence on $n$
(\cite{FischerG17,GhaffariHK17} and Section~\ref{sect:tree-LLL})
use graph shattering internally.  One interpretation of
Chang, Kopelowitz, and Pettie's derandomization~\cite[Theorem 3.1]{ChangKP19}
is that graph shattering is intrinsic to fast randomized algorithms
in the $\LOCAL$ model, and cannot be completely avoided.}

\paragraph{\Lovasz{} Local Lemma.} Consider a set of independent random variables $\mathcal{V}$ and a set of \emph{bad}
events $\mathcal{E}$, where each $A \in \mathcal{E}$
depends on a subset $\vbl(A)\subset \mathcal{V}$.
Define the dependency graph as $G_{\mathcal{E}} = (\mathcal{E},\, \{(A,B)~ | ~\vbl(A)\cap \vbl(B)\neq\emptyset)\})$.
Symmetric versions of the \Lovasz{} local lemma are stated in terms of
$d$, the maximum degree in $G_{\mathcal{E}}$, and $p = \max_{A\in \mathcal{E}} \Pr[A]$.
A standard version of the LLL says that if $ep(d+1) < 1$ then
$\Pr[\cap_{A \in \mathcal{E}}\overline{A}]>0$, i.e.,
it is \emph{possible} to avoid all bad events.
The constructive LLL problem is to assign values to all variables in $\mathcal{V}$ such that no event in $\mathcal{E}$ happens.

\paragraph{Distributed \Lovasz{} Local Lemma.}
In the \emph{distributed} LLL problem the communications network is identical to $G_{\mathcal{E}}$.  Every node
$A$ is identified with an event, which is aware of the distribution on the random variables $\vbl(A)\subseteq \mathcal{V}$.
The goal is to collectively assign values to all variables in $\mathcal{V}$ such that no event in $\mathcal{E}$ happens.

In distributed coloring algorithms it is typical to see $d = \poly(\Delta)$ and $p = \exp(-d^{\Omega(1)})$,
i.e., \emph{any} polynomial LLL criterion of the form $p(ed)^c < 1$ where $c=O(1)$ is good enough.
Chung, Pettie, and Su~\cite{ChungPS17} provided an $O(\log_{1/epd^2} n)$ time algorithm under the LLL criterion $epd^2 < 1$.  This remains the fastest distributed LLL algorithm under a polynomial criterion when $d$ is arbitrary.  There are faster LLL algorithms~\cite{FischerG17,GhaffariHK17} 
when $d$ is small, and slower LLL algorithms~\cite{ChungPS17,Ghaffari16} 
under the stricter
criterion $ep(d+1)<1$; see Section~\ref{sect:previouswork} and Table~\ref{table:LLL}.

\subsection{New Lower Bounds}\label{sect:newresultsLB}

\paragraph{Round Elimination.}
In Section~\ref{sect:lowerbound}, 
we show a lower bound on $(2\Delta-2)$-edge coloring
that follows the same lines as Brandt et al.'s~\cite{BrandtEtal16}
lower bound on $\Delta$-vertex coloring.  Both proofs establish
hardness for a coloring problem by reduction from \emph{sinkless orientation},\footnote{Orient the edges of the (undirected)
input graph so that no vertex is a sink.}
but one subtlety here is that we are dealing with 
\emph{two irreconcilable versions} of sinkless orientation.  
Brandt et al.~\cite{BrandtEtal16} 
prove that sinkless orientation on a graph 
\emph{that comes equipped
with a $\Delta$-edge coloring}
is reducible to $\Delta$-vertex coloring on the same graph.  
Hence, any lower bound on sinkless orientation (that is aware of the edge coloring) 
extends to $\Delta$-vertex coloring.  We show that sinkless orientation
on a bipartite graph that comes equipped with
(i) a 2-vertex coloring, and (ii) a $(2\Delta-1)$-edge coloring,
is reducible to $(2\Delta-2)$-edge coloring on the same graph.
We then prove that this version of sinkless orientation takes
$\Omega(\log_\Delta\log n)$ time in $\RandLOCAL$ and $\Omega(\log_\Delta n)$ time in $\DetLOCAL$, matching~\cite{BrandtEtal16,ChangKP19}.\footnote{
It is impossible to reconcile these two versions of sinkless orientation. 
The problem can be solved \emph{without communication}, 
given a 2-vertex coloring and a $k$-edge
coloring for any $k\in [\Delta,2\Delta-2]$.}

Roughly speaking, the idea of Brandt et al.~\cite{BrandtEtal16} is to convert any randomized $t$-round algorithm with local error probability $p$
into a $(t-1)$-round algorithm with error probability $\approx p^{1/\Delta}$.  By iterating the procedure they obtain
a 0-round algorithm with error probability $\approx p^{\Delta^{t}}$. If any 0-round algorithm must have constant probability of failure, then
$t = \Omega(\log_\Delta\log p^{-1})$.  By setting $p=1/\poly(n)$ we get $\Omega(\log_\Delta\log n)$ $\RandLOCAL$ lower bounds for some problems, e.g., sinkless orientation.

Our proof uses a simplified round elimination technique that \emph{appears} to give quantitatively worse bounds,
but which can be \emph{automatically} strengthened to match those of~\cite{BrandtEtal16}.
Rather than try to shave one round off the running time of \emph{every} processor, it is significantly simpler
to do it piecemeal, which leads us to the useful concept of an \emph{irregular} time profile.
Suppose that the graph is initially $(2\Delta-1)$-edge colored.
An algorithm has irregular time profile 
$\mathbf{t} = (t_1,\ldots,t_{2\Delta-1})$ if edges with input color
$i$ choose their output color by examining only their $t_i$-neighborhood.
In our round-elimination method, 
we show that any algorithm with time
profile $(\underbrace{t,t,\cdots,t}_{i}, \underbrace{t-1,\cdots, t-1}_{(2\Delta-1)-i})$ and error probability $p$
can be transformed into one with time profile
$(\underbrace{t,t,\cdots,t}_{i-1}, \underbrace{t-1,\cdots, t-1}_{(2\Delta-1)-i+1})$
and error probability $O(p^{1/3})$, 
\emph{only} by changing the algorithm for edges initially colored $i$.
By iterating this process we arrive at $\Omega(\Delta^{-1}\log\log p^{-1})$ lower bounds, which has a weaker
dependence on $\Delta$ than~\cite{BrandtEtal16}.
By following the proofs of Chang, Kopelowitz, and Pettie~\cite{ChangKP19}, any randomized lower bound
of this type implies $\Omega(\log_\Delta n)$ lower bounds in $\DetLOCAL$~\cite[Theorem 5]{ChangKP19}, which then implies
$\Omega(\log_\Delta \log n)$ lower bounds in 
$\RandLOCAL$~\cite[Theorem 3]{ChangKP19}.

\paragraph{Lower Bound for Distributed Vizing's Theorem.} Suppose that a distributed $(\Delta+1)$-edge coloring algorithm
begins with a partial coloring and iteratively recolors subgraphs, always increasing the subset of colored edges.
If this algorithm works correctly given \emph{any} partial coloring,
we prove In Section~\ref{sect:Vizing} that it takes $\Omega(\Delta\log n)$ time in the $\LOCAL$ model, with or without randomization.
More generally, any $(\Delta+c)$-coloring that is based on recoloring subgraphs takes $\Omega(\frac{\Delta}{c}\log n)$ time.  This establishes
a quantitative difference between the ``locality'' of Vizing's theorem and Brooks' theorem~\cite{PanconesiS95}.

\paragraph{Subsequent Work.} Subsequent to the initial publication of this work~\cite{ChangHLPU18}, Ghaffair, Kuhn, Maus, Uitto~\cite{GhaffariKMU18} showed that a $\Delta +O(\log n\cdot \log(2 + \Delta/\log n))$-edge coloring can be computed in $\DetLOCAL$ in $O(\poly (\log n, \Delta))$ rounds. Very recently, Su and Vu~\cite{SuV19} improved this bound and showed that in $O(\poly (\log n, \Delta))$ rounds, it is possible to compute a $\Delta + O(\log_\Delta n)$-edge coloring in $\DetLOCAL$ or a 
$(\Delta+2)$-edge coloring in $\RandLOCAL$, 
which is only one color away from Vizing's theorem.
All these upper bounds have time complexity of the form $O(\poly (\log n, \Delta))$. It is still an intriguing open question as to whether an edge coloring with palette size significantly smaller than 
$\Delta + \tilde{O}(\sqrt{\Delta})$ can be computed in $O(\poly \log n)$ time, regardless of $\Delta$.

\subsection{New Upper Bounds}\label{sect:newresultsUB}

\paragraph{Upper Bounds on General Graphs.} The $(1+\epsilon)\Delta$-edge coloring algorithms of~\cite{DubhashiGP98,ElkinPS15} are slow (with a polynomial dependence on $\epsilon^{-1}$)
and have limits on how small $\epsilon$ can be, as a function of $\Delta$. In Section~\ref{sect:randAlg}, we prove that the most ``natural'' randomized algorithm (\oneshot)
converges exponentially faster with $\epsilon^{-1}$ and can achieve palette sizes close to the minimum of
$\Delta + \tilde{O}(\sqrt{\Delta})$ allowed by the nibble method.
In particular, for any $\epsilon = \tilde{\Omega}(1/\sqrt{\Delta})$,
we show that $(1+\epsilon)\Delta$-edge coloring is
reducible to $O(\log\epsilon^{-1})$ instances of the \Lovasz{} local lemma with local failure
probability $p = \exp(-\epsilon^2 \Delta^{1-o(1)})$, plus one instance of $O(\Delta)$-edge coloring, which can be solved quickly using~\cite{BEPS16,ElkinPS15,GhaffariHKMSU17}.
When $\epsilon^2\Delta \gg \log n$ the  local failure
probability is already $1/\poly(n)$; otherwise we can invoke a distributed LLL algorithm~\cite{MoserT10,ChungPS17,FischerG17,GhaffariHK17}.

The running time of our algorithm varies for different choices of $\Delta$ and palette size. It can be shown that our algorithm computes a $(1+\epsilon)\Delta$-edge coloring in  at most $O(\log n)$ time when  $\epsilon \geq (\log^3 \Delta) / \sqrt{\Delta}$, or at most $O(\log_{\Delta} n) + (\log \log n)^{3 + o(1)}$ time when $\epsilon = \Omega(1)$.
These times reflect the use of Chung, Pettie, and Su's LLL algorithm~\cite{ChungPS17}.
Applying one of the Ghaffari-Harris-Kuhn LLL algorithms~\cite{GhaffariHK17}
leads to a $(1+\epsilon)\Delta$-edge coloring algorithm running
in $O(\log\epsilon^{-1}\cdot \Delta^6 + 2^{O(\sqrt{\log\log n})})$ 
time.\footnote{The running time of~\cite{GhaffariHK17} 
is (at least) quadratic in the degree $d$
of the dependency graph, and in our case $d=\Theta(\Delta^3)$.}
Our $(\Delta+\tilde{O}(\sqrt{\Delta}))$-edge coloring
algorithm is simple, but tricky to analyze, and requires a general distributed LLL algorithm to be made efficient.  
Resolving the complexity of the distributed LLL problem is a major open problem~\cite{ChangP19} 
but one that is unlikely to be completely settled any time soon,
given its connection to computing general network decompositions~\cite{FischerG17,GhaffariKM17}.

\paragraph{Upper Bounds on Trees.} 
There is still a significant gap between the upper bound of our $(1+\epsilon)\Delta$-edge coloring $\RandLOCAL$ algorithm on general graphs in Section~\ref{sect:randAlg} and our $\Omega(\log_{\Delta} \log n)$ $\RandLOCAL$ lower bound in Section~\ref{sect:lowerbound}, which applies even to trees.
We prove that this lower bound can be matched when the underlying network
is a tree, 
at least when $\epsilon=\Omega(1)$ and $\Delta < \poly(\log\log n)$.
In particular, 
our $(1+\epsilon)\Delta$-edge coloring algorithm for trees takes
$O(\log(1 / \epsilon)) \cdot 
\max\left\{\fr{\log\log n}{\log\log\log n},\, \log_{\log \Delta} \log n\right\}$ time when  $\epsilon \geq (\log^3 \Delta) / \sqrt{\Delta}$, or
$O\left( \max\left\{\fr{\log\log n}{\log\log\log n},\, \log_{\Delta} \log n\right\}\right)$ time  when $\epsilon = \Omega(1)$.

This improvement is achieved by developing a new distributed LLL algorithm
 for \emph{tree structured} dependency graphs, which appears in Section~\ref{sect:tree-LLL}.
 Specifically, if $T=(V,E)$ is a tree and $r=O(1)$, we say that $T^r = (V,\{(u,v) \;|\; \dist_T(u,v)\le r\})$ is \emph{tree-structured}.
 This type of dependency graph arises naturally when we 
 run $O(1)$-round probabilistic algorithms on trees.
 
 Our new LLL algorithm is based on the graph shattering framework. 
 We first apply a randomized algorithm that fixes the output of most of the vertices such that each connected component of the remaining part of the graph is small.  We then apply a new deterministic LLL algorithm for tree-structured instances to each component in parallel.
 
Fischer and Ghaffari~\cite{FischerG17} showed that one can obtain a $\DetLOCAL$ LLL algorithm using a \emph{network decomposition} algorithm as a black box.
Based on this idea,
we give a deterministic $O(\max\{\log_\lambda n, \log n/\log\log n\})$-time LLL algorithm for tree-structured instances under criterion $p(ed)^\lambda < 1$, $\lambda \ge 2$.
The algorithm is based on two new network decomposition algorithms
for tree-structured graphs, presented in Section~\ref{sect:tree-decomp}.

For the randomized part of the   graph shattering routine,
the goal is to design an algorithm to compute a good partial
assignment $\phi$ such that the connected components induced by the unassigned part of the dependency graph are small.
We give an algorithm for tree-structured instances that achieves this goal in
time $O(\log_\lambda \log n)$, improving the $O(d^2 + \log^* n)$-time shattering routine of~\cite{FischerG17} when $d$ is not too small.
At a high level, our approach is to consider the following process.
First, draw a total assignment $\phi$ to $\mathcal{V}$ according to the distribution of the variables.
Whenever the probability that a bad event $E(v)$ occurs under the current partial assignment $\phi$ is higher than a certain threshold, 
update $\phi$ by \emph{unsetting} all variables in $\vbl(E(v))$.
This can be viewed as a \emph{contagion dynamic} played out on the dependency graph.  Vertices  that have unset their variables are said to have been {\em infected}, and infected vertices can cause nearby neighbors to become infected.
If this contagion process were actually simulated, it would take $\Omega(\log n)$ parallel steps to reach a stable state, which is too slow.
We develop a different method to achieve a stable state 
that is exponentially faster, by avoiding a direct simulation.

By composing these results we obtain a randomized $O(\max\{\log_\lambda \log n,\, \log\log n/\log\log\log n\})$ LLL algorithm for tree-structured instances, when $\lambda$ is at least a sufficiently large constant depending on $r$.
Our upper bound essentially matches the  $\RandLOCAL$ lower bound of Brandt et al.~\cite{BrandtEtal16}, which is of the form $\Omega(\log_{\log p^{-1}} \log n)$ under the LLL criterion $p \cdot f(d) \leq 1$ for any $f(d) \leq 2^d$.

\medskip
A major open problem is to extend this \emph{contagion dynamic} idea
to general dependency graphs, and show that they, too, can be shattered
in $O(\log\log n)$ time.  In light of~\cite{ChangKP19,BrandtEtal16},
this is a necessary first step towards proving
Conjecture 1 from Chang and Pettie~\cite{ChangP19}, 
namely that the $\RandLOCAL$ complexity
of the LLL under a polynomial criterion is $O(\log\log n)$.

\paragraph{Additional Results on Trees.}
In Section~\ref{sect:upperbound} we prove some additional results on the complexity of edge coloring trees. 
We design an $O(\log_\Delta n)$-time $\DetLOCAL$ algorithm for $\Delta$-edge coloring a tree $T$ with maximum degree $\Delta\ge 3$.
A tree is said to be   \emph{oriented} if the tree is rooted and each vertex that is not the root knows its parent. 
We show that a $(\Delta+1)$-edge coloring of an oriented tree can be found in  $O(\log^* n)$ time, 
but $\Delta$-edge coloring  takes $\Omega(\log_\Delta n)$ time.

\paragraph{Remark.}
After the initial publication of this work in~\cite{ChangHLPU18}, we learned that Molloy and Reed~\cite{MolloyR00} also obtained a similar bound 
of $\Delta + O(\sqrt{\Delta} \log^4 \Delta)$ on the palette size for edge coloring.  Their algorithm was more general in that it extends to $k$-uniform hypergraphs (with palette size $\Delta + O(\Delta^{1-1/k} \log^4 \Delta)$) and applies to \emph{list} edge coloring. The main difference between our work and theirs~\cite{MolloyR00} is the analysis. 
We use a concentration bound~\cite[Equation (8.5)]{DubhashiPanconesi09} that takes into account the variance of each variable. The analysis of~\cite{MolloyR00} is based on Talagrand's concentration inequality. Our result is slightly better in terms of the polylog-factor, and it also improves the existential bound on the palette size for list edge coloring. Specifically, if each edge is given a list of $(1+\epsilon)\Delta$ with $\epsilon = \omega((\log^{2.5}\Delta) / \sqrt{\Delta})$ colors, then the graph admits a proper list edge coloring.

\subsection{Related Work}\label{sect:previouswork}
In this section, we walk though the rich history of 
distributed edge coloring and the distributed LLL.

We begin with reviewing
previous edge coloring algorithms;
see Table~\ref{table:edgecoloring} for a summary.
Edge coloring can be interpreted as a \emph{vertex} coloring problem on the \emph{line graph} $L(G)$, which has has maximum degree
$\hat{\Delta} = 2\Delta-2$.  Applied to $L(G)$, Linial's~\cite{Linial92} vertex coloring algorithm will compute an
$O(\hat{\Delta}^2)$-edge coloring in $O(\log^* n - \log^*\hat{\Delta} + 1)$ time.  Using the fastest deterministic
$(\hat{\Delta}+1)$-vertex coloring algorithms~\cite{PanconesiS96,FraigniaudHK16}, $(2\Delta-1)$-edge coloring
is solved in $\min\{2^{O(\sqrt{\log n})},\, \tilde{O}(\sqrt{\Delta}) + O(\log^* n)\}$ time.
Barenboim, Elkin, and Maimon~\cite{BarenboimEM} gave deterministic algorithms for $(2^{k}\Delta)$-edge coloring ($k\ge 2$)
in $\tilde{O}(k\Delta^{1/2k}) + O(\log^\ast n)$ time.

Barenboim, Elkin, Pettie, and Schneider~\cite{BEPS16} proved that $O(\log \Delta)$ iterations of
the natural randomized $(2\Delta-1)$-edge coloring algorithm
effectively \emph{shatters} the graph into uncolored components of $n' = \poly(\log n)$ vertices;
then we can employ a deterministic list coloring algorithm to color these components in
$2^{O(\sqrt{\log n'})} = 2^{O(\sqrt{\log \log n})}$ time~\cite{PanconesiS96}. Thus, the total time complexity is $O(\log \Delta) + 2^{O(\sqrt{\log \log n})}$.

Elkin, Pettie, and Su~\cite{ElkinPS15} proved that when $\Delta > (\log n)^{1+\gamma}$ (for some constant $\gamma$),
$(2\Delta-1)$-edge coloring can be solved in $O(\log^* n)$ time in $\RandLOCAL$.
Recently, Fischer, Ghaffari, and Kuhn~\cite{FischerGK17} proved that $(2\Delta-1)$-edge coloring can be solved in
$O(\log^7 \Delta \log n)$ time in $\DetLOCAL$. 
This bound was later improved to $O(\log^4 \Delta \log^2 n)$ by Ghaffari, Harris, and Kuhn~\cite{GhaffariHK17}, and then to $\tilde{O}(\log^2 \Delta \log n)$ by Harris~\cite{harris2018distributed}.
Together with~\cite{BEPS16} and~\cite{ElkinPS15},
this implies a $\RandLOCAL$ algorithm taking $\tilde{O}(\log^3 \log n)$
time.
Using a slightly larger palette of $(2+\epsilon)\Delta$ colors, $\epsilon > 1/\log\Delta$,
Ghaffari et al.~\cite{GhaffariHKMSU17} 
gave an $O(\epsilon^{-1}\log^2\Delta\log\log\Delta(\log\log\log\Delta)^{1.71}\log n)$-time $\DetLOCAL$ edge coloring algorithm, improving a previous work~\cite{GhaffariS17}.

We cannot hope to use fewer than $\Delta+1$ colors on general graphs.  Vizing~\cite{Vizing64} proved that $\Delta+1$
suffices for any graph, and Holyer~\cite{Holyer81} proved that it is NP-hard to tell if a graph is $\Delta$-colorable.
The best \emph{sequential} $(\Delta+1)$-edge coloring algorithms~\cite{Arjomandi82,GabowNKLT85} run in
$O(\min\{\Delta m\log n, \, m\sqrt{n\log n}\})$ time and are not suited for implementation in the $\LOCAL$ model.
When the palette size is small
a natural way to solve the coloring problem~\cite{Arjomandi82,GabowNKLT85}
is to begin with any maximal
partial coloring, and then iteratively \emph{recolor} portions of the graph (e.g., along ``alternating paths'')
so that at least one uncolored edge can be legally colored.  This approach was successfully employed
by Panconesi and Srinivasan~\cite{PanconesiS95} in their distributed algorithm for Brooks' theorem, which states that any graph with $\Delta \ge 3$ having no $(\Delta+1)$-cliques is $\Delta$-vertex colorable.
They proved that for \emph{any} partial coloring, 
there exists an alternating path with length $O(\log_\Delta n)$, and that given a $(\Delta+1)$-vertex coloring, a
$\Delta$-vertex coloring could be computed in $O(\log^2 n\log_\Delta n)$ additional time.
This bound was recently improved by Ghaffari et al.~\cite{GhaffariHKM18}, which offers some
improved $\Delta$-vertex coloring  algorithms.

\begin{table}
\small
\centering
\begin{tabular}{|l|l|l|}
\multicolumn{1}{l}{\textbf{\textsc{Palette Size}}}
& \multicolumn{1}{l}{\textbf{\textsc{Time}}\hfill (\underline{$\mathsf{R}$}$\mathsf{and}$)}
& \multicolumn{1}{l}{\textbf{\textsc{Notes \hfill References}}}\\\hline
$f(\Delta)$				&$\Omega(\log^* n)$ \hcm[2]					\hfill \R	& $\Delta = O(1)$ \hcm[2.5]\hfill \cite{Linial92,Naor91}\istrut[2]{4}\\\hline
$O(\Delta^2)$			& $O(\log^* n - \log^* \Delta + 1)$					\hfill $\star$ & Vertex coloring $L(G)$ \hfill \cite{Linial92}\istrut[2]{4}\\\hline
$\Delta^{1+\epsilon}$	& $O(\log\Delta + \log^\ast n)$		\hfill $\star$					& \hfill \cite{BarenboimE2013}\istrut[2]{4}\\\hline
$O(\Delta\log n)$			& $O(\log^4 n)$										& \hfill \cite{CzygrinowHK01}\istrut[2]{4}\\\hline
$t  (2\Delta-2)$		& $(\Delta / t)^{O(1)} \cdot  O(\log n)$								& Vertex coloring $L(G)$ \hfill \cite{BarenboimE11}\istrut[2]{4}\\\hline
$2^k\Delta$			& $\tilde{O}(k\Delta^{1/2k}) + O(\log^* n)$					\hfill $\star$ & $k\ge 2$ \hfill \cite{BarenboimEM}\istrut[2]{4}\\\hline
					& $O(\epsilon^{-3}\log^{11} n)$								& \hfill\cite{GhaffariS17}\istrut[2]{4}\\\cline{2-3}
\rb{3}{$(2+\epsilon)\Delta$}	& $O(\epsilon^{-1}\log\Delta^{2+o(1)}\log n)$			\hfill & $\epsilon > 1/\log \Delta$\hfill\cite{GhaffariHKMSU17}\istrut[2]{4}\\\hline
					& \rb{-0.5}{$2^{O(\sqrt{\log n})}$}								& Vertex coloring $L(G)$ \hfill \cite{PanconesiS96}\istrut[2]{4}\\\cline{2-3}
					& \rb{-0.5}{$\tilde{O}(\sqrt{\Delta}) + O(\log^* n)$}				\hfill $\star$ & Vertex coloring $L(G)$ \hfill \cite{FraigniaudHK16}\istrut[2]{4}\\\cline{2-3}
					& \rb{-0.5}{$O(\log\Delta) + 2^{O(\sqrt{\log\log n})}$} \hfill \R					& Vertex coloring $L(G)$ \hfill \cite{BEPS16}\istrut[2]{4}\\\cline{2-3}
$2\Delta-1$			& $O(\log^* n)$								\hfill \R$\star$   & $\Delta > (\log n)^{1+o(1)}$ \hfill \cite{ElkinPS15}\istrut[2]{4}\\ \cline{2-3}
					& \rb{-0.5}{$2^{O(\sqrt{\log\log n})}$}					\hfill \R	& \hfill \cite{ElkinPS15}\istrut[2]{4}\\\cline{2-3}
					& $O(\log^7 \Delta\log n)$						\hfill  & \hfill \cite{FischerGK17}\istrut[2]{4}\\\cline{2-3}
					& $O(\log^4 \Delta\log^2 n)$						\hfill & \hfill \cite{GhaffariHK17}\istrut[2]{4}\\\cline{2-3}
					& $\tilde{O}(\log^2 \Delta\log n)$						\hfill $\star$ & \hfill \cite{harris2018distributed}\istrut[2]{4}\\\cline{2-3}
					& $\tilde{O}((\log\log n)^3)$		\hfill  \R$\star$	& \hfill \cite{BEPS16}+\cite{ElkinPS15}+\cite{harris2018distributed}\istrut[2]{4}\\\hline
					&  $\Omega(\log_\Delta\log n)$		\hfill \R	& \hfill {\bf new}\istrut[2]{4}\\\cline{2-3}
\rb{2.5}{$2\Delta-2$}		&   $\Omega(\log_\Delta n)$					& \hfill {\bf new}\istrut[2]{4}\\\hline
$1.6\Delta$			& $O(\log n)$	\hfill \R							& $\Delta > \log^{1+o(1)} n$\hfill \cite{PanconesiS97}
\\
\hline
& $O(\epsilon^{-1}\log\epsilon^{-1} + \log n)$											\hfill \R	
& $\Delta > (\log n)^{1+\gamma(\epsilon)}$ \hfill \cite{DubhashiGP98}\istrut[2]{4}\\\cline{2-3}
\rb{-4}{$(1+\epsilon)\Delta$}	
& $O\left((\epsilon^{-2}\log\epsilon^{-1} + \log^*\Delta)\ceil{\frac{\log n}{\epsilon^2\Delta^{1-o(1)}}}\right)$			
\hfill \R	
& $\Delta > \Delta_\epsilon$ 
\hfill \cite{ElkinPS15}\istrut[2]{5}
\\
\cline{2-3}
					& $O\left(\log\epsilon^{-1}\ceil{\frac{\log n}{\epsilon^2\Delta^{1-o(1)}}} + \log^* n\right)$					\hfill \R$\star$	& $\epsilon\Delta > (\log n)^{1+o(1)}$ \hfill {\bf new}\istrut[2]{5}\\\cline{2-3}
					& $O\left(\log\epsilon^{-1}\ceil{\frac{\log n}{\epsilon^2\Delta^{1-o(1)}}} + (\log\log n)^{3+o(1)}\right)$ 		\hfill \R$\star$ 	& $\epsilon = \omega((\log^{2.5}\Delta)/\sqrt{\Delta})$ \hfill {\bf new}\istrut[2]{5}\\ \cline{2-3}
\hline
$\Delta + O(\log_{\Delta} n)$	& $O(\Delta^{6+\epsilon}\log^3 n)$		\hfill $\star$  &  
\hfill \cite{SuV19}\istrut[2]{5}\\\hline
$\Delta + 2$	& $O(\Delta^{13}\log^3 n)$		\hfill \R$\star$ &  
\hfill \cite{SuV19}\istrut[2]{5}\\\hline
$\Delta+1$			& $\diam(G)$																\hfill $\star$ & \hfill \cite{Vizing64}\istrut[2]{4}\\\hline
\end{tabular}
\caption{\label{table:edgecoloring}A history of notable edge coloring algorithms and lower bounds, in descending
order by palette size.
Some $(2\Delta-1)$-edge coloring algorithms that follow from vertex coloring~$L(G)$, such as
\cite{AwerbuchGLP89,KuhnW06,BarenboimEK14,Barenboim15}, have been omitted for brevity.
$\RandLOCAL$ algorithms are marked with \R; all others work in $\DetLOCAL$.
Those algorithms that are the ``best'' in any sense are marked with a $\star$.
}
\end{table}

\newcommand{\MIS}{\mathsf{MIS}}
\newcommand{\WeakMIS}{\mathsf{WeakMIS}}

\begin{table}
\small
\centering
\begin{tabular}{|r|l|l||}
\multicolumn{1}{l}{\textbf{\textsc{Criterion}}} &
\multicolumn{1}{l}{\textbf{\textsc{Time}} \hfill $\Rand/\Det$} &
\multicolumn{1}{l}{\textbf{\textsc{Notes}} \hfill \textbf{\textsc{Reference}}}\\\hline
				& 	$O(\MIS \cdot \log_{1/ep(d+1)} n)$								\hfill$\Rand$ & also asymmetric criterion\hfill \cite{MoserT10}\istrut[2]{4}\\\cline{2-3}
$ep(d+1) < 1$		&	$O(\WeakMIS \cdot \log_{1/ep(d+1)} n)$							\hfill$\Rand$ & also asymmetric criterion\hfill \cite{ChungPS17}\istrut[2]{4}\\\cline{2-3}
				&	$O(\log d\cdot \log_{1/ep(d+1)} n)$								\hfill$\Rand$ & also asymmetric criterion \; \hfill \cite{Ghaffari16}+\cite{ChungPS17}\istrut[2]{4}\\\hline
$epd^2 < 1$		& 	$O(\log_{1/epd^2} n)$										\hfill$\Rand$ & also asymmetric criterion\hfill \cite{ChungPS17}\istrut[2]{4}\\\hline
 $p 2^d \poly(d) < 1$	&	$O(\log n/\log\log n)$											\hfill$\Rand$ & \hfill \cite{ChungPS17}\istrut[2]{4}\\\hline
$p(ed)^\lambda < 1$	&	$O(n^{1/\lambda}\cdot 2^{O(\sqrt{\log n})})$						\hfill$\Det$	& Any $\lambda \ge 1$\hfill \cite{FischerG17}\istrut[2]{4.5}\\\hline
$p(ed)^{4\lambda} < 1$ &	$O(d^2) + (\log n)^{1/\lambda}\cdot 2^{O(\sqrt{\log\log n})}$ \hcm[.3]		\hfill$\Rand$	& Any $\lambda \ge 8$\hfill \cite{FischerG17}\istrut[2]{4.5}\\\hline
$p(ed)^{32} < 1$	& 	$2^{O(\sqrt{\log\log n})}$										\hfill$\Rand$ &  $d \leq (\log\log n)^{1/5}$\hcm[.4]\hfill \cite{FischerG17}\istrut[2]{4.5}\\\hline
$20000p d^8 < 1$	& $\exp^{(i)} \left( O\left(\log d + \sqrt{\log^{(i+1)} n}\right) \right)$				\hfill$\Rand$ &  $i\ge 1$.
\hfill \cite{GhaffariHK17}\istrut[2]{4.5}\\\hline
$p(ed)^{d^2+1} < 1$	&	$O(d^2 + \log^* n)$											\hfill$\Det$	& \hfill \cite{FischerG17}\istrut[2]{4}\\\hline\hline
\multicolumn{3}{c}{}\\
\multicolumn{3}{c}{\bf Lower Bounds (apply to tree-structured instances)}\\\hline
$p\cdot f(d) \leq 1$	 & 	$\Omega(\log^* n)$											\hfill$\Rand$ & Any $f$ \hfill \cite{ChungPS17}\istrut[2]{4}\\\hline
$p\cdot f(d) \le 1$  &	$\Omega(\log_{\log(1/p)} \log n)$								\hfill$\Rand$ & Any $f(d) \le 2^d$ \hfill \cite{BrandtEtal16}\istrut[2]{4}\\\hline
$p\cdot f(d) \le 1$  &	$\Omega(\log_{d} n)$										\hfill$\Det$  & Any $f(d) \le 2^d$ \hfill \cite{ChangKP19}\istrut[2]{4}\\\hline\hline
\multicolumn{3}{c}{}\\
\multicolumn{3}{c}{\bf LLL for Tree-Structured Instances}\\\hline
$p(ed)^2 < 1$			&	$O(\log n)$											\hfill$\Det$	& \hfill {\bf new}\\\hline
$p(ed)^\lambda < 1$		&	$O(\max\{\log_\lambda n, \, \f{\log n}{\log\log n}\})$					\hfill$\Det$	& $\lambda\ge 2$ \hfill {\bf new}\\\hline
$p(ed)^\lambda < 1$		&	$O(\max\{\log_\lambda \log n, \, \f{\log\log n}{\log\log\log n}\})$			\hfill$\Rand$	& $\lambda\ge 2(4^{r}+ 8r)$ \hfill {\bf new}\\\hline\hline
\end{tabular}
\caption{\label{table:LLL}A survey of distributed LLL algorithms (with a symmetric LLL criterion).
$\MIS = O(\min\{d + \log^* n,\, \log d + 2^{O(\sqrt{\log\log n})}\})$~\cite{BarenboimEK14,Ghaffari16}
is the complexity of computing a maximal independent set in a graph with maximum degree $d$.
$\WeakMIS = O(\log d)$~\cite{Ghaffari16} is the task of finding an independent set $I$ such that the probability
that $v$ is not in/adjacent to $I$ is $1/\poly(d)$.
If $T=(V,E)$ is a tree, 
$T^r = (V,\{(u,v) \;:\; \dist_T(u,v) \le r\})$ is \emph{tree-structured}, where $r=O(1)$.
All lower bounds apply even to tree-structured instances.
We do not try to optimize the LLL criterion 
$\lambda \geq 2(4^{r}+ 8r)$ in the last line.}
\end{table}

\paragraph{Lower Bounds}
Linial's~$\Omega(\log^* n)$ lower bound for $O(1)$-coloring the ring~\cite{Linial92,Naor91}
implies that $f(\Delta)$-edge coloring also cannot be computed in $o(\log^* n)$ time, for any function $f$.
To the best of our knowledge, none of the other published lower bounds applies directly to the edge coloring problem.
Kuhn, Moscibroda, and Wattenhofer's~$\Omega\left(\min\left\{\frac{\log\Delta}{\log\log\Delta}, \, \sqrt{\frac{\log n}{\log\log n}}\right\}\right)$
lower bounds apply to MIS and maximal matching, but not to any vertex or edge coloring problem.
Linial's $\Omega(\log_\Delta n)$ lower bound~\cite{Linial92} (see~\cite[p. 265]{PettieS15})
on $o(\Delta/\ln\Delta)$-vertex coloring trees does not imply anything for edge coloring trees.
The lower bounds of Brandt et al.~\cite{BrandtEtal16} ($\RandLOCAL$  $\Omega(\log_\Delta \log n)$)
and Chang, Kopelowitz, and Pettie~\cite{ChangKP19} ($\DetLOCAL$  $\Omega(\log_\Delta n)$)
for \emph{sinkless orientation} and \emph{$\Delta$-vertex coloring trees} do not naturally generalize to edge coloring.

\paragraph{Distributed \Lovasz{} Local Lemma.}
Table~\ref{table:LLL} summarizes distributed LLL algorithms under different symmetric criteria $p\cdot f(d) < 1$,
where $p$ is the local probability of failure and $d$ is the maximum degree in the dependency graph.
Chang and Pettie~\cite{ChangP19} conjectured that the $\RandLOCAL$ complexity of the LLL under
some polynomial LLL criterion is $O(\log\log n)$, matching the Brandt et al.~\cite{BrandtEtal16} lower bound.
If this conjecture were true, due to the necessity of graph shattering \cite[Theorem 3]{ChangKP19},
an optimal randomized LLL algorithm should be structured as follows.
It must combine an $O(\log n)$-time deterministic LLL algorithm and an
$O(\log\log n)$-time randomized \emph{graph shattering} routine to break the dependency graph into $\poly(\log n)$-size LLL instances.
Fischer and Ghaffari~\cite{FischerG17} exhibited a deterministic $n^{1/\lambda + o(1)}$-time algorithm for LLL criterion
$p(ed)^\lambda < 1$, and an $O(d^2 + \log^* n)$ routine to shatter the dependency graph into $\poly(\log n)$-size components.
More recently, Ghaffari, Harris, and Kuhn~\cite{GhaffariHK17} developed a 
generic derandomization method for the $\LOCAL$
model that implies randomized LLL algorithms running in time
$\exp^{(i)} \left( O\left(\log d + \sqrt{\log^{(i+1)} n}\right) \right)$, for any $i\ge 1$.
For example, when $d < 2^{O(\sqrt{\log\log n})}$, their LLL algorithm runs in 
$2^{O(\sqrt{\log\log n})}$ time.

\subsection{Organization}
In Section~\ref{sect:lowerbound} we give lower bounds on $(2\Delta-2)$-edge coloring.
In Section~\ref{sect:Vizing} we give lower bounds on a class of ``recoloring'' algorithms for Vizing's theorem.
In Section~\ref{sect:randAlg} we give a randomized $(1+\epsilon)\Delta$-edge coloring algorithm,
which requires a distributed LLL algorithm when $\epsilon^2\Delta$ is sufficiently small.
In Section~\ref{sect:tree-LLL} we give new LLL algorithms for tree-structured dependency graphs.
In Section~\ref{sect:tree-decomp} we present new network decomposition algorithms for trees, which are used in Section~\ref{sect:tree-LLL}.
In Section~\ref{sect:upperbound} we prove some new bounds on the complexity 
of $\Delta$- and $(\Delta+1)$-edge coloring trees, both in the oriented and unoriented cases.
Much of the analysis of the randomized edge-coloring algorithm (Section~\ref{sect:randAlg}) 
appears in Appendix~\ref{sect:concentration}.

%% file: lowerBound.tex
\newcommand{\InfTree}{\mathcal{T}_\Delta}
\newcommand{\orientation}{\operatorname{Orient}}
\newcommand{\bottom}{\perp}

\newcommand{\BadZero}{\mathcal{E}_0}
\newcommand{\BadOne}{\mathcal{E}_1}
\newcommand{\DangerZero}{\mathcal{E}_0^\star}
\newcommand{\DangerOne}{\mathcal{E}_1^\star}

\section{Lower Bound for $(2\Delta-2)$-Edge Coloring}\label{sect:lowerbound}

The \emph{sinkless orientation} problem~\cite{BrandtEtal16} is to 
direct the edges such that no vertex has out-degree zero.
Since this problem becomes \emph{harder} with fewer edges, in this 
section we write $\mindeg$  and  $\maxdeg$ 
to denote the \emph{minimum} and the \emph{maximum}  degree.
We follow the method of~\cite{BrandtEtal16} and~\cite{ChangKP19},
who proved that $\Delta$-coloring graphs (even trees) requires
$\Omega(\log_\Delta \log n)$ in $\RandLOCAL$ and $\Omega(\log_\Delta n)$ in $\DetLOCAL$.
Brandt et al.~\cite{BrandtEtal16} begin by reducing the sinkless orientation
problem, \emph{in which nodes initially know a $\Delta$-edge coloring of the graph},
to $\Delta$-vertex coloring.  Having the $\Delta$-edge coloring available is essential
for making the reduction work, and intuitively it
leaks no information helpful for solving
either problem.   In Theorem~\ref{thm:so-ec-reduction} we begin with a similar reduction,
showing that sinkless orientation on bipartite graphs, 
\emph{in which nodes initially know a proper $2$-vertex coloring}, 
is reducible to $(2\Delta-2)$-edge coloring.  We then proceed to prove
lower bounds on sinkless orientation, given the aforementioned $2$-vertex coloring,
and even given a proper $(2\Delta-1)$-edge coloring.  (By Theorem~\ref{thm:so-ec-reduction},
reducing the edge-coloring palette to $2\Delta-2$ would trivialize the sinkless orientation problem.)

\begin{theorem}\label{thm:so-ec-reduction}
Suppose $\mathcal{A}_{e.c.}$ is a $t$-round $(2\Delta-2)$-edge coloring algorithm with local failure probability $p$ on graphs with maximum degree $\maxdeg \leq \Delta$.
There is a $(t+1)$-round sinkless orientation algorithm $\mathcal{A}_{s.o.}$ for
2-vertex colored bipartite graphs with minimum degree $\mindeg \geq \Delta$ whose local failure probability is $p$.
\end{theorem}

\begin{proof}
$\mathcal{A}_{e.c.}$ produces a proper \emph{partial} $(2\Delta-2)$-edge coloring
$\phi : E \rightarrow \{1,\ldots,2\Delta-2,\bottom\}$ such that for all $v\in V$,
$\Pr[\exists (u,v) : \phi(u,v) =\: \bottom] \le p$, i.e., a vertex errs if not all of its edges are colored.
Suppose we are given a bipartite graph $G=(V,E)$ with a 2-coloring $V\rightarrow \{0,1\}$ and minimum degree $\mindeg \geq \Delta$.
In the first round of $\mathcal{A}_{s.o.}$, each vertex selects $\Delta$ of its incident edges arbitrarily and notifies the other endpoint whether it was selected.
Let $G' = (V,E')$ be the subgraph of edges selected by \emph{both} endpoints.
The algorithm $\mathcal{A}_{s.o.}$ runs $\mathcal{A}_{e.c.}$ on $G'$ for $t$ rounds
to get a partial coloring $\phi : E'\rightarrow \{1,\ldots,2\Delta-2,\bottom\}$,
and then it orients the edges  as follows.
Recall that the underlying graph $G$ is 2-vertex colored.
Let $e = \{u_0,u_1\}\in E$ be an edge with $u_j$ colored $j\in\{0,1\}$.
If both $u_0$ and $u_1$ do not select $e$, then $e$ is oriented arbitrarily.
Otherwise, $\mathcal{A}_{s.o.}$ orients $e$ as follows.
\[
\mathcal{A}_{s.o.}(\{u_0,u_1\}) = \left\{
\begin{array}{l@{\hcm[.5]}l}
\rb{-3}{$0 \to 1$} & \mbox{ if $\{u_0,u_1\} \in E'$ and $\phi(u_0,u_1) \in \{1,2,\ldots,\Delta-1,\bottom\}$,}\\
		& \mbox{\ \ \  or if only $u_0$ selected $\{u_0,u_1\}$.}\\
		&\\
\rb{-3}{$0 \gets 1$} & \mbox{ if $\{u_0,u_1\} \in E'$ and $\phi(u_0,u_1) \in \{\Delta,\ldots,2\Delta-2\}$,}\\
		& \mbox{\ \ \  or if only $u_1$ selected $\{u_0,u_1\}$.}
\end{array}\right.
\]
The only way a vertex $v$ can be a sink is when
(i) $v$ has degree exactly $\Delta$ in $G'$,
(ii) $v$ is colored $1$, and
(iii) each edge $e$ incident to $v$ has $\phi(e) \in \{1,2,\ldots,\Delta-1,\bottom\}$.
Criterion (iii) only occurs with probability at most $p$.
\end{proof}

Thus, any lower bound for sinkless orientation on 2-vertex colored bipartite graphs also applies to $(2\Delta-2)$-edge coloring.

\paragraph{Infinite $\Delta$-regular Tree  $\InfTree$.} Define $\InfTree$ to be an infinite $\Delta$-regular tree whose vertices are properly 2-colored by $\{0,1\}$ and whose
edges are assigned a proper $(2\Delta-1)$-edge coloring as follows.
Pick an edge and assign it a random color, then iteratively pick any vertex $u$ with one incident edge colored,
choose $\Delta-1$ colors at random from the ${2\Delta-2 \choose \Delta -1}$ possibilities, then assign them to
$u$'s remaining uncolored edges uniformly at random.


\paragraph{Information Stored in the Processors.}
For simplicity we suppose that the \emph{edges} host processors, 
and that two edges can communicate if they are adjacent in the line graph $L(\InfTree)$.
Define $N^t(e)$ to be all edges within distance $t$ of $e$ in the line graph;
we also use $N^t(e)$ to refer to \emph{all information} stored in the processors within $N^t(e)$; this includes 
edge coloring, vertex coloring, and the random bits.

Randomized algorithms that run on $\InfTree$ know the edge coloring
and how it was generated.  Thus, the probability of \emph{failure} depends on the random bits generated by the algorithm,
and those used to generate the edge coloring.

\paragraph{Irregular Time Profile.}
We say that an algorithm on a $k$-edge colored graph $G$  has irregular time profile $\mathbf{t} = (t_1,\ldots,t_k)$ if edges with input color
$i$ decide their output by examining only their $t_i$-neighborhood.
By definition, a time-$t$ algorithm has time profile $(t,t,t,\ldots,t)$. In the subsequent discussion, we will apply this concept to $\InfTree$. Recall that the edges in $\InfTree$ are properly $(2\Delta - 1)$-edge colored, and so an irregular time profile for an algorithm on $\InfTree$ is a $(2\Delta - 1)$-tuple.

\begin{lemma}[Round Elimination Lemma]\label{lem:roundelimination}
Suppose $\mathcal{A}_{s.o.}$ is a sinkless orientation algorithm for $\InfTree$ with local error probability $p$
and time profile 
$(\underbrace{t,t, \ldots,t}_{i},
\underbrace{t-1, \ldots, t-1}_{(2\Delta-1)-i})$, i.e., edges colored $\{1,\ldots,i\}$ halt after $t$
rounds and the others after $t-1$ rounds.  There exists a sinkless orientation algorithm $\mathcal{A}'_{s.o.}$  for $\InfTree$
with local error probability $3p^{1/3}$ and time profile $(\underbrace{t,t, \ldots,t}_{i-1}, \underbrace{t-1, \ldots, t-1}_{(2\Delta-1)-(i-1)})$.
\end{lemma}

\begin{proof}
Only edges colored $i$ modify their algorithm; all others behave identically under $\mathcal{A}'_{s.o.}$ and $\mathcal{A}_{s.o.}$.
Let $e_0=\{u_0,u_1\}$ be an edge colored $i$ with $u_j$ colored $j\in\{0,1\}$
and let the remaining edges incident to $u_0$ and $u_1$ be $\{e_1,\ldots,e_{\Delta-1}\}$ and
$\{e_\Delta,\ldots,e_{2\Delta-2}\}$, respectively.
Consider the following two events regarding the output of $\mathcal{A}_{s.o.}$.
\begin{align*}
\BadZero 	&: \forall j\in [1,\Delta-1], \mathcal{A}_{s.o.}(e_j) = 0\gets 1	& \mbox{I.e., $u_0$ has outdegree 0 in $G-\{e_0\}$}\\
\BadOne 	&: \forall j\in [\Delta,2\Delta-2], \mathcal{A}_{s.o.}(e_j) = 0\to 1	& \mbox{I.e., $u_1$ has outdegree 0 in $G-\{e_0\}$}
\end{align*}
If both events hold, then either $u_0$ or $u_1$ must be a sink, so
\begin{equation}\label{eqn:BadZeroBadOne}
\Pr[\BadZero \cap \BadOne] \le 2p.
\end{equation}
On edge $e_0$, $\mathcal{A}_{s.o.}'$ 
gathers its $(t-1)$-neighborhood $N^{t-1}(e_0)$
and evaluates whether the following two events
$\DangerZero,\DangerOne$ occur.
Intuitively, $\DangerZero$ indicates that
$\BadZero$ is \emph{dangerously likely} to happen,
conditioned on $N^{t-1}(e_0)$, and likewise
with $\DangerOne$ and $\BadZero$.  See Figure~\ref{fig:dangerous-independent}.
\[
\DangerZero	: \sqbrack{\, \Pr[\BadZero \:|\: N^{t-1}(e_0)] \ge p^{1/3}\, },
\hcm[1.5]
\DangerOne	: \sqbrack{\, \Pr[\BadOne \:|\: N^{t-1}(e_0)] \ge p^{1/3}\, }.
\]%
Notice that if we inspect $N^{t-1}(e_0)$, and condition on the information seen in $N^{t-1}(e_0)$,
the events $\BadZero$ and $\BadOne$ become independent,
since they now depend on disjoint sets of random variables. Specifically, $\BadZero$ depends
on $\bigcup_{j\in[1,\Delta-1]} N^{t}(e_j) \backslash N^{t-1}(e_0)$ and $\BadOne$ depends on
$\bigcup_{j\in[\Delta,2\Delta-2]} N^{t}(e_j) \backslash N^{t-1}(e_0)$.\footnote{Here the analysis relies on the following fact, which is a consequence of how we generate the $(2\Delta - 1)$-edge coloring of $\InfTree$. Conditioning on the colors of the edges in $N^{t-1}(e_0)$, the  colors of the edges in $\bigcup_{j\in[1,\Delta-1]} N^{t}(e_j) \backslash N^{t-1}(e_0)$ and the  colors of the edges in $\bigcup_{j\in[\Delta,2\Delta-2]} N^{t}(e_j) \backslash N^{t-1}(e_0)$ are independent.}
Thus,
\begin{align}
\Pr[\BadZero \cap \BadOne \:|\; N^{t-1}(e_0)] &= \Pr[\BadZero \:|\: N^{t-1}(e_0)]\cdot \Pr[\BadOne \:|\: N^{t-1}(e_0)].\label{eqn:independent}\\
\intertext{Since $\DangerZero,\DangerOne$ are determined by $N^{t-1}(e_0)$, (\ref{eqn:independent}) implies
that $\Pr[\BadZero\cap\BadOne \:|\: \DangerZero\cap\DangerOne] \ge p^{2/3}$, and
with (\ref{eqn:BadZeroBadOne}) we deduce that}
\Pr[\DangerZero\cap\DangerOne] &\le 2p^{1/3}.\label{eqn:DangerZeroDangerOne}
\end{align}

\begin{figure}
    \centerline{\scalebox{.45}{\includegraphics{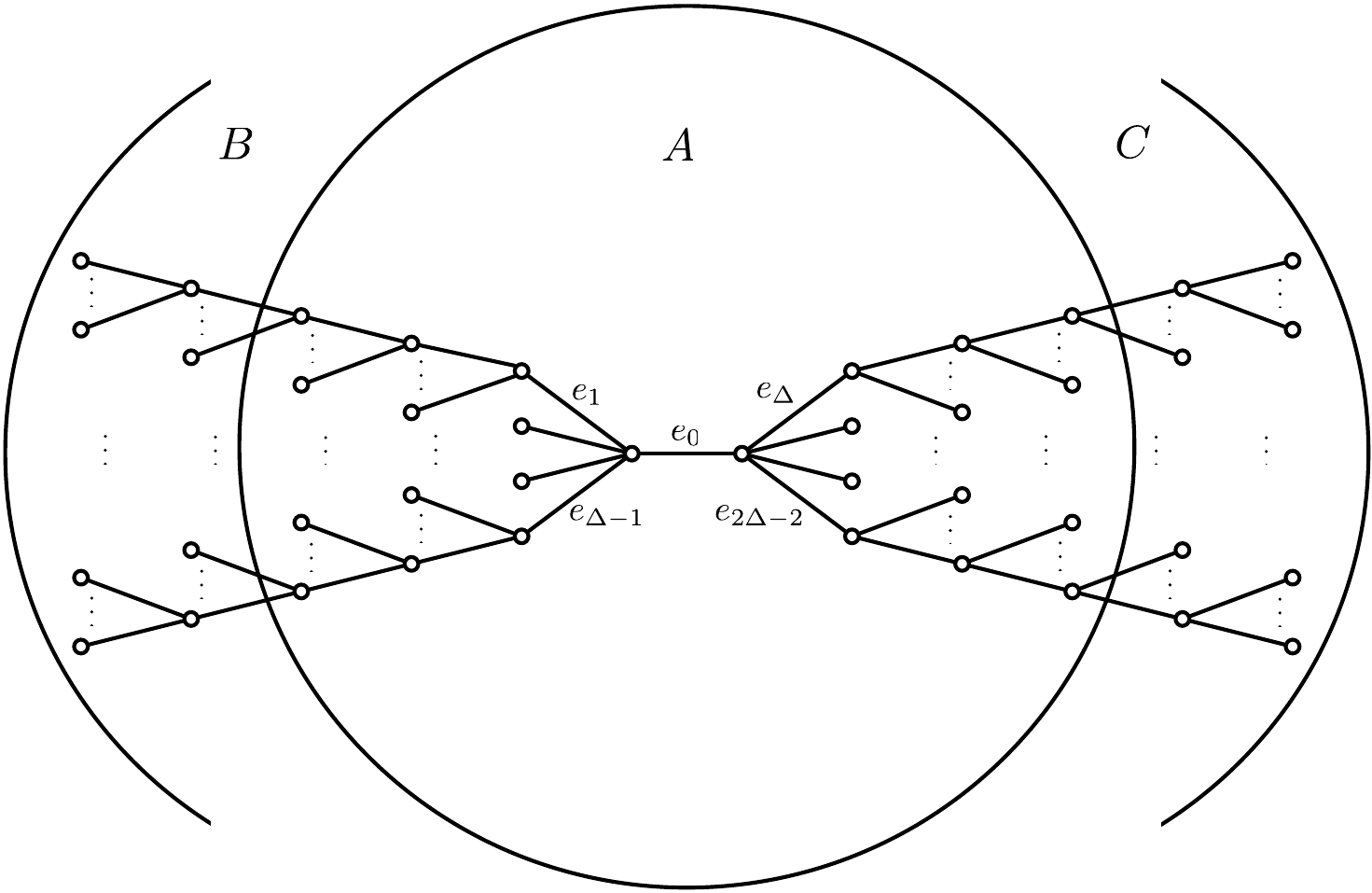}}}
    \caption{The algorithm for $e_0$ inspects its $(t-1)$-neighborhood $A=N^{t-1}(e_0)$ and estimates the probability that $\BadZero$ and $\BadOne$ occur.  Observe that $\BadZero$ is a function of $\bigcup_{j\in [1,\Delta-1]} N^t(e_j)$, which is completely contained in $A\cup B$, and $\BadOne$ is a function of $\bigcup_{j\in[\Delta,2\Delta-2]} N^t(e_j)$, which is completely
    contained in $A\cup C$.  Hence, after conditioning on $\DangerZero\cap \DangerOne$ (which depends only on $A$), $\BadZero$ and $\BadOne$ become independent since $B\cap C = \emptyset$.}
    \label{fig:dangerous-independent}
\end{figure}

The algorithm $A_{s.o.}'$ orients $e_0$ as follows.
\[
\mathcal{A}_{s.o.}'(e_0) = \left\{
\begin{array}{l@{\hcm[.5]}l}
0\to 1		& \mbox{ if $\DangerZero$ holds}\\
0\gets 1		& \mbox{ otherwise}
\end{array}\right.
\]
We now calculate the failure probabilities of $u_0$ and $u_1$.
\begin{align*}
\Pr[\mbox{$u_0$ is a sink}] &= \Pr[\overline{\DangerZero} \cap \BadZero]\\
					&\le \Pr[\BadZero \:|\: \overline{\DangerZero}] \;\le\; p^{1/3},		& \mbox{by definition of $\DangerZero$}\\
\Pr[\mbox{$u_1$ is a sink}] &= \Pr[\DangerZero \cap \BadOne]\\
					&\le \Pr[\DangerZero \cap \DangerOne] + \Pr[\BadOne \cap \overline{\DangerOne}]\\
					&\le 2p^{1/3} + p^{1/3} \;=\; 3p^{1/3},		& \mbox{by (\ref{eqn:DangerZeroDangerOne}) and the definition of $\DangerOne$.}
\end{align*}
The failure probability of the remaining vertices (those not incident to any edge colored $i$) is the same under $\mathcal{A}_{s.o.}$ and $\mathcal{A}_{s.o.}'$.
\end{proof}

\begin{lemma}\label{lem:zeroround}
Any sinkless orientation algorithm for $\InfTree$ with local error probability $p$
has time complexity $\Omega(\Delta^{-1}\log\log p^{-1})$.
\end{lemma}

\begin{proof}
Let $\mathcal{A}_{s.o.}$ be a $t$-round algorithm with error probability $p$, i.e.,
it has time profile $(t,t,\ldots,t)$.
Applying Lemma~\ref{lem:roundelimination} $t(2\Delta-1)$ times we get
an algorithm $\mathcal{A}_{s.o.}'$ with time profile $(0,0,\ldots,0)$ and
error probability $p_0 = O(p^{3^{-t(2\Delta-1)}})$.
We now claim that $p_0$ must also be at least $8^{-\Delta}$.  Any 0-round orientation
algorithm can be characterized by a real vector $(q_1,\ldots,q_{2\Delta-1})$,
where $q_i$ is the probability that an edge colored $i$ is oriented as $0\to 1$.
Without loss of generality, suppose that $q_1,\ldots,q_\Delta \ge 1/2$.
Fix any $v\in V(\InfTree)$ labeled 1.  The probability that $v$ is a sink is at least
the probability that its edges are initially colored $\{1,\ldots,\Delta\}$ and that they are all oriented to $v$,
hence $p_0 \ge {2\Delta-1 \choose \Delta}^{-1} \cdot 2^{-\Delta} \ge 2^{-3\Delta}$.
Combining the upper and lower bounds on $p_0$ we have
\begin{align*}
2^{3\Delta} &\ge p_0^{-1} \,=\, \Omega\left((p^{-1})^{3^{-t(2\Delta-1)}}\right),
\intertext{and taking logs twice we have}
\log(3\Delta) &\ge \log\log p^{-1} - t(2\Delta-1)\log 3 - O(1),
\end{align*}
\noindent which implies that $t = \Omega(\Delta^{-1}\log\log p^{-1})$.
\end{proof}

\begin{theorem}\label{thm:EC-lower-bound}
Even on 2-vertex colored trees or
2-vertex colored graphs of girth $\Omega(\log_\Delta n)$,
sinkless orientation and
$(2\Delta-2)$-edge coloring require $\Omega(\log_\Delta\log n)$ time in $\RandLOCAL$
and $\Omega(\log_\Delta n)$ time in $\DetLOCAL$.
\end{theorem}

\begin{proof}
Consider any sinkless orientation or $(2\Delta-2)$-edge coloring algorithm with local probability of failure $p$.
Lemma~\ref{lem:zeroround} applies to any vertex $v$ and any radius $t$ such that $N^t(v)$ is consistent with
a subgraph of $\InfTree$.  Thus, on degree-$\Delta$ trees or graphs of girth $\Omega(\log_\Delta n)$~\cite{Dahan14,Bollobas78},
we get $\Omega(\min\{\Delta^{-1}\log\log p^{-1},\, \log_\Delta n\})$ lower bounds.   Following the same proof
as~\cite[Theorem 5]{ChangKP19}, this implies an $\Omega(\log_\Delta n)$ lower bound in $\DetLOCAL$,
which, according to \cite[Theorem 3]{ChangKP19}, implies an
$\Omega(\log_\Delta\log n)$ lower bound in $\RandLOCAL$.  In other words, the weak $\RandLOCAL$ lower bound
$\Omega(\Delta^{-1}\log\log n)$ implied by Lemma~\ref{lem:zeroround}
\emph{automatically} implies a stronger lower bound.
\end{proof}

%% file: augment.tex
\section{Lower Bounds for Recoloring-Type Algorithms}\label{sect:Vizing}

In this section, we show that for $c\in[1,\frac{\Delta}{3}]$, any algorithm for $(\Delta+c)$-edge coloring based on 
\emph{extending partial colorings by recoloring subgraphs} 
needs $\Omega(\frac{\Delta}{c}\log\frac{cn}{\Delta})$ rounds.

\begin{theorem}\label{thm4}
Let $\Delta$ be the maximum degree and $c\in[1,\frac{\Delta}{3}]$.
For any $n$, there exists an $n$-vertex graph $G=(V,E)$
and a partial edge coloring
$\phi : E \rightarrow \{1,\ldots,\Delta+c,\bottom\}$, with exactly one uncolored edge $e_0$ ($\phi(e_0) =\; \bottom$)
satisfying the following property.
For any total edge coloring $\phi' : E\rightarrow \{1,\ldots,\Delta+c\}$ of $G$,
$\phi$ and $\phi'$ differ on a subgraph of diameter $\Omega(\fr{\Delta}{c}\log(\fr{cn}{\Delta}))$.
\end{theorem}

Suppose that $G$ is a partially $(\Delta+c)$-edge colored graph, where an edge $e_0$ in uncolored. A natural approach to color $e_0$ is to find an ``alternating path'' $e_0e_1\cdots e_\ell$, and then recolor the path. That is, for $0\leq i\leq \ell-1$, let the new color of $e_i$ be the old color of $e_{i+1}$, and then color the last edge $e_\ell$ by choosing any available color (if possible).
This type of approach has successfully led
to a distributed algorithm for Brooks' theorem~\cite{PanconesiS95}. Specifically, given a $(\Delta+1)$-vertex coloring, it was shown in~\cite{PanconesiS95} that
a $\Delta$-coloring can be computed in $\poly(\log n)$ time, independent of $\Delta$.  See Ghaffari et al.~\cite{GhaffariHKM18} for several faster algorithms.
However, Theorem~\ref{thm4} implies the existence of a graph where any alternating subgraph has diameter
$\Omega(\frac{\Delta}{c}\log\frac{cn}{\Delta})$, which is expensive for large $\Delta$.
The remainder of this section is a proof of Theorem~\ref{thm4}.

\paragraph{Construction.}
Without loss of generality, assume that $\Delta+c$ is even, and let $k=\frac{\Delta+c}{2}$.
We divide the color palette $\{1,\ldots,\Delta+c\}$ into two equal-size sets
$S_0=\{1,\ldots,k\}$ and $S_1=\{k+1,\ldots,\Delta+c\}$.
(One may refer to Figure~\ref{fig:augpath} for an example, with $\Delta=5, c=1$.  In the figure
blue edges are colored from palette $S_0$ and pink edges from $S_1$.)
Let $k'=\Delta-k$.

The graph $G^*(\ell,\Delta,c)$ consists of one uncolored edge $e_0=\{u_0,v_0\}$; all other
vertices are arranged in layers $1,\ldots,\ell$ and all other edges connect two vertices in adjacent layers
or layers $i$ and $i+3$, for some $i$.   In $G^*(\ell,\Delta,c)$, $e_0$ is a bridge and the subgraphs attached to $u_0$ and $v_0$
are structurally isomorphic, but colored differently.  Thus, we focus on the half of $G^*$ attached to $u_0$.

\paragraph{Base Case.}
Layer 1 consists of $k$ vertices attached to $u_0$.  They are initially colored with distinct colors from $S_0$.

\begin{figure}
\begin{center}
\includegraphics[width=.65\linewidth]{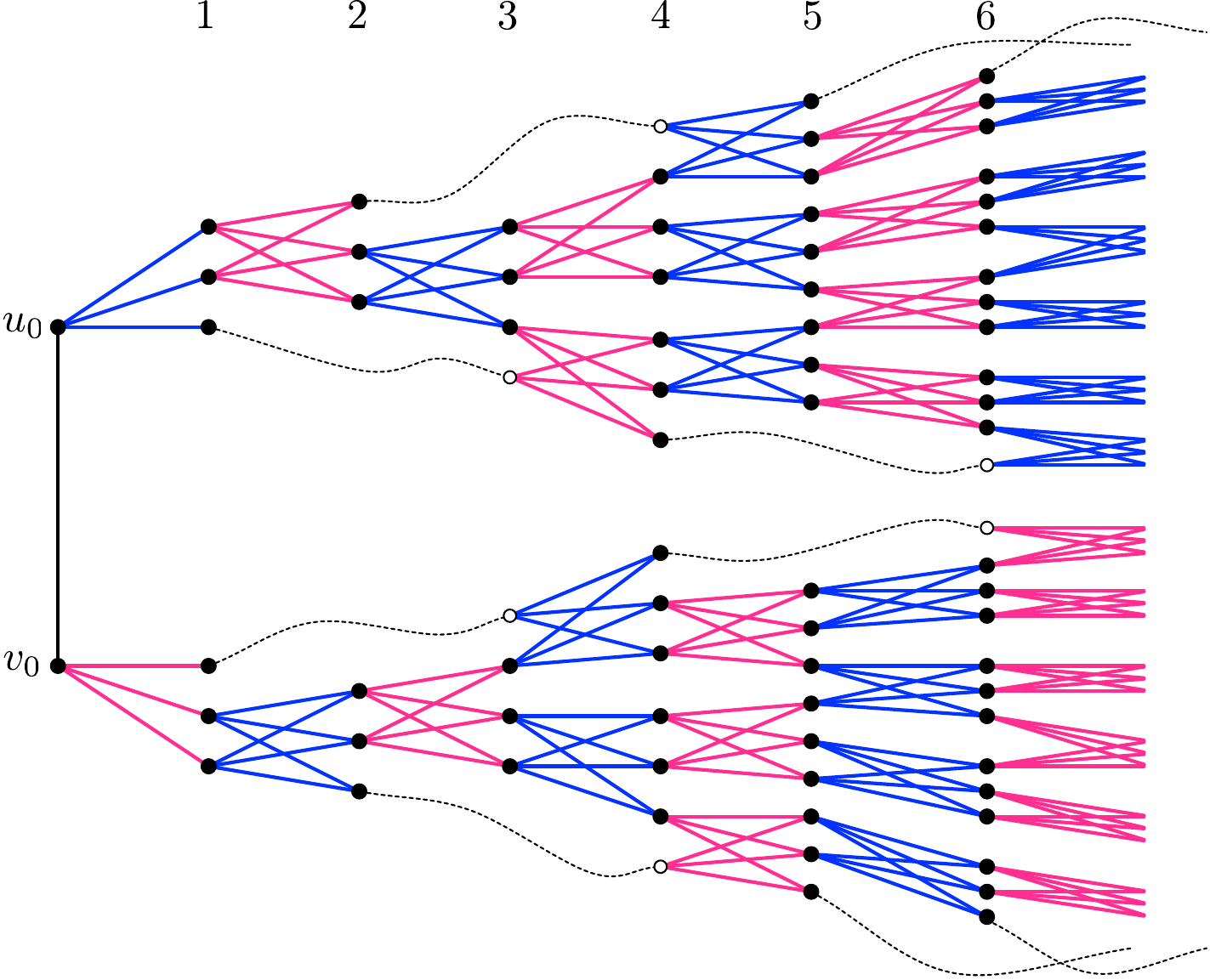}
\end{center}
\caption{An example of the construction when $\Delta=5$, $c=1$, $k=3$, $k'=2$, and $\ell \ge 7$.
Edges colored by palette $S_0=\{1,2,3\}$ are blue, and edges colored by palette $S_1=\{4,5,6\}$ are pink.
Leftover vertices in layer $i-2$ are also depicted (hollow) in layer $i$, and joined by a dashed curve.
They represent the same vertex, not two different vertices.}
\label{fig:augpath}
\end{figure}

\paragraph{Inductive Step.}
The $(i+1)$th layer is constructed as follows.  We take all the vertices at layer $i$ and the
\emph{leftover} vertices at layer $i-2$ and partition them into groups of size $k'$; any ungrouped
vertices are called \emph{leftovers} at level $i$.  (In Figure~\ref{fig:augpath} a leftover vertex in layer $i-2$ is drawn
twice, solid in layer $i-2$ and hollow when it is \emph{promoted} to layer $i$; they are connected by a
dashed line.) The grouping
is arbitrary, so long as all vertices promoted from layer $i-2$ are grouped.
Each group forms the lefthand side of a complete bipartite graph $K_{k',k}$.
Layer $i+1$ consists of the righthand side of all the (disjoint) copies of $K_{k',k}$.
All the edges in these graphs are properly colored with $S_b$ where $b=i\, \operatorname{mod} 2$.
(The subgraph attached to $v_0$ is constructed in the same way, except that we
flip the parity: the complete bipartite graphs are colored with $S_b$, $b=(i+1)\, \operatorname{mod} 2$.)

Define $n_i$ and $l_i$ as the number of layer-$i$ vertices and layer-$i$ leftover vertices.\footnote{The leftover vertices at layer $i-2$ are still considered as layer $i$ vertices, even though they have been promoted to layer $i$.}  According to the construction,
$(n_i)$ and $(l_i)$ satisfy the following recurrences.
\begin{align*}
n_1 &= k\\
l_{-1} = l_0 &= 0\\
n_{i+1} &= k\floor{\frac{n_i + l_{i-2}}{k'}} 	& \mbox{ for $i+1 \ge 2$}\\
l_{i}	&= (n_i + l_{i-2}) \, \operatorname{mod} k'			& \mbox{ for $i\ge 1$}
\end{align*}

Clearly $n_i = \Theta((k/k')^i)$.  Define $\epsilon = \frac{2c}{\Delta-c}$ so that $k/k' = \frac{\Delta+c}{\Delta-c} = 1+\frac{2c}{\Delta-c} = 1+\epsilon$.
The total number
of vertices in $G^*(\ell,\Delta,c)$ is $n = \Theta(\epsilon^{-1} n_\ell) = \Theta(\epsilon^{-1}(1+\epsilon)^\ell)$
and $\ell = \Theta(\log_{1+\epsilon}(\epsilon n)) = \Theta(\frac{\Delta}{c}\log\frac{cn}{\Delta})$.
In particular, when $c$ is constant and $\Delta < n^{1-\Omega(1)}$,
$\ell = \Omega(\Delta\log n)$.
The diameter of the graph is at least $\ell/3$ since, by construction, no edge crosses more than 3 layers.
 We remark that the purpose of the requirement $c \leq \Delta/3$ in the statement of Theorem~\ref{thm4} is to make $\epsilon \leq 1$. Our construction still applies to the case of $c > \Delta/3$, but it gives a worse bound on $\ell$ when $c$ is close to $\Delta$. 

Let $\phi$ be the initial partial edge-coloring of $G^*(\ell,\Delta,c)$, with $e_0$ left uncolored,
and $\phi'$ be any total edge-coloring.
We claim that $\phi'$ recolors at least one edge in the subgraph induced by layers $\ell-5,\ldots,\ell$.
Suppose otherwise.  Fix any vertex $v$ in layer $\ell-6$.
It has exactly $k$ neighbors in a higher layer,
either $\ell-5$ (if $v$ is not a leftover vertex) or $\ell-3$ (if $v$ is a leftover vertex); each such neighbor $u$ is adjacent to $k$ edges to a higher layer,
all of which are colored from the palette $S_1$ (without loss of generality, assume $\ell$ is even).
That means that all edges connecting $v$ to a higher layer must be colored from $S_0$.
By a reverse induction from $\ell-6$ down to $0$, it follows that all edges from $u_0$ to layer 1 must be colored with $S_0$.
A symmetric argument on $v_0$'s side shows that all edges from
$v_0$ to layer 1 must be colored with $S_1$, hence $e_0$ cannot be properly colored by $\phi'$.

%% file: randAlg.tex
\section{Randomized Edge Coloring Algorithm}\label{sect:randAlg}

Elkin, Pettie, and Su~\cite{ElkinPS15} showed that for any {\em constant} $\epsilon > 0$, 
there is a number $\Delta_\epsilon$ (depending only on $\epsilon$) such that for $\Delta > \Delta_\epsilon$,  $\Delta(1+\epsilon)$-edge coloring can be solved in
\[
O(T_{LLL}(n, \poly(\Delta),  \exp(-\epsilon^2\Delta/ \poly (\log \Delta))) + T^\ast(n, O(\Delta)))
\]
rounds in the $\RandLOCAL$ model, where
\begin{itemize}
\item $T_{LLL}(n,d,p)$ is the $\RandLOCAL$ complexity for constructive LLL with the parameters $d$ and $p$ on an $n$-vertex dependency graph.
\item $T^\ast(n,\Delta')$ is the $\RandLOCAL$ complexity for $5\Delta'$-edge coloring on an $n$-vertex graph of maximum degree $\Delta'$.
\end{itemize}
It is unclear to what extent the algorithm of~\cite{ElkinPS15} (or its predecessor~\cite{DubhashiGP98})
still works if we allow $\epsilon = o(1)$.  
For instance, is it possible to solve $(\Delta+\Delta^{0.7})$-edge coloring  in $\RandLOCAL$ in $O(\poly \log n)$ time?

\paragraph{Challenges to Reducing the Number of Colors.}
The analysis of our algorithm is substantially more involved than all previous edge coloring algorithms~\cite{PanconesiS97,DubhashiGP98,ElkinPS15}.  Here we give a short technical review of the types of
issues faced in distributed edge coloring.

Previous algorithms~\cite{ElkinPS15,DubhashiGP98} are based on the \emph{R\"{o}dl Nibble} method.
In each round, every uncolored edge nominates itself to be colored with probability $O(\epsilon)$ and remains idle otherwise;
a self-nominated edge picks a free color from its available palette and \emph{permanently} colors itself if the colors
selected by adjacent edges do not conflict with it.  The goal is to show that natural quantities (palette size, degree of vertices
in the uncolored graph, etc.) are sharply concentrated around their expectations.
The first issue is finding the right concentration bound.
Chernoff bounds are insufficient for several reasons, one of which is the need for independence (or negative dependence~\cite{DubhashiPanconesi09,DubhashiR98})
between the events of interest.
Azuma's inequality and variants fall short due to the weakness of Lipschitz properties (bounded differences).\footnote{This can be
seen by considering the problem of bounding the $c$-degree of a vertex $v$ (the number of edges incident to $v$ with color $c$ in their palettes).
This quantity potentially depends on the choices of $\Omega(\Delta^3)$ edges within distance 3 of $v$, 
and each such choice could affect $v$'s $c$-degree by 1 or more.  The sum of these Lipschitz constants completely dwarfs the expected 
$c$-degree, which makes Azuma-type inequalities inapplicable.}
The algorithm of Dubhashi, Grable, and Panconesi~\cite{DubhashiGP98}
used a specialized concentration inequality of Grable~\cite{Grable98}, whereas our algorithm and that of 
Elkin, Pettie, and Su~\cite{ElkinPS15} use one~\cite[Thm.~(8.5)]{DubhashiPanconesi09} 
that is syntactically closer to Chernoff/Hoeffding/Azuma-type inequalities.  
(It is restated as Theorem~\ref{lem:concentration} in Appendix~\ref{sect:concentration}.) 

The purpose of the ``self-nomination'' step in~\cite{DubhashiGP98,ElkinPS15} is to simplify certain aspects of the analysis.  
For example, the probability that an edge is successfully colored, \emph{conditioned on it nominating itself}, 
is a very high $1-O(\epsilon)$.  Because of this, we can afford to toss out any color $c$ from $e$'s palette if any 
\emph{self-nominated} 
edge $e'$ adjacent to $e$ selects $c$ --- regardless of whether $e'$ successfully colors itself. 
This type of subtle change generally makes things simpler.  Some events which would ordinarily be dependent 
become independent, and some variables (e.g., a vertex's $c$-degree) now depend on $\Theta(\Delta^2)$ variables
rather than $\Theta(\Delta^3)$.  The downside of this approach is that $\Omega(\epsilon^{-1})$ steps are necessary
to color a large fraction of the graph, and \emph{with each coloring step}
the quantities we are monitoring ($c$-degree, palette size, etc.) deviate further from their expectations.
When $\epsilon^{-1}$ is polynomial in $\Delta$, the accumulated deviation errors make it impossible to achieve
palette sizes as small as $\Delta + \tilde{O}(\sqrt{\Delta})$.  

\paragraph{Our Approach.}  Our algorithm is more ``natural'' than~\cite{DubhashiGP98,ElkinPS15}.  Roughly speaking, in each step each
edge chooses a color uniformly at random from its available palette and permanently colors itself if there are no local
conflicts (\oneshot).  I.e., we dispense with the low probability self-nomination step.  Let $p_i$ be a lower bound on the palette
size after $i$ such steps, and $d_i,t_i$ be upper bounds on uncolored degree and $c$-degree of any vertex, respectively.
It is straightforward to show that if everything behaves precisely according to expectation, the $(d_i)$ sequence 
shrinks by a $(1-e^{-2})$ factor in each step and both $(p_i), (t_i)$ shrink by a $(1-e^{-2})^2$ factor.  In reality
these quantities \emph{do} deviate from their expectations, and even tiny, $(1+o(1))$-factor deviations 
compound themselves and spin out of control.  One reason our analysis is more complex than~\cite{DubhashiGP98,ElkinPS15} 
is that we look at concentration up to \emph{lower order terms}.  For example, although $p_i \approx t_i$,
we bound $\beta_i = \frac{p_i}{t_i}-1$, which captures accumulated errors beyond the leading constants.

\paragraph{The Use of the Distributed LLL.} As in~\cite{ElkinPS15}, we obtain good concentration on $d_i,p_i,t_i$
with probability $1-\exp(-\epsilon^2\Delta/\log^{4+o(1)}\Delta)$,
which is $1-1/\poly(n)$ if $\Delta$ and $\epsilon$ are sufficiently large.  If not, we must invoke a distributed LLL
algorithm to make sure each random coloring experiment introduces bounded deviation errors in $d_i,p_i,t_i$.
A constant fraction of the edges are colored in each step. 
For many parameter regimes the running
time is dominated by $O(\log\epsilon^{-1})$ calls to an distributed LLL algorithm, as our algorithm needs to run \oneshot\ for this many iterations.

\subsection{Our Result}
In this section, we prove the following theorem, which improves upon the algorithms of~\cite{DubhashiGP98,ElkinPS15}. Here $T^\ast(n, \Delta')$ is the $\RandLOCAL$ complexity of the $5 \Delta'$-edge coloring problem, and $\TLLL(n,p,d)$ is the complexity of distributed LLL with parameters $p$ and $d$.

\begin{theorem}\label{thm:randAlg-main}
Let $\epsilon = \omega\left(\frac{\log^{2.5} \Delta}{\sqrt{\Delta}}\right)$
be a function of $\Delta$.  
There is a $\RandLOCAL$ algorithm for $(1+\epsilon)\Delta$-edge coloring 
in time
\[
O\left(\log(1 / \epsilon)\right) \cdot T_{LLL}\left(n, d, p\right)
\;+\; T^\ast\left(n, O(\epsilon \Delta)\right),
\]
where $p = \exp(-\epsilon^2 \Delta / \log^{4+o(1)} \Delta) = \exp(-\omega(\log \Delta))$ and $d = O(\poly(\Delta))$.
\end{theorem}

The statement of Theorem~\ref{thm:randAlg-main} guarantees that whenever $\epsilon$ and $\Delta$ satisfy the specified condition, we always have $\exp(-\epsilon^2 \Delta / \log^{4+o(1)} \Delta) = \exp(-\omega(\log \Delta))$, and so we may use a distributed LLL algorithm under any criterion $p(ed)^\lambda < 1$.
There is an inherent tradeoff between the palette size and the runtime in Theorem~\ref{thm:randAlg-main}.
Selecting smaller $\epsilon$ allows us to use fewer colors, but it leads to a higher
$p = \exp(-\epsilon^2 \Delta / \log^{4+o(1)} \Delta)$, which may increase the runtime of the LLL algorithm.

\paragraph{Runtime of $5\Delta'$-edge Coloring.}
It is known that $T^\ast(n,\Delta')$ is at most $O(\log \Delta')$ plus the $\DetLOCAL$ complexity of
 $3\Delta'$-edge coloring on $\poly(\log n)$-size graphs.
This is achieved by applying the $(\tilde{\Delta}+1)$-vertex coloring algorithm of~\cite{BEPS16} to the line graph, where $\tilde{\Delta} = 2\Delta'-2$ is the maximum degree of the line graph.

For the special case of $\Delta' = \log^{1+\Omega(1)} n$, $(2\Delta' - 1)$-edge coloring can be solved in $\RandLOCAL$ $O(\log^\ast n)$ rounds~\cite{ElkinPS15}. The state-of-the-art $\DetLOCAL$ 
algorithm~\cite{GhaffariHKMSU17} for $(2+x)\Delta'$-edge coloring has complexity
\[
O(\log^2 \Delta'\cdot x^{-1}\cdot \log \log \Delta' \cdot \log^{1.71} \log \log \Delta' \cdot \log n)
\]
for any $x > 1/\log \Delta'$.
Thus, combining~\cite{ElkinPS15,BEPS16,GhaffariHKMSU17} with $x=1$, we have
\[
T^\ast(n,\Delta') = O(\log^3 \log n \cdot \log \log \log n \cdot \log^{1.71} \log \log \log n ) = (\log\log n)^{3+o(1)}.
\]
This is achieved as follows. If $\Delta' = \Omega(\log^2 n)$, we run the $O(\log^\ast n)$-time $\RandLOCAL$ algorithm of~\cite{ElkinPS15}. Otherwise, we run the $\RandLOCAL$ graph shattering phase of~\cite{BEPS16} (using the first $2\Delta'$ colors) followed by the
$\DetLOCAL$ algorithm of~\cite{GhaffariHKMSU17} (using the remaining $3\Delta'$ colors) on each component.

\subsection{Time Complexity Analysis}
We calculate the time complexity  for Theorem~\ref{thm:randAlg-main} in different parameter regimes of $\Delta$ and $\epsilon$.

\paragraph{Running Time on General Graphs.}   Our algorithm computes a $(1+\epsilon)\Delta$-edge coloring in  $O(\log n)$ time when 
$\epsilon = \omega(\log^{2.5} \Delta) / \sqrt{\Delta})$. Observe that in this parameter regime, we have $\log(1/\epsilon) = O(\log \Delta)$. Applying the distributed LLL algorithm of Chung, Pettie, and Su~\cite{ChungPS17} under the criterion $e p d^2 < 1$, we obtain $\TLLL(n,d,p) = O(\log_{1/ep(d+1)} n) = O(\log_{\Delta} n)$, as $\log (1/p) = \omega(\log \Delta)$.  Therefore, $O\left(\log(1 / \epsilon)\right) \cdot T_{LLL}\left(n, d, p\right) = O(\log n)$.
By applying one of the Ghaffari-Harris-Kuhn LLL algorithms~\cite{FischerG17,GhaffariHK17},
the cost per LLL is $O(d^2 + 2^{O(\sqrt{\log\log n})}) = O(\Delta^6 + 2^{O(\sqrt{\log\log n})})$.

The term $T^\ast\left(n, O(\epsilon \Delta)\right) = (\log \log n)^{3 + o(1)}$ can become the 
dominant term when $\Delta$ is sufficiently large.
In particular, when $\epsilon = \Omega(1)$, we have $\log(1/\epsilon) = O(1)$, and so our algorithm is able to finish in
$O(\log_{\Delta} n) + (\log \log n)^{3 + o(1)}$ time.

\paragraph{Runtime on Trees.}
Consider running our algorithm on a tree with palette size $(1+\epsilon)\Delta$, where
$\epsilon = \Omega\left(\frac{\log^{2.5+x} \Delta}{\sqrt{\Delta}}\right)$, for some positive constant $x$.
Then the LLL parameters are $d = \poly(\Delta)$ and $p = \exp(-\epsilon^2 \Delta / \log^{4+o(1)} \Delta)$ in Theorem~\ref{thm:randAlg-main}, which satisfy the
criterion $p(ed)^\lambda < 1$ with $\lambda = \Omega(\log^x \Delta)$.
Using our randomized LLL algorithm for tree-structured dependency graphs (Section~\ref{sect:tree-LLL}), we have
\[
T_{LLL}\left(n, \poly(\Delta), \exp\left(-\epsilon^2 \Delta / \log^{4+o(1)}\right)\right) = O\left(\max\left\{\fr{\log\log n}{\log\log\log n},\, \log_{\log \Delta} \log n\right\}\right).
\]

We claim that $T^\ast(n,\Delta')=O(\log^\ast \Delta' + \log_{\Delta'} \log n)$ on trees. This is achieved as follows. First, do a $O(\log^\ast \Delta')$-time randomized procedure to partially color the graph using the first $2\Delta'$ colors so that the remaining uncolored components have size $\poly(\log n)$. This can be done using the algorithm of~\cite{ElkinPS15} without invoking any distributed LLL algorithm. Then, apply our deterministic $O(\log_{\Delta'} \tilde{n})$-time algorithm for $\Delta'$-edge coloring trees
(Section~\ref{sect:upperbound}) to each uncolored component separately, using a set of $\Delta'$ fresh colors.

To sum up, the time complexity of $(1+\epsilon)\Delta$-edge coloring trees is
\begin{align*}
&O\left(\log(1 / \epsilon) \cdot \max\left\{\fr{\log\log n}{\log\log\log n},\, \log_{\log \Delta} \log n\right\} + \log^\ast \Delta + \log_\Delta \log n  \right)\\
&= O\left(\log(1 / \epsilon) \cdot \max\left\{\fr{\log\log n}{\log\log\log n},\, \log_{\log \Delta} \log n\right\}\right).
\end{align*}
This nearly matches our $\Omega(\log_\Delta \log n)$ lower bound (Section~\ref{sect:lowerbound}).

For the case  $\epsilon = \Omega(1)$, our algorithm runs faster, as we can use $\lambda = \Delta / \poly \log \Delta$, and so the running time of the distributed LLL becomes $O\left( \max\left\{\fr{\log\log n}{\log\log\log n},\, \log_{\Delta} \log n\right\}\right)$. 
In this case, our algorithm finds a $(1+\epsilon)\Delta$-edge coloring in $O\left( \max\left\{\fr{\log\log n}{\log\log\log n},\, \log_{\Delta} \log n\right\}\right)$ time.

\subsection{The Algorithm and Its Invariants}\label{sect:alg-phase1}
Our algorithm has two phases. The goal of the first phase is to color a subset of the edges using the colors from $\mathcal{C}_1 \bydef \{1, \ldots, \Delta(1+\xi)\}$ such that the subgraph induced by the uncolored edges has degree less than $\Delta' = \frac{1}{5}(\epsilon  - \xi)\Delta = \Theta(\epsilon \Delta)$. The first phase consists of $O(\log(1 / \epsilon))$ executions of a distributed \Lovasz\ Local Lemma algorithm. The second phase colors the remaining edges using the colors from  $\mathcal{C}_2 \bydef \{\Delta(1+\xi)+1, \ldots, \Delta(1+\epsilon)\}$ using the fastest available coloring algorithm, which takes $T^\ast(n, \Delta')$ time.

\paragraph{Algorithm.} In what follows we focus on the first phase. We write $G_i$ to denote the graph induced by the set of uncolored edges at the beginning of the $i$th iteration. Each edge $e$ in $G_i$ has a palette $\Psi_i(e) \subseteq \mathcal{C}_1$. We write $\deg_i(v)$ to denote the number of edges incident to $v$ in $G_i$ and
$\deg_{c,i}(v)$ to denote the number of edges incident to $v$ that have color $c$ in their palettes.
For the base case, we set $G_1 = G$ and $\Psi_i(e)=\mathcal{C}_1$ for all edges.
In the graph $G_i$ we maintain the following invariant $\mathcal{H}_i$.

\medskip
\noindent{\bf Invariant $\mathcal{H}_i$:} For each edge $e$, vertex $v$, and color $c$,  we have:
\begin{align*}
\deg_i(v)&\leq d_i,\\
\deg_{c,i}(v)&\leq t_i,\\
|\Psi_i(e)|&\geq p_i.
\end{align*}

\paragraph{Parameters.}  Given two numbers $\eta \geq 1$ and $\xi \in (0,\epsilon)$ (which are functions of $\Delta$),
we define three sequences of numbers $\{d_i\}$, $\{t_i\}$, and $\{p_i\}$ as follows.

\medskip

Base case ($i=1$):
$$d_1 \bydef \Delta  \ \ \ \ \ \ \ \  t_1 \bydef \Delta \ \ \ \ \ \ \ \ p_1 \bydef \Delta(1 + \xi)$$

Inductive step ($i>1$):
\begin{align*}
&d_{i} \bydef (1+\delta_{i-1}) d_{i-1}^\diamond & & d_{i-1}^\diamond  \bydef d_{i-1} \cdot \left(1 - (1 - 1/p_{i-1})^{2(t_{i-1} - 1)}\right) \\
&t_{i} \bydef (1+\delta_{i-1}) t_{i-1}^\diamond & & t_{i-1}^\diamond  \bydef t_{i-1} \cdot \left(1 - \frac{t_{i-1}}{p_{i-1}}(1 - 1/p_{i-1})^{2t_{i-1}}\right)\left(1 - (1 - 1/p_{i-1})^{2t_{i-1}}\right)\\
&p_{i} \bydef (1-\delta_{i-1}) p_{i-1}^\diamond & & p_{i-1}^\diamond  \bydef p_{i-1} \cdot \left(1 - \frac{t_{i-1}}{p_{i-1}}(1 - 1/p_{i-1})^{2t_{i-1}}\right)^2
\end{align*}

Drifts (all $i$):
$$\delta_i \bydef \frac{\beta_i}{\eta}
\ \ \ \ \ \ \ \  \ \ \ \ \beta_i \bydef \frac{p_i}{t_i} - 1
\ \ \ \ \ \ \ \  \ \ \ \ \text{(Notice that $\beta_1 = \xi$)}$$

\medskip

The choice of parameters are briefly explained as follows. Consider an ideal situation where $\deg_{i-1}(v)=d_{i-1}$, $\deg_{c,i-1}(v) = t_{i-1}$, and $|\Psi_{i-1}(e)| = p_{i-1}$ for all $c$, $e$, and $v$. 
Consider a very simple experiment called \oneshot{} 
in which each uncolored edge attempts to color itself by 
selecting a color uniformly at random from its available palette. 
An edge $e$ successfully colors itself with probability $(1 - 1/p_{i-1})^{2(t_{i-1} - 1)}$, 
since there are $2(t_{i-1}-1)$ edges competing with $e$ for $c \in \Psi_{i-1}(e)$, and each of these $2(t_{i-1}-1)$ edges selects $c$ with probability $1/p_{i-1}$.
Thus, by linearity of expectation, the expected degree of $v$ after \oneshot\ is $d_{i-1}^\diamond$,
and the parameter $d_{i}$ is simply $d_{i-1}^\diamond$ with some slack.
The parameters $\{t_{i-1}^\diamond, t_i, p_{i-1}^\diamond, p_i\}$ carry analogous meanings.
The term $\beta_i$ represents the {\em second-order} error. We need control over $\{\beta_i\}$ since it influences the growth of the three sequences $\{d_i\}$, $\{t_i\}$, and $\{p_i\}$.

For the base case, it is straightforward to see that we have $\deg_1(v) = \Delta$, $\deg_{c,1}(v) = \Delta$, and $|\Psi_1(e)| = \Delta(1+\xi)$, and thus $G_1$ satisfies the invariant $\mathcal{H}_1$.
For the inductive step, given that $\mathcal{H}_i$ is met in $G_i$, we use a distributed LLL algorithm (based on \oneshot) to color a subset of edges in $G_i$ so that the next graph $G_{i+1}$ 
induced by the uncolored edges satisfies $\mathcal{H}_{i+1}$.

\paragraph{Termination of the First Phase.}
The number of iterations of our algorithm will be 
$i^\star-1 = O(\log (1/\epsilon))$ (Lemma~\ref{lem:estimate2}). We will later see that after the $(i^\star-1)$th iteration,  the degree of the vertices in the remaining uncolored part of the graph  $G_{i^\star}$ satisfies the ``terminating condition''  $d_{i^\star} \leq \frac{1}{5} (\epsilon - \xi)\Delta$. Then we proceed to the second phase.

The purpose for requiring this condition is to create a sufficiently large gap between the maximum degree $\Delta'$ (in the remaining uncolored part of the graph) and the number of available colors (the colors $\mathcal{C}_2 = \{\Delta(1+\xi)+1, \ldots, \Delta(1+\epsilon)\}$ reserved for the second phase), so that we can run a $5\Delta'$-edge coloring algorithm to color all remaining edges in the second phase of the algorithm.

\paragraph{Analysis.}
Recall that $\epsilon = \omega\left(\frac{\log^{2.5} \Delta}{\sqrt{\Delta}}\right)$.
We set $\eta$ to be any function of $\Delta$ that is $\omega(\log\Delta)$
such that $\epsilon \ge \frac{\eta^{2.5}}{\sqrt{\Delta}}$.
We set $\xi = \frac{\epsilon}{6 \eta}$. 
The following lemma shows that under certain criteria, the parameters $\{d_i\}$, $\{t_i\}$, $\{p_i\}$, and $\{\beta_i\}$ are very close to their ``ideal'' values.
The proof is deferred to Section~\ref{sect:estimate}.

\begin{lemma}\label{lem:estimate}
Consider an index $i>1$.
Suppose $\min\{d_{i-1},t_{i-1},p_{i-1}\}=\omega(\log \Delta)$,  $\beta_{i-1} = o(1/ \log \Delta)$, and $\delta_{i-1} = o(\beta_{i-1}/ \log \Delta)$.
Then the following four equations hold.
\begin{align*}
d_i &= d_{i-1} \cdot (1 \pm o(1/\log \Delta)) \cdot (1 - e^{-2}) \\
t_i &= t_{i-1} \cdot (1 \pm o(1/\log \Delta)) \cdot  (1 - e^{-2})^2 \\
p_i &= p_{i-1} \cdot (1 \pm o(1/\log \Delta)) \cdot  (1 - e^{-2})^2 \\
\beta_i &= \beta_{i-1} \cdot  (1 \pm o(1/\log \Delta))  \,/\, (1 - e^{-2})
\end{align*}
\end{lemma}

Based on Lemma~\ref{lem:estimate}, we have the following lemma.

\begin{lemma}\label{lem:estimate2}
Let $i^\star = O(\log (1/\epsilon)) = O(\log \Delta)$ be the largest index such that $\beta_{i^\star-1} \leq 1/\eta$.
Then the following four equations hold for any $1 < i \leq i^\star$.
\begin{align*}
d_i &= (1 \pm o(1/\log \Delta))^{i-1} \Delta (1 - e^{-2})^{i-1}  = (1 \pm o(1)) \Delta (1 - e^{-2})^{i-1} \\
t_i &= (1 \pm o(1/\log \Delta))^{i-1} \Delta (1 - e^{-2})^{2(i-1)} = (1 \pm o(1)) \Delta (1 - e^{-2})^{2(i-1)} \\
p_i &= (1 \pm o(1/\log \Delta))^{i-1} \Delta (1 - e^{-2})^{2(i-1)} = (1 \pm o(1)) \Delta (1 - e^{-2})^{2(i-1)} \\
\beta_i &= (1 \pm o(1/\log \Delta))^{i-1} \xi / (1 - e^{-2})^{i-1} = (1 \pm o(1)) \xi / (1 - e^{-2})^{i-1}
\end{align*}
\end{lemma}
\begin{proof}
To prove the lemma, it suffices to show that the condition of Lemma~\ref{lem:estimate} is met for all indices $1 < i \leq i^\star$.
We prove this by an induction on $i$.
By the induction hypothesis the four equations hold at index $i-1$.
We show that the condition of Lemma~\ref{lem:estimate} is met for the index $i$, and so the four equations also hold for index $i$.
Due to $1/\eta = o(1/\log \Delta)$, we already have $\beta_{i-1} = o(1/ \log \Delta)$ and $\delta_{i-1} = o(\beta_{i-1}/ \log \Delta)$.
It remains to prove that $\min\{d_{i-1},t_{i-1},p_{i-1}\}=\omega(\log \Delta)$.
\begin{align*}
&\min\{d_{i-1},t_{i-1},p_{i-1}\} \\
&\geq (1 \pm o(1)) \Delta (1 - e^{-2})^{2(i-1)} & \text{(Induction hypothesis for $d_{i-1},t_{i-1},p_{i-1}$)}\\
&= (1 \pm o(1)) \Delta (1 - e^{-2})^{2(i-2)} (1-e^{-2})^2\\
&= (1 \pm o(1)) \Delta \cdot \left( \frac{(1 - e^{-2} \pm o(1))\xi}{\beta_{i-1}}\right)^2 & \text{(Induction hypothesis for $\beta_{i-1}$)}\\
&\geq (1 - e^{-2} \pm o(1)) \xi^2 \eta^2 \Delta 							&(\beta_{i-1} \leq 1/\eta)\\
&= \Omega(\eta^5) 												& 
\left(\xi = \Omega\left(\frac{\eta^{1.5}}{\sqrt{\Delta}}\right)\right)\\
&= \omega(\log \Delta) & \qedhere
\end{align*}
\end{proof}

It remains to show that (i) the terminating condition $d_{i^\star} \leq \frac{1}{5} (\epsilon - \xi)\Delta$ is satisfied at the end of the $(i^\star-1)$th iteration,
 and (ii)
in each iteration, in $T_{LLL}\left(n, \poly(\Delta), \exp(-\epsilon^2 \Delta / \log^{4+o(1)} \Delta)\right)$ time, invariant $\mathcal{H}_i$ can be maintained.
By Lemma~\ref{lem:estimate2}, we have:
\begin{align*}
d_{i^\star} &= (1 \pm o(1)) \Delta (1 - e^{-2})^{i^\star-1} & (\text{Lemma~\ref{lem:estimate2} for $d_{i^\star}$})\\
&= (1 \pm o(1)) \Delta \cdot \xi / \beta_{i^\star} & (\text{Lemma~\ref{lem:estimate2} for $\beta_{i^\star}$}) \\
&\leq (1 \pm o(1)) \xi \eta \Delta & (\beta_{i^\star} > 1/\eta)
\end{align*}
For our choices of $\eta$ and $\xi$, we have $d_{i^\star} \approx \xi \eta \Delta = \frac{\epsilon \Delta}{6}$.
Since $\frac{1}{5} (\epsilon - \xi)\Delta > \frac{\epsilon \Delta}{6}$, the condition $d_{i^\star} \leq \frac{1}{5} (\epsilon - \xi)\Delta$ is satisfied.

For each $1 < i \leq i^\star$, we have:
\begin{align*}
\delta_i^2 \cdot  \min\{d_{i},t_{i},p_{i}\} &= \beta_i^2 t_i / \eta^2  &(\text{Definition of $\delta_i$})\\
&= (1 \pm o(1)) \cdot \left( \xi / (1 - e^{-2})^{i-1}\right)^2 \cdot \left(\Delta (1 - e^{-2})^{2(i-1)} \right) / \eta^2
& (\text{Lemma~\ref{lem:estimate2} for $t_{i}$, $\beta_{i}$})\\
&= (1 \pm o(1)) \cdot \Delta(\xi / \eta)^2  \\
&= \Omega(\epsilon^2 \Delta / \eta^4) &\text{(Definition of $\xi$)}\\
&= \omega(\log \Delta). &\text{(Definition of $\epsilon$)}
\end{align*}
We will later see in Section~\ref{sect:oneshot} that this implies that any LLL algorithm with parameters
$d = \poly(\Delta)$ and $p = \exp(-\Omega(\Delta \epsilon^2 / \eta^4))$ suffices to maintain the invariant in each iteration.
Notice that if we select $\eta = \log^{1 + o(1)}\Delta$, then $p = \exp(-\epsilon^2 \Delta / \log^{4+o(1)} \Delta)$, as desired.

\subsection{Maintenance of the Invariant}\label{sect:oneshot}
In this section we show how to apply a distributed LLL algorithm, with parameters $d = \poly(\Delta)$ and $p = \exp(-\Omega\left(\delta_i^2 \cdot \min\{d_i,t_i,p_i\}\right)$, to achieve the following task: given a graph $G_i$ meeting the property $\mathcal{H}_i$, color a subset of edges of $G_i$ so that the graph induced by the remaining uncolored edges satisfies the property $\mathcal{H}_{i+1}$. We write $\Psi(e) = \Psi_i(e)$ for notational simplicity.

\paragraph{Achieving Uniform Progress.}
Consider the following modifications to the underlying graph $G_i$:
\begin{itemize}
\item Each edge $e$ discards some arbitrary colors from its palette to achieve uniform palette size $p_i$.
\item Each vertex $v$ locally simulates some \emph{imaginary} subtrees attached to $v$ and obeying $\mathcal{H}_i$ to achieve uniform color degree $t_i$. That is, if a color $c$ appears in the palette of some edge incident to a vertex $v$, then $c$ must appear in the palette of {\em exactly} $t_i$ edges incident to $v$.
\end{itemize}

These modifications to the underlying graph are introduced to enforce broadly \emph{uniform}
progress in every part of the 
graph.\footnote{The algorithm will likely work if it is run on the actual
graph (that is, without hallucinating imaginary subtrees), but we do not see a way to enforce the same invariants.  For example, suppose a vertex $v$ in $G_i$ has degree exactly $d_i$ but because the palettes in 
$v$'s neighborhood happen to be advantageously configured, $v$'s 
degree after one coloring step is likely to be \emph{much less} than $d_{i+1}$.  Surely this is a good outcome!  Yet, 
if more edges are colored than we expect, the remaining edges will
lose \emph{more colors from their palettes} than we expect, 
possibly violating the lower bound on $p_{i+1}$.  These concerns 
motivate us to enforce more uniform progress, 
hence the introduction of imaginary trees.}

Observe that if $\mathcal{H}_{i}$ applies to the imaginary graph it 
also applies to the true graph as well, since we are concerned with \emph{lower} bounds on palette sizes
and \emph{upper} bounds on $c$-degrees. 

To increase the $c$-degree of each vertex $v$ to $t_i$, we might need to add so many imaginary edges to $v$ such that the degree of $v$ exceeds $d_i$ if we take into account these imaginary edges. This is fine, as we will later see that  we only consider the real edges when we analyze the shrinking rate of the degree.\footnote{In particular,
if we want to increase the $c$-degree of $v$ by $k$, we can add $k$ imaginary edges $e_1, \ldots, e_k$ incident to $v$ such that the palette of each newly added edge contains $c$. Other than the color $c$, there is no overlap between the palettes 
of the newly added edges and other edges incident to $v$.}

\paragraph{One Shot Coloring.}
Our analysis focusses largely on how the
following $O(1)$-round procedure affects the imaginary graph.

\begin{framed}
\noindent \oneshot.
\begin{description}
\item (1) Each edge $e$ selects a color $\Color(e) \in \Psi(e)$ uniformly at random.
\item (2) An edge $e$ successfully colors itself  $\Color(e)$ if no neighboring edge also selects $\Color(e)$.
\end{description}
\end{framed}

We write $S(v)$ to denote the set of \underline{real} edges incident to $v$, and we write $N_c(v)$ to denote the set of \underline{real and imaginary} edges incident to $v$ that have $c$ in their palettes. Let $S^\diamond(v)$ (resp., $N_c^\diamond(v)$) be the subset of $S(v)$ (resp., $N_c^\diamond(v)$) that are still uncolored after \oneshot. Let $\Psi^\diamond(e)$ be the result of removing all colors $c$ from $\Psi(e)$ such that some edge incident to $e$ successfully colors itself by $c$.

The following concentration bound implies that $\mathcal{H}_{i+1}$ holds with high probability in the graph induced by the \underline{real}
uncolored edges after \oneshot, and thus we can apply a distributed LLL algorithm to obtain $G_{i+1}$ that meets the invariant $\mathcal{H}_{i+1}$.
See Appendix~\ref{sect:concentration} for proof.

\begin{lemma}\label{thm:concentration}
Suppose that $\mathcal{H}_i$ holds.
The following concentration bounds hold for any $\delta > 0$.
\begin{align*}
&\Prob\left[|S^\diamond(v)| > (1+\delta) d_i^\diamond \right]  =  \exp\left(-\Omega(\delta^2 d_i)\right) \\
&\Prob\left[|N_c^\diamond(v)| > (1+\delta) t_i^\diamond \ | \ N_c^\diamond(v) \neq \emptyset \right]   =  \exp\left(-\Omega(\delta^2 t_i)\right) \\
&\Prob\left[|\Psi^\diamond(e)| < (1-\delta) p_i^\diamond \ | \ e \text{ remains uncolored } \right]   =  \exp\left(-\Omega(\delta^2 p_i)\right)
\end{align*}
\end{lemma}

We write $N^k(v)$ to denote the set of all vertices within distance $k$ of $v$.
It is straightforward to see that (i) $S^\diamond(v)$ depends only on the colors selected by the edges whose endpoints are both 
in $N^2(v)$, (ii) $N_c^\diamond(v)$ depends only on the colors selected by the edges whose endpoints are both in $N^3(v)$, and (iii) $\Psi^\diamond(e)$ depends only on the colors selected by the edges whose endpoints are both in $N^2(u) \cup N^2(v)$, where $e = \{u,v\}$.

Thus, the parameters for the LLL are $d = \poly(\Delta)$ and $p = \exp\left(-\Omega\left(\delta_i^2 \cdot \min\{d_i,t_i,p_i\}\right)\right)$ by Lemma~\ref{thm:concentration}. Recall from the calculation in Section~\ref{sect:alg-phase1} that $\delta_i^2 \cdot \min\{d_i,t_i,p_i\} = \Omega(\epsilon^2 \Delta / \eta^4)$.  We obtain the bound $p = \exp(-\epsilon^2 \Delta / \log^{4+o(1)} \Delta) = \exp(-\omega(\log \Delta))$ required in the statement of Theorem~\ref{thm:randAlg-main} by selecting $\eta = \log^{1 + o(1)}\Delta$.

%% file: estimate.tex
\subsection{Proof of Lemma~\ref{lem:estimate}}\label{sect:estimate}
{
\allowdisplaybreaks

In this section, we prove Lemma~\ref{lem:estimate}.
We assume $\min\{d_{i-1},t_{i-1},p_{i-1}\}=\omega(\log \Delta)$,  $\beta_{i-1} = o(1/ \log \Delta)$, and $\delta_{i-1} = o(\beta_{i-1}/ \log \Delta)$.
The two terms $(1 - 1/p_{i-1})^{2t_{i-1}}$ and $\frac{t_{i-1}}{p_{i-1}}(1 - 1/p_{i-1})^{2t_{i-1}}$  show up in the definition of $d^\diamond_{i-1}$, $t^\diamond_{i-1}$, and $p^\diamond_{i-1}$. We begin by showing that these two terms are both $e^{-2} (1 + o(1/\log \Delta))$. We use the fact that $\frac{t_{i-1}}{p_{i-1}} = \frac{1}{\beta_{i-1} + 1}$ in the following calculation.

\begin{align*}
(1 - 1/p_{i-1})^{2t_{i-1}}
&= e^{-2t_{i-1}/p_{i-1}} (1 - O(t_{i-1}/p_{i-1}^2))							& \mbox{(Taylor expansion of $e^x$)}\\
&= e^{-2} \cdot e^{2(1 - t_{i-1}/p_{i-1})} (1 - O(t_{i-1}/p_{i-1}^2))\\
&= e^{-2} \cdot e^{2(1 - t_{i-1}/p_{i-1})} \left(1 - O\left(\frac{1}{(1+\beta_{i-1})p_{i-1}}\right)\right)	& \mbox{(Defn.~$\beta_{i-1}$)}\\
&= e^{-2} \cdot e^{2(1 - t_{i-1}/p_{i-1})}(1 - o(1/\log \Delta))								& \mbox{($p_{i-1} = \omega(\log\Delta)$)}\\
&= e^{-2} \cdot e^{2\beta_{i-1} / (\beta_{i-1}+1)} (1 - o(1/\log \Delta))  & \mbox{(Defn.~of $\beta_{i-1}$)}\\
&= e^{-2} \cdot (1 + O({2\beta_{i-1} / (\beta_{i-1}+1)})) (1 - o(1/\log \Delta))\\
&= e^{-2} \cdot (1 + o(1/\log \Delta)) (1 - o(1/\log \Delta))\\
&= e^{-2} (1 + o(1/\log \Delta)). & (*)\\
\displaybreak[0]\\
\frac{t_{i-1}}{p_{i-1}}(1 - 1/p_{i-1})^{2t_{i-1}}
&= e^{-2} \cdot \frac{t_{i-1}}{p_{i-1}} \cdot (1 + o(1/\log \Delta)) & \text{by (*)}\\
&= e^{-2} (1 + o(1/\log \Delta)) / (1+ \beta_{i-1})\\
&= e^{-2} (1 + o(1/\log \Delta)) / (1 + o(1/\log \Delta))\\
&= e^{-2} (1 \pm o(1/\log \Delta)).	& (**)\\
\end{align*}

We are in a position to derive the first three equations in Lemma~\ref{lem:estimate} (i.e., estimates of $d_i$, $t_i$, and $p_i$).
Recall that $\delta_{i-1} =  o(1/\log^2 \Delta)$ and $1/p_{i-1} =o(1/\log \Delta)$.
\begin{align*}
d_i &=  d_{i-1} \cdot  (1 + \delta_{i-1}) \left(1 - (1 - 1/p_{i-1})^{2(t_{i-1} - 1)}\right) \\
&= d_{i-1} \cdot (1 + o(1/\log^2 \Delta)) \left( 1 - e^{-2} (1 + o(1/\log \Delta)) / (1 - 1/p_{i-1})^2\right) & \mbox{By (*)}\\
&= d_{i-1} \cdot (1 + o(1/\log^2 \Delta)) \left( 1 - e^{-2} (1 + o(1/\log \Delta))\right) \\
&= d_{i-1} \cdot (1 \pm o(1/\log \Delta)) (1 - e^{-2}). \\
\displaybreak[0]\\
t_i &=  t_{i-1} \cdot  (1+\delta_{i-1})\left(1 - \frac{t_{i-1}}{p_{i-1}}(1 - 1/p_{i-1})^{2t_{i-1}}\right)\left(1 - (1 - 1/p_{i-1})^{2t_{i-1}}\right)\\
&= t_{i-1} \cdot (1 + o(1/\log^2 \Delta)) \left(1 - e^{-2} (1 \pm o(1/\log \Delta))\right)^2			& \mbox{By (**)}\\
&= t_{i-1} \cdot (1 \pm o(1/\log \Delta))  (1 - e^{-2})^2. \\
\displaybreak[0]\\
p_i &=  p_{i-1} \cdot  (1-\delta_{i-1}) \left(1 - \frac{t_{i-1}}{p_{i-1}}(1 - 1/p_{i-1})^{2t_{i-1}}\right)^2 \\
&= p_{i-1} \cdot (1 - o(1/\log^2 \Delta)) \left(1 - e^{-2} (1 \pm o(1/\log \Delta))\right)^2			& \mbox{By (**)}\\
&= p_{i-1} \cdot (1 \pm o(1/\log \Delta))  (1 - e^{-2})^2.
\end{align*}

Finally, we derive the last equation in Lemma~\ref{lem:estimate}: an estimate of the second-order error $\beta_i$.
\begin{align*}
\beta_i &= \frac{p_i}{t_i} - 1\\
&= \frac{(1-\delta_{i-1}) p_{i-1}^\diamond}{(1+\delta_{i-1}) t_{i-1}^\diamond} - 1\\
&= (1 - O(\delta_{i-1})) \cdot \frac{p_{i-1}}{t_{i-1}} \cdot \frac{1 - \frac{t_{i-1}}{p_{i-1}}(1 - 1/p_{i-1})^{2t_{i-1}}}{1 - (1 - 1/p_{i-1})^{2t_{i-1}}} - 1 & \text{Definition of $p_{i-1}^\diamond$ and $t_{i-1}^\diamond$}\\
&= (1 - O(\delta_{i-1})) \cdot  \frac{\frac{p_{i-1}}{t_{i-1}} - (1 - 1/p_{i-1})^{2t_{i-1}}}{1 - (1 - 1/p_{i-1})^{2t_{i-1}}} - 1\\
&=  \frac{
\left(\frac{p_{i-1}}{t_{i-1}} - 1\right) +
O(\delta_{i-1})\left(- \frac{p_{i-1}}{t_{i-1}} + (1 - 1/p_{i-1})^{2t_{i-1}}\right)}
{1 -  (1 - 1/p_{i-1})^{2t_{i-1}}} \\
&=  \frac{
\left(\frac{p_{i-1}}{t_{i-1}} - 1\right) +
O(\delta_{i-1})\left(- \frac{p_{i-1}}{t_{i-1}} + (1 - 1/p_{i-1})^{2t_{i-1}}\right)}
{1 - e^{-2}(1 + o(1/\log \Delta))} & \text{by (*)}\\
&= \frac{\beta_{i-1} - O(\delta_{i-1})}{(1 - e^{-2})(1 - o(1/\log \Delta))} &-\frac{p_{i-1}}{t_{i-1}} + (1 - 1/p_{i-1})^{2t_{i-1}} = -\Theta(1)\\
&= \frac{\beta_{i-1} (1 - o(1/\log^2 \Delta))}{(1 - e^{-2})(1 - o(1/\log \Delta))} &\delta_{i-1} =  o(1/\log^2 \Delta)\\
&=  \beta_{i-1} \cdot  (1 \pm o(1/\log \Delta))  / (1 - e^{-2}).
\end{align*}
}

%% file: treeLLL.tex
\section{Distributed Lov\'{a}sz Local Lemma on Trees}\label{sect:tree-LLL}

\paragraph{Tree-structured Dependency Graphs.} In this section, we study the distributed LLL on  tree-structured dependency graphs, which we define as follows.
Let $T$ be a tree.  Each vertex $v$ holds some variables $\variable(v)$ and is associated with a bad
event $E(v)$ that depends only on variables
within distance $r/2$ of $v$; that is, $\vbl(E(v)) = \bigcup_{u\in N^{r/2}(v)} \variable(u)$.
If $S$ is a subset of the vertices, we use $\vbl(S)$ to be short for $\bigcup_{v\in S} \vbl(E(v)) = \bigcup_{v\in S}\bigcup_{u\in N^{r/2}(v)} \variable(u)$.
 We assume that $r$ is a constant, and we do not analyze the dependence on $r$ in the time complexity.

The dependency graph for the set of bad events $\mathcal{E}$  is exactly
$T^r$, which is the graph obtained by adding edges to all pairs of vertices of distance at most $r$ in $T$.
Thus, the maximum degree of the dependency graph is at most $\Delta^r$, where $\Delta$ is the maximum degree of $T$.
We fix the parameter $d = \Delta^r$.

The tree-structured dependency graphs (with parameter $r$)
arise naturally from any $r/2$-time $\RandLOCAL$ experiment that is run on a tree $T$.
Throughout this section we assume $r/2 \ge 1$ is an integer and that $\Delta \ge 3$.

\subsection{Deterministic LLL Algorithm}

\paragraph{Network Decomposition.} A {\em $(\lambda,\gamma)$-network decomposition} is a partition of the vertex set into $V_1, \ldots, V_{\lambda}$
such that each connected component induced by each $V_i$ has diameter at most $\gamma$.  Fischer and Ghaffari~\cite{FischerG17} showed
that given a $(\lambda,\gamma)$-decomposition of $G_{\mathcal{E}}^2$,
an LLL instance satisfying $p(ed)^\lambda < 1$ is solvable in $O(\lambda(\gamma+1))$ time.
We use a slight generalization of standard network decompositions.
A $(\lambda_1,\gamma_1,\lambda_2,\gamma_2)$-network decomposition is a partition of the vertices into
$V_1, \ldots, V_{\lambda_1},U_1, \ldots, U_{\lambda_2}$ such that connected components
induced by $V_i$ have diameter at most $\gamma_1$ and those induced by $U_i$ have diameter at most $\gamma_2$.

\paragraph{Strong and Weak Diameter.}
There are two standard notions of diameter in network decompositions. For a subgraph $H=(V',E')$ of $G$, its \emph{weak diameter} is $\max_{u,v\in V'} \dist_G(u,v)$, whereas its \emph{strong diameter} is $\max_{u,v\in V'} \dist_H(u,v)$. We remark that either notion is sufficient for applying Lemma~\ref{lem:LLL-aux-2}. Therefore, unless stated otherwise, we do not distinguish between strong and weak diameter in the subsequent discussion.

\begin{lemma}[Fischer and Ghaffari~\cite{FischerG17}]\label{lem:LLL-aux-2}
Suppose that a $(\lambda_1,\gamma_1,\lambda_2,\gamma_2)$-network decomposition of $G_{\mathcal{E}}^2$ is given.
Any LLL instance on $G_{\mathcal{E}}$ satisfying
$p(ed)^{\lambda_1+\lambda_2} < 1$ can be solved in $\DetLOCAL$ in
$O(\lambda_1(\gamma_1+1) + \lambda_2(\gamma_2+1))$ time.
\end{lemma}

The proof of Theorem~\ref{thm:detLLL-tree} is based on the network decompositions for trees found in
Section~\ref{sect:tree-decomp}.
A {\em distance-$d$ dominating set} of a graph $G$ is a vertex set $S$ such that for each vertex
$v$ in the graph $G$, there exists $u \in S$ such that $\dist(u,v) \leq d$.

\begin{theorem}\label{thm:detLLL-tree}
Any tree-structured LLL satisfying
$p(ed)^{\lambda} < 1$ with $\lambda\ge 2$
can be solved in $\DetLOCAL$ in
$O\left(\max\left\{\log_\lambda s,\fr{\log s}{\log\log s}\right\} + \log^\ast n\right)$ time, where $s\le n$ is the
size of any distance-$O(1)$ dominating set of the tree $T$.
\end{theorem}

\begin{proof}
Recall that the dependency graph is $T^r$ for some tree $T$ and constant $r$.
In Section~\ref{sect:tree-decomp} we show that a standard $(2,O(\log s))$-decomposition for $(T^{r})^2 = T^{2r}$
is computable in $O(\log s + \log^\ast n)$ time, and if $\lambda=\Omega(1)$ is sufficiently large, a
$(1,O(\log_\lambda s),O(\lambda^2),1)$-decomposition for $T^{2r}$
is computable in $O(\log_\lambda s+\log^\ast n)$ time, i.e., one
part of the partition has diameter $O(\log_\lambda s)$, while each
of the remaining $O(\lambda^2)$ parts induces connected components
of diameter at most~$1$ in $T^{2r}$.

If we want to use Lemma~\ref{lem:LLL-aux-2} to solve the given LLL instance satisfying
$p(ed)^{\lambda} < 1$, we need a $(\lambda_1,\gamma_1,\lambda_2,\gamma_2)$-network decomposition of $T^{2r}$ satisfying $\lambda_1+\lambda_2 \leq \lambda$, i.e., the number of parts is at most $\lambda$.

When $\lambda=O(1)$ is sufficiently small, we apply Lemma~\ref{lem:LLL-aux-2} with the first network decomposition.
Because the decomposition has two parts, this works with LLL criterion $p(ed)^\lambda < 1$ for any $\lambda \geq 2$. The resulting LLL algorithm takes time $O(\log s + \log^\ast n)$.

When $\lambda$ is sufficiently large, we compute a $(1,O(\log_{\hat{\lambda}} s), O(\hat{\lambda}^2), 1)$-decomposition in $O(\log_{\hat{\lambda}} s + \log^\ast n)$
time, where $\hat{\lambda}$ is chosen as the largest number such that $\hat{\lambda} \leq \sqrt{\fr{\log s}{\log\log s}}$ and the number of parts $\lambda_1 + \lambda_2 = O(\hat{\lambda}^2)$ in the decomposition is at most $\lambda$. We have
$\hat{\lambda} = \min\left\{ O(\sqrt{\lambda}), \, \sqrt{\fr{\log s}{\log\log s}}\right\}$.
We solve the LLL by applying Lemma~\ref{lem:LLL-aux-2},
which takes time $O(\hat{\lambda}^2+\log_{\hat{\lambda}} s + \log^\ast n) = O\left(\max\left\{\log_\lambda s, \fr{\log s}{\log\log s}\right\} + \log^\ast n\right)$.
Observe that because of the $\hat{\lambda}^2$ term, we cannot benefit from LLL instances with $\lambda \gg \fr{\log s}{\log\log s}$.
\end{proof}

Notice that the time bound for Theorem~\ref{thm:detLLL-tree} is in terms of $s$ rather than $n$.
We will apply Theorem~\ref{thm:detLLL-tree} after performing a graph shattering step, the output
of which creates many disjoint tree-structured instances with size $\Delta^{O(1)} \cdot O(\log n)$,  each of them admitting a
 distance-$O(1)$ dominating set of size at most $s = O(\log n)$.  We want the time bound to be
in terms of $s= O(\log n)$, independent of $\Delta$.

For a given LLL instance with criterion $p(ed)^{\lambda}<1$, the shattering routine of Fischer and Ghaffari~\cite{FischerG17} achieves the above requirement in time $O(d^2 + \log^* n)$ in such a way that the resulting LLL instances after the shattering routine satisfy the criterion $p(ed)^{\lambda/2}<1$.
If we combine this with
Theorem~\ref{thm:detLLL-tree}, we obtain a $O\left(d^2 + \max\left\{\log_\lambda \log n, \fr{\log\log n}{\log\log\log n}\right\}\right)$-time $\RandLOCAL$ LLL algorithm for
criterion $p(ed)^{\lambda}<1$, $\lambda \ge 4$, which is efficient only when $d$ is small. Notice that we need $\lambda/2 \geq 2$ in order to apply Theorem~\ref{thm:detLLL-tree} on LLL instances with criterion $p(ed)^{\lambda/2}<1$. 

In Section~\ref{sect:shatterLLL} we present a new method (Lemma~\ref{lem:LLL-shattering}) for computing a partial assignment to the variables that effectively
shatters a large dependency graph into many independent subproblems, each satisfying a polynomial LLL criterion
w.r.t.~the unassigned variables.

\subsection{Randomized LLL Algorithm}\label{sect:shatterLLL}

Consider a tree-structured LLL instance $T^r$ with
LLL criterion $p(ed)^\lambda < 1$. In subsequent discussion, unless otherwise stated, the underlying graph is, by default, assumed to be $T$.  Our shattering routine will work towards finding a {\em good} partial assignment.

\begin{definition}\label{def:good}
A partial assignment $\phi$ to the variables in the LLL system is {\em good} if it satisfies the following two properties.
\begin{enumerate}
\item Conditioned on the partial assignment $\phi$, the probability of any
bad event $E(v)$ is at most $p' = \sqrt{p}$.
\item Let $V'$ be the set of all vertices $v$ such that $\vbl(E(v))$
contains some unassigned variables.  Each connected component
$C$ induced by $V'$  has size at most $\Delta^{O(1)} \cdot O(\log n)$,
and $C$ contains a distance-$2r$ dominating set with size at most $O(\log n)$.
\end{enumerate}
\end{definition}

Due to Definition~\ref{def:good}(1), conditioned on a good partial assignment
$\phi$, the bad events in each connected component $C$ induced by $V'$ form an
LLL system with the LLL criterion $p'(ed)^{\lambda/2} < 1$. 
Definition~\ref{def:good}(2) guarantees that each component is of small size.
Thus, a good partial
assignment $\phi$ is able to shatter the tree $T$ into small components, each of which is an independent LLL system.  
In Sections~\ref{sect:criterion}--\ref{sect:findset} we prove the following efficient ``shattering lemma.''

\begin{lemma}\label{lem:LLL-shattering}
Suppose we are given a tree-structured LLL instance $T^r$ satisfying
LLL criterion $p(ed)^\lambda < 1$, where $\lambda \ge 2(4^{r}+ 8r)$.
There is a $\RandLOCAL$ algorithm that computes a good partial assignment $\phi$ in
$O(\log_\lambda\log n)$ time.
\end{lemma}

The overall algorithm is obtained by composing~Lemma~\ref{lem:LLL-shattering}
and Theorem~\ref{thm:detLLL-tree}, which is summarized in Theorem~\ref{thm:randLLL-tree}.  In particular, the algorithm has the usual two-phase graph shattering structure.
\begin{description}
\item[Shattering.] Given the LLL instance with dependency graph
$T^r$, find a \emph{good} partial assignment $\phi'$ using Lemma~\ref{lem:LLL-shattering}.  
Each component induced by events having at least one unset variable 
has size $\poly(\Delta)\cdot O(\log n)$ and contains a distance-$2r$
dominating set with size $O(\log n)$.  Moreover, each such component 
is an LLL instance with parameters $d$ and $p'=\sqrt{p}$ satisfying
criterion $p'(ed)^{\lambda/2}<1$.
\item[Post-shattering.] We extend $\phi'$ to a total assignment 
by independently 
fixing the variables in each component of the shattered LLL instance.
By Theorem~\ref{thm:detLLL-tree}, this can be done in 
$O\left(\max\left\{\log_{\lambda/2} s, \frac{\log s}{\log\log s}\right\}\right)$
time, where in our case $s=O(\log n)$.
\end{description}

\begin{theorem}\label{thm:randLLL-tree}
Let $T^r$ be a tree-structured LLL instance satisfying criterion $p(ed)^\lambda < 1$ with $\lambda \ge 2(4^{r}+ 8r)$.
This LLL can be solved in $\RandLOCAL$ in
$O\left(\max\left\{\log_\lambda \log n,\, \fr{\log \log n}{\log\log\log n}\right\}\right)$ time.
\end{theorem}

We briefly overview the ideas behind the proof of  Lemma~\ref{lem:LLL-shattering}.
The goal is to design an algorithm to compute a good partial
assignment $\phi$.
Consider the following process.
First, draw a total assignment $\phi$ to $\mathcal{V}$ according to the distribution of the variables in the underlying LLL instance.
If any bad event $E(v)$ occurs under $\phi$, update $\phi$ by \emph{unsetting} all variables in $\vbl(E(v))$.
More generally, whenever $\Pr[E(v) | \phi]$ exceeds $\sqrt{p}$, update $\phi$ by unsetting all variables in $\vbl(E(v))$.
This can be viewed as a \emph{contagion dynamic} played out on the dependency graph.  Bad events that occur under
the initial total assignment are {\em infected}, and infected vertices can cause nearby neighbors to become infected. At the end of the  contagion process, we obtain a partial assignment satisfying Definition~\ref{def:good}(1).

If this contagion process were actually simulated, it would take $\Omega(\log n)$ parallel steps to reach a stable state, which is too slow.
We will provide a different method to compute a stable state (i.e., a partial assignment satisfying Definition~\ref{def:good}(1)) that is exponentially faster, by avoiding a direct simulation.

The proof of Lemma~\ref{lem:LLL-shattering} appears at the end of Section~\ref{sect:contagion}.  It uses Lemma~\ref{thm:LLL-contagion},
which concerns the problem of finding a stable state in a contagion
process, and Lemma~\ref{lem:LLL-reduction}, which connects the problem 
of shattering a dependency graph $T^r$ to a contagion played out on $T^r$.

\subsection{Criterion for Infection}\label{sect:criterion}

Let $u$ be a vertex in the undirected tree $T$.
Then $T - \{u\}$ consists of $\deg(u)$ subtrees $T_1, \ldots, T_{\deg(u)}$; we call $T_k$ the \emph{$k$th subtree} of $u$.
Define $C_u(k,[i,j])$ to be the set of vertices in the $k$th subtree of $u$ whose distance to $u$ lies in the interval $[i,j]$.
For example, $C_u(k,[1,1])$ only contains the $k$th neighbor of $u$.
For any vertex set $S$, define $\degB_S(u)$ as follows,
\[
\degB_S(u) = \left|
\left\{ k \;\,:\,\;
C_u(k,[1,r])\cap S\neq\emptyset
\right\}
\right|.
\]
In other words, it is the number of \emph{distinct} subtrees
of $u$ containing at least one $S$-vertex within distance $r$.

Let $\mu\ge 4$ and $\lambda'\ge 1$ be two integers such that $\lambda \ge 2(\mu^{r}+ \lambda')$.
The following bad events $B(S,v)$ and $B(v)$ are defined w.r.t.~the following process.
First, we fix a total assignment $\phi$ to the variables,
then progressively add vertices to the set $S$.  All variables in $\vbl(S)$ are considered \emph{unset};
for example, conditioning on ``$\vbl(E(v)) \backslash \vbl(S)$'' means keeping $\phi$'s assignment
to $\vbl(E(v)) \backslash \vbl(S)$ and \emph{resampling} $\vbl(S)$ according to their distribution in the underlying LLL instance.

\begin{align*}
B(S,v) &: \Big[ \Pr\left[E(v)~|~\vbl(E(v)) \backslash \vbl(S)\right]\geq(ed)^{-\lambda/2}\Big],\\
B(v) &: \Bigg[\bigcup_{S\subset N^{r}(v), \; |S|\leq \mu^{r}} B(S,v)\Bigg].
\end{align*}

In other words, $B(S,v)$ is the event that, if we \emph{were} to resample $\vbl(S)$,
the probability that $E(v)$ occurs is at least $(ed)^{-\lambda/2}$.  The event $B(v)$ occurs if it is \emph{possible}
to find a subset $S$ of  cardinality at most $\mu^{r}$ such that $B(S,v)$ occurs.

We can now consider the probability that these events occur, over a \emph{randomly} selected initial total assignment $\phi$.
\begin{align*}
\Pr_\phi[B(S,v)] &\leq \frac{\Pr_\phi[E(v)]}{\Pr_\phi[E(v)~|~B(S,v)]}
	\:\leq\: \frac{(ed)^{-\lambda}}{(ed)^{-\lambda/2}}
	\:=\: (ed)^{-\lambda/2}
	\:\leq\: (ed)^{-({\mu}^{r} + \lambda')}.\\
\intertext{By a union bound over the $|N^{r}(v)|^{\mu^{r}} \leq d^{\mu^{r}}$ choices of $S$ (recall that $d=\Delta^r$),}
\Pr_\phi[B(v)] &\leq \sum_S \Pr_\phi[B(S,v)] < (ed)^{-\lambda'}.
\end{align*}

Intuitively, $B(v)$ is the event that $E(v)$ is \emph{too close} to happening.
That is, relatively few variables need to be resampled to give $E(v)$ a likely probability of happening.
Lemma~\ref{lem:LLL-condition} shows that the criterion for infection
``$\degB_S(v) > \mu$'' is a good proxy for the harder-to-analyze criterion ``$E(v)$ is too close to happening''.

\begin{lemma}\label{lem:LLL-condition}
Fix a total variable assignment $\phi$.
Let $S$ be any vertex set such that, for each vertex $v$, if  $B(v)$ occurs under $\phi$ or  $\degB_{S}(v) > \mu$, then $v$ must be in $S$.
Then $\Pr[E(v)~|~\vbl(E(v)) \setminus \vbl(S)] < (ed)^{-\lambda/2}$ for each vertex $v$.
\end{lemma}

\begin{proof}
If $v \in S$, then the probability of seeing $E(v)$ after resampling
$\vbl(S)$ is, according to the original LLL criterion, at most
$p < (ed)^{-\lambda}$.  In what follows we assume $v\not\in S$.

To prove the lemma, it suffices to show that there exists a vertex
set $S'$ such that (i) $S'\subset N^{r}(v)$, (ii) $|S'|\leq \mu^{r}$, and (iii) $\vbl(S')\cap \vbl(E(v)) = \vbl(S) \cap \vbl(E(v))$. Notice that
(iii) implies that resampling $\vbl(S')$ is equivalent to resampling $\vbl(S)$ from $v$'s point of view.
Since $v\not\in S$, by assumption, event $B(v)$ does not occur.
Since $|S'|\leq \mu^{r}$, event $B(S', v)$ does not occur.
Hence $\Pr[E(v)~|~\vbl(E(v)) \setminus \vbl(S')] < (ed)^{-\lambda/2}$, as desired.

Root the tree at $v$.  We call a vertex $u\in S$ ``highest'' if $u$ is in $N^{r}(v)$ and no ancestor of $u$ is in $S$.
Observe that if $H$ is the set of highest vertices, then $\vbl(S) \cap \vbl(E(v)) = \vbl(H) \cap \vbl(E(v))$.
To see this, observe that if $u' \in S$ is not highest,
and is a descendant of some highest $u\in S$,
that $\vbl(E(u')) \cap \vbl(E(v))$ is contained in $\vbl(E(u))\cap \vbl(E(v))$.

Thus, we only need to bound $|H|$ by $\mu^{r}$.
Suppose, for the sake of contradiction, that $|H| \ge \mu^{r}+1$.  Define the path $(v=v_0,v_1,\ldots,v_{r})$ by selecting
$v_i$ as the child of $v_{i-1}$ that maximizes the  number of vertices in $H$ contained in the subtree rooted at $v_i$.  We prove by induction that the subtree rooted at $v_i$ contains
at least $\mu^{r-i}+1$ $H$-vertices. The base case $i=0$ holds by assumption.  If there are $\mu+1$ subtrees of $v_i$ containing
$H$-vertices, then $v_i$ would be infected. Thus, by the pigeonhole principle, the number of $H$-vertices in the subtree rooted at
$v_{i+1}$ must be at least $\ceil{(\mu^{r-i}+1)/\mu} = \mu^{r-(i+1)}+1$. Hence the subtree rooted at $v_r$ contains $\mu^0+1=2$ $H$-vertices; this is a contradiction since the only vertex in this subtree eligible to be in $H$ is $v_{r}$ itself.
\end{proof}

\subsection{Contagion Process}\label{sect:contagion}

A {\em $(q_0,r,\mu)$-contagion process} on an $n$-vertex tree $T$ is played out as follows.
Initially, each vertex is infected with probability $q_0$, and these events are independent
for vertices at distance greater than $r$.
If $S$ is the set of infected vertices at some time and $\widehat{\deg}_S(v) > \mu$, then $v$ becomes infected.
In this section our goal is, given the initially infected vertices,
to compute a superset of those vertices that is \emph{stable} and \emph{small}, defined as follows.

\begin{definition}\label{def:smallstable}
Let $S_0$ be the initially infected vertices and $S\supset S_0$.
\begin{itemize}
\item $S$ is called \emph{stable} if it causes no more infection.
\item $S$ is called \emph{small} if each connected component induced by $\bigcup_{v \in S} N^{r}(v)$
contains a distance-$2r$ dominating set of size at most $O(\log n)$.
\end{itemize}
\end{definition}

In Lemma~\ref{thm:LLL-contagion}, we show that one can efficiently compute a set $S$ that is both \emph{stable} and \emph{small}.

\begin{lemma}\label{thm:LLL-contagion}
Consider a $(q_0,r,\mu)$-contagion process played on an $n$-vertex tree $T$ with maximum degree $\Delta$.
There is a $\RandLOCAL$ algorithm that computes a {small stable set} $S$ in $O(\log_\mu\log n)$ time,
where $r$ is constant, $q_0 \leq (ed)^{-8r}$, $d=\Delta^r$, and $\mu\geq 4$.
\end{lemma}

The proof of Lemma~\ref{thm:LLL-contagion} is deferred to Section~\ref{sect:findset}.
Lemma~\ref{lem:LLL-reduction} connects the contagion problem to finding a good partial assignment.

\begin{lemma}\label{lem:LLL-reduction}
Suppose there is a $\tau$-round $\RandLOCAL$ algorithm for finding a small stable set $S$ for
a $((ed)^{-\lambda'},r,\mu)$-contagion process.
Then there exists a $(\tau+O(1))$-round $\RandLOCAL$ algorithm for finding
a good partial assignment $\phi$ to a tree-structured LLL instance with criterion
$p(ed)^\lambda < 1$, where $\lambda\geq 2(\mu^{r}+\lambda')$.
\end{lemma}

\begin{proof}
Let $q_0 = (ed)^{-\lambda'}$. Consider the $(q_0,r,\mu)$-contagion process defined by choosing a random assignment $\phi'$ to the variables in the LLL system and initially infecting all vertices $v$ such that $B(v)$ occurs. The lower bound on $\lambda$ implies $\Pr[B(v)] \le q_0 = (ed)^{-\lambda'}$.
Given the small stable set $S$, we let $\phi$ be the result of unassigning all variables in $\vbl(S) = \bigcup_{v\in S} \vbl(E(v)) = \bigcup_{v\in S}\bigcup_{u\in N^{r/2}(v)} \variable(u)$.

We now verify that $\phi$ is a good partial assignment.
Since $S$ is stable, for each vertex $v$, if  $B(v)$ occurs under $\phi$ or  $\degB_{S}(v) > \mu$, then $v$ must be in $S$.
By Lemma~\ref{lem:LLL-condition}, $\Pr[E(v)~|~\vbl(E(v)) \backslash \vbl(S)] < (ed)^{-\lambda/2} < \sqrt{p}$ for each vertex $v$,
and so Definition~\ref{def:good}(1) is satisfied.
Let $V' = \bigcup_{v \in S} N^r(v)$ be the set of all vertices $v$ such that $\vbl(E(v))$ contains some unassigned variables.
Since $S$ is small, each connected component $C$ induced by $V'$ contains a distance-$2r$ dominating set with size at most $O(\log n)$. Since $2r = O(1)$, the cardinality of $C$ is at most $\poly(\Delta)\cdot O(\log n)$. Hence Definition~\ref{def:good}(2) is also satisfied.
\end{proof}

We are now in a position to prove Lemma~\ref{lem:LLL-shattering}.

\begin{proof}
Recall that the LLL criterion of in Lemma~\ref{lem:LLL-shattering} is $\lambda\geq 2(4^{r}+ 8r)$.
We pick the largest \emph{even} integer $\mu$ such that $\lambda\geq 2(\mu^{r}+ 8r)$, and we set $\lambda' = 8r$.
Notice that $\mu\geq 4$ and $\log\mu = \Theta(\log\lambda)$.
By Lemma~\ref{thm:LLL-contagion}, a small stable set $S$ for the $((ed)^{-8r},r,\mu)$-contagion process can be computed in $O(\log_\mu\log n)=O(\log_\lambda\log n)$ time.
By Lemma~\ref{lem:LLL-reduction}, this implies a $O(\log_\lambda\log n)$-time $\RandLOCAL$ algorithm to finding
a good partial assignment $\phi$ under the LLL criterion $p(ed)^\lambda < 1$.
\end{proof}

\subsection{Finding a Small Stable Set}\label{sect:findset}

We prove Lemma~\ref{thm:LLL-contagion} in  this section.
The algorithm for Lemma~\ref{thm:LLL-contagion} simulates a more virulent contagion process for $\tau$ steps using
threshold $\mu/2$ rather than $\mu$, then simulates a reverse-contagion for $\tau$ steps,
where vertices become \emph{uninfected} if they were not initially infected and they
have nearby infected vertices in at most $\mu$ subtrees.
We prove that when $\tau=\Theta(\log_\mu \log n)$, the final infected set $S = L_\tau$ is both stable and small.
This process is called {\sf Find-Small-Stable-Set}.  The sets generated by this process satisfy that
$U_0 \subseteq \cdots \subseteq U_\tau = L_0 \supseteq \cdots \supseteq L_\tau$.

\medskip

\centerline{
\framebox{\parbox{5in}{
\noindent {\sf Find-Small-Stable-Set}.
\begin{description}
\item (1) $U_0 \leftarrow \{u\in V~|~\text{$u$ is initially infected}\}$. That is, $u\in U_0$ if $B(u)$ occurs initially.
\item (2) For $1\leq i\leq \tau$, do $U_i \leftarrow U_{i-1}\cup\{u\in V~|~\degB_{U_{i-1}}(u)>\mu/2\}$.
\item (3) $L_0 \leftarrow U_{\tau}$.
\item (4) For $1\leq i\leq \tau$, do $L_i \leftarrow L_{i-1} \setminus \{u\in L_{i-1} \setminus U_0~|~\degB_{L_{i-1}}(u)\leq\mu\}$.
\item (5) Return $L_{\tau}$.
\end{description}
}}}

\medskip

We show that $S = L_\tau$ is stable in Lemma~\ref{lem:stable}.
Let $L_{\tau+1}$ be the set of all vertices $u$ such that $\degB_{L_{\tau}}(u)> \mu$.
Our goal is to show that if $u \notin  L_\tau$, then $\degB_{L_\tau}(u)\leq\mu$ (i.e., $u \notin L_{\tau+1}$) with high probability.

Root $T$ at an arbitrary vertex, and let $T'$ refer to the rooted version.
Define $T'_u$ to be the subtree of $T'$ rooted at $u$,
and define $C'_u(k,[i,j])$ as $C_u(k,[i,j])\cap T'_u$.
Given a vertex set $W$,
define $\degA{W}(u)$ as the number of different
$k$ such that $C'_u(k,[1,r])\cap W\neq\emptyset$.
Although the original contagion process is played on $T$, it is easier to analyze a similar process
played on $T'$, where only descendants can cause a vertex to become infected.

In general, if $\{X(u)\}_{u \in V}$ is an ensemble of events associated with vertices and $W$ a subset of vertices,
we write $X(W)$ to denote the event $\bigcup_{u\in W} X(u)$, i.e., there exists $u \in W$ such that $X(u)$ occurs.
We write $X$ to denote the set of vertices $\{u \in V ~|~ X(u) \mbox{ occurs}\}$. For any two events $A$ and $B$, we write $A \Rightarrow B$ to denote $A \subseteq B$, i.e., $A$ implies $B$. With respect to a vertex $u$, consider the following three sequences of events.
\begin{align*}
&(\eventB_i(u)): &
&\text{for each $0\leq i\leq \tau$, let $\eventB_i(u)$ be $(u\notin U_{i})\wedge(u\in L_{i+1})$.}\\
&(\eventC_i(u)): &
&\text{let $\eventC_0(u)$ be $(u\in U_0)$; for each $0\leq i < \tau$, let $\eventC_{i+1}(u)$ be $\eventC_0(u)\lor(\degA{\eventC_{i}}(u)\geq\mu/2)$.}\\
&(\eventA_i(u)): &
&\text{let $\eventA_0(u)$ be $~\eventC_\tau(u)$; for each $0\leq i <  \tau$, let $\eventA_{i+1}(u)$ be $\degA{\eventA_{i}}(u)\geq\mu/2$.}
\end{align*}

\begin{lemma}\label{lem:UtauLtau}
No vertex can belong to both $U_{\tau} \setminus L_\tau$ and $L_{\tau+1}$.
\end{lemma}

\begin{proof}
Suppose there were such a vertex $u$.   If $u\in L_{\tau+1}$ then it must have more than $\mu$ neighbors
in $L_{\tau}$, which were also in $L_{\tau-1} \subseteq \cdots \subseteq L_0 = U_\tau$.
But if $u\in U_{\tau}$ then it would also remain in $L_0,\ldots,L_\tau$,
contradicting the assumption that $u\in U_\tau\setminus L_\tau$.
\end{proof}

By Lemma~\ref{lem:UtauLtau}, to prove that $S=L_\tau$ is stable, it suffices to prove that
\[
\Prob[\eventB_\tau(u)]= \Prob[(u\notin L_{\tau})\wedge(u\in L_{\tau+1})] = 1/\poly(n).
\]

Lemma~\ref{lem:obs} connects the true contagion process on $T$ to an imagined one played on $T'$.

\begin{lemma}\label{lem:obs}
For each vertex $u$ in $T$, and for each $0 \leq i \leq \tau$, we have  $\eventB_i(u) \Rightarrow \eventA_i(u)$.
\end{lemma}
\begin{proof}
We first show that $(u\in U_i) \Rightarrow \eventC_i(u)$, for each $0 \leq i \leq \tau$.
The base case ($i=0$) follows from the definition of $\eventC_0(u)$.
Assume by inductive hypothesis that $(u\in U_{i-1})\Rightarrow \eventC_{i-1}(u)$. We have:
\[
\left(u\in U_i \setminus U_0\right) \Rightarrow
\left(\degB_{U_{i-1}}(u)>\mu/2\right)  \Rightarrow
\left(\degA{U_{i-1}}(u)\geq\mu/2\right)  \Rightarrow
\left(\degA{\eventC_{i-1}}(u)\geq\mu/2\right).
\]
This implies $(u\in U_i)\Rightarrow \eventC_i(u)$, since  $(u\in U_0) \Rightarrow \eventC_0(u) \Rightarrow \eventC_i(u)$.

Next, we prove by induction that $\eventB_i(u) \Rightarrow \eventA_i(u)$, for each $0 \leq i \leq \tau$.
The base case $i=0$ follows from the above result:
\[
\eventB_0(u)\Rightarrow (u\in L_{1})\Rightarrow (u\in L_0 = U_\tau)\Rightarrow \eventC_\tau(u)\Rightarrow \eventA_0(u).
\]
Assume inductively that $\eventB_{i-1}(u)\Rightarrow \eventA_{i-1}(u)$.
Let $u$ be any vertex in $L_{i+1} \setminus U_i$, i.e., the event $\eventB_{i}(u)$ occurs.
Since $u \notin U_i \supseteq U_0$,
the only way {\sf Find-Small-Stable-Set} could put $u\in L_{i+1}\setminus U_i$ is if
\begin{align*}
\degB_{L_{i}}(u) &> \mu\\
\mbox{ and \ }
\degB_{U_{i-1}}(u) &\leq \mu/2,
\intertext{which implies}
\degB_{\eventB_{i-1}}(u) \,=\, \degB_{L_{i}}(u) - \degB_{U_{i-1}}(u) &> \mu / 2.
\intertext{and hence}
\degA{\eventB_{i-1}}(u) &\geq \mu/2.
\end{align*}
By inductive hypothesis, we have
\[
\left(\degA{\eventB_{i-1}}(u) \geq \mu/2\right) \Rightarrow \left(\degA{\eventA_{i-1}}(u)\geq\mu/2\right) \Rightarrow \eventA_i(u),
\]
which completes the induction.
\end{proof}

For brevity, define $p_i = \max_u\Pr[\eventA_i(u)]$ and $q_i=\max_u\Pr[\eventC_i(u)]$. We prove two auxiliary lemmas.

\begin{lemma}\label{lem:stable-aux-1}
$p_\tau\leq(\Delta^{2((r^2 / 2)+1)}p_0)^{(\frac{\mu}2)^{\tau/(r/2)}}$.
\end{lemma}
\begin{proof}
Suppose that $u$ is a vertex such that $\eventA_i(u)$ occurs.
Then, by definition of $\eventA_i(u)$, there exist $\mu/2$ different indices $k$ such that $\eventA_{i-1}(C'_u(k,[1,r]))$ occurs.
A consequence of this observation is that
\[
\eventA_{i-1}(C'_u(k,[1,r])) \Rightarrow \eventA_{i-2}(C'_u(k,[2,2r])) \Rightarrow \eventA_{i-3}(C'_u(k,[3,3r]))
\cdots
\Rightarrow \eventA_{i-(r/2)}(C'_u(k,[r/2,r^2/2])).
\]

Therefore, if $\eventA_i(u)$ occurs, there must exist $\mu/2$ indices $k$ such that $\eventA_{i-(r/2)}(C'_u(k,[r/2,r^2/2]))$ occurs.  The $\mu/2$ events $\{\eventA_{i-(r/2)}(C'_u(k,[r/2, r^2 / 2]))\}$ are independent, since
$\eventA_i(v)$  depends only on $\vbl(T'_v) = \bigcup_{w \in N^{r/2}(v) \cup T'_v} \variable(w)$.
This independence property is one reason why it is easier to analyze a contagion on $T'$ rather than $T$.

By a union bound over all vertices in $C'_u(k,[r/2, r^2 / 2])$, we have
\[
\Pr\left[\eventA_{i-(r/2)}(C'_u(k,[r/2,r^2 / 2]))\right]\leq \Delta^{r^2 / 2-1}p_{i-(r/2)}.
\]
Taking a union bound over at most $\binom{\Delta}{\mu/2}$ choices of $\mu/2$ distinct indices $k$, we infer that
\[
p_i\leq \Delta^{\mu/2}\left(\Delta^{r^2 / 2-1}p_{i-(r/2)}\right)^{\mu/2}
\leq
\left(\Delta^{(r^2 / 2)} p_{i-(r/2)}\right)^{\mu/2}
\]
for each $r/2 \leq i \leq \tau$.
Assume $\tau$ is a multiple of $r/2$, and recall $\mu/2 \geq 2$. We can bound $p_\tau$ as follows.
\[p_\tau
\leq p_0^{(\frac{\mu}2)^{\tau/(r/2)}} \cdot \prod_{j=1}^{\tau/(r/2)} \left(\Delta^{(r^2 / 2)}\right)^{(\frac{\mu}2)^{j}}
\leq \left(\Delta^{r^2}p_0\right)^{(\frac{\mu}2)^{\tau/(r/2)}}. \qedhere
\]
\end{proof}

\begin{lemma}\label{lem:stable-aux-2}
$p_0 = q_\tau \leq \Delta^{r/2}q_0$.
\end{lemma}

\begin{proof}
Recall that $\eventC_i(u)$ is $(u\in H_0) \vee (\degA{\eventC_{i-1}}(u)\geq\mu/2)$.
This implies that
\[
\eventC_{i-1}(C'_u(k,[1,r]))\Rightarrow H_0(C'_u(k,[1,r])) \lor \eventC_{i-2}(C'_u(k,[2,2r])).
\]
Repeating this $(r/2) - 1$ times, $\eventC_{i-1}(C'_u(k,[1,r]))$ implies that
\[
H_0(C'_u(k,[1, r(r/2 - 1)]) \lor \eventC_{i-(r/2)}(C'_u(k,[r/2, r^2 / 2])).
\]
Since  $H_0(C'_u(k,[1, r(r/2 - 1)]) \Rightarrow \eventC_{i-(r/2)}(C'_u(k,[r/2,  r^2 / 2]))$, we conclude that
\[
\eventC_{i-1}(C'_u(k,[1,r])) \Rightarrow  H_0(C'_u(k,[1,r/2 - 1]) \vee \eventC_{i-(r/2)}(C'_u(k,[r/2,r^2/2])).
\]
Thus, if $\eventC_i(u)$ occurs, then either (i)  $H_0(N^{r/2 - 1}(u))$ occurs, or (ii) there exist $\mu/2$  different indices $k$ such that $\eventC_{i-(r/2)}(C'_u(k, [r/2, r^2 / 2]))$ occurs. The events $\eventC_{i-(r/2)}(C'_u(k,[r/2, r^2 / 2]))$ for all $k$ are independent, since $\eventC_i(v)$  depends only on
$\vbl(T'_v) = \bigcup_{w \in N^{r/2}(v) \cup T'_v} \variable(w)$.

By a union bound, $\Pr[\eventC_{i-(r/2)}(C'_u(k, [r/2,r^2 / 2]))]\leq \Delta^{r^2/2-1}q_{i-r/2}$.
Suppose that $\tau$ is a multiple of $r/2$.
Taking a union bound over at most $\binom{\Delta}{\mu/2}$ choices of $\mu/2$ distinct indices $k$, we have
\begin{align*}
q_\tau &\leq
\Prob\left[ H_0(N^{r/2 - 1}(u))\right] + \binom{\Delta}{\mu/2} \cdot \Delta^{r^2/2-1}q_{\tau-(r/2)}\\
&\leq
\Delta^{r/2 - 1}q_0+\Delta^{\mu/2}\left(\Delta^{r^2/2-1}q_{\tau-(r/2)}\right)^{\mu/2}\\
&\leq
\Delta^{r/2 - 1}q_0+\left(\Delta^{r^2/2}q_{\tau-(r/2)}\right)^{\mu/2}\\
&\leq
\Delta^{r/2 - 1}q_0+ q_0^{(\frac{\mu}2)^{\tau/(r/2)}} \cdot \prod_{j=1}^{\tau/(r/2)} \left(\Delta^{r^2/2}\right)^{(\frac{\mu}2)^{j}} \\
&\leq
\Delta^{r/2 - 1}q_0 + \left(\Delta^{2(r^2/2)}q_0\right)^{(\frac{\mu}2)^{\tau/(r/2)}}
	&\text{($\mu/2 \geq 2$)}\\
&\leq
\Delta^{r/2 - 1}q_0 + \left(\Delta^{2(r^2/2)}q_0\right)^2
	&\text{($({\mu} / 2)^{\tau/(r/2)} \geq 2$)}\\
&\leq
\Delta^{r/2 - 1}q_0 + \Delta^{4(r^2/2) - 8r^2} q_0
	&\text{($q_0 \leq (ed)^{-8r}$ and $d=\Delta^r$)}\\
&\leq \Delta^{r/2}q_0. && \qedhere
\end{align*}
\end{proof}

We are now ready to prove that $S = L_\tau$ is stable.

\begin{lemma}\label{lem:stable}
For each vertex $u \notin  L_\tau$, $\degB_{L_\tau}(u)\leq\mu$ with high probability, and so $L_\tau$ is stable.
\end{lemma}
\begin{proof}
It suffices to show that $\Prob[\eventB_\tau(u)] = 1/\poly(n)$.
By Lemma~\ref{lem:obs}, $\Prob[\eventB_\tau(u)] \leq \Prob[\eventA_\tau(u)] = p_\tau$. We show that $p_\tau = 1/\poly(n)$.
\begin{align*}
p_\tau 
&\leq \left(\Delta^{r^2}p_0\right)^{(\frac{\mu}{2})^{\tau/(r/2)}} &\text{(Lemma~\ref{lem:stable-aux-1})}\\
&\leq \left(\Delta^{r^2  + r/2}q_0\right)^{(\frac{\mu}{2})^{\tau/(r/2)}} &\text{(Lemma~\ref{lem:stable-aux-2})}\\
&\leq \left(\Delta^{r^2  + r/2 - 8r^2}\right)^{(\frac{\mu}{2})^{\tau/(r/2)}} &\text{($q_0 \leq (ed)^{-8r}$ and $d=\Delta^r$)}\\
&\leq \left(\Delta^{-27}\right)^{(\frac{\mu}{2})^{\tau/(r/2)}} 			&\text{($r\ge 2$)}\\
&\leq \left(\Delta^{-27}\right)^{\Theta(\log n)} &\text{($\tau = \Theta(\log_\mu \log n)$ and $r=O(1)$)}\\
&\leq 1/\poly(n). & &\qedhere 
\end{align*}
\end{proof}

In Lemma~\ref{lem:small} we prove that $U_\tau$ is small, which implies that $S = L_\tau \subseteq U_\tau$ is also small.
We write $T^{[a,b]}$ to denote the graph defined by the vertex set $V(T)$ and the edge set $\{\{u,v\} \ | \ \dist_T(u,v) \in [a,b]\}$.
We first prove an auxiliary lemma.

\begin{lemma}\label{lem:small-aux}
Fix a $c \geq 1$. With probability $1 - n^{-\Omega(c)}$, the graph $H=T^{[r+1,4r]}$ has no connected subgraph $D$ such that (i) $|D|\geq c\log n$, and (ii) there is a subset $D' \subseteq D \cap U_0$ containing at least half of the
vertices in $D$, and $\dist_T(u,v) > r$ for distinct $u, v \in D'$.
\end{lemma}

\begin{proof}
The proof is similar to that of~\cite[Lemma~3.3]{BEPS16}.
Suppose that such $D$ exists, and consider a tree $\hat{T}$ in $H$ spanning $D$.
There are at most $4^{c\log n}$ different rooted unlabeled $c\log n$-node trees;
and each of them can be embedded into $H$ in less that $n\cdot \Delta^{4r(c\log n-1)}$ ways.
Moreover, there are at most $2^{c\log n}$ ways of selecting a subset $D' \subseteq D$.
Since $|D'| \geq c\log n / 2$ and $\dist_T(u,v) > r$ for distinct $u, v \in D'$,
the probability that such $\hat{T}$ exists is at most $q_0^{c\log n/2}$.

Recall that $q_0 \leq (ed)^{-8r}$, $d=\Delta^r$, $r \geq 2$, and $\Delta \geq 3$.
A union bound over all possibilities of $\hat{T}$ implies that such $D$ exists with probability at most
\begin{align*}
p' &= 4^{c\log n}\cdot n\cdot \Delta^{4r(c\log n-1)}\cdot 2^{c\log n}\cdot q_0^{c\log n/2}\\
&\leq n^{3c+1} \Delta^{-4c(r^2 - r) \log n}  e^{-4cr \log n}\\
&\leq n^{(4 - 4(r^2-r) \log \Delta - 4 \log e)c}\\
&\leq n^{-14c}. \qedhere
\end{align*}
\end{proof}

Recall from Definition~\ref{def:smallstable} that 
 $U_i$ is small if  each connected component  induced by $\bigcup_{v \in U_i} N^{r}(v)$
contains a distance-$2r$ dominating set of size at most $O(\log n)$. 

\begin{lemma}\label{lem:small}
With high probability, each connected component  induced by $\bigcup_{v \in U_\tau} N^{r}(v)$
contains a distance-$2r$ dominating set of size at most $O(\log n)$, and so $U_\tau$ is small.
\end{lemma}

\begin{proof}
Let $C$ be any connected component induced by $\bigcup_{v \in U_\tau} N^{r}(v)$.
We pick a distance-$2r$ dominating set $D$ of $C$ greedily, preferring vertices
in $U_0$ over $U_1$, and $U_1$ over $U_2$, etc.
Each time a vertex $v$ is picked we remove from consideration all vertices in $N^r(v)$.
Recall that $U_0 \subseteq \cdots \subseteq U_\tau$.
The set $D$ is obviously a distance-$r$ dominating set of
$U_\tau \cap C$. Since $U_\tau \cap C$ is itself a
distance-$r$ dominating set of $C$, the set $D$ is a
distance-$2r$ dominating set of $C$.

We write $u_i$ to denote the $i$th vertex added to $D$, and define
$D_i = \{u_1, \ldots, u_i\}$.
Let $m_i$ denote the number of connected components induced by $D_i$ in the graph $T^{[r+1,2r]}$ (rather than $T$).
We claim that if $u_i \notin U_0$, then $m_i < m_{i-1}$. This implies that at least half of the vertices in $D$ belong to $U_0$.
Observe that the set $D$ is connected in $H=T^{[r+1,4r]}$ (since $D$ is a distance-$2r$ dominating set of $C$), and so by Lemma~\ref{lem:small-aux}, $|D| = O(\log n)$ with high probability.

We prove the above claim in the remainder of the proof.
Consider the moment some $u_i \notin U_0$ is added to $D$.
We will show that the connected component of $D_i$ {in the graph $T^{[r+1,2r]}$} that contains $u_i$ is formed by merging $u_i$ with at least two connected components of $D_{i-1}$ {in the graph $T^{[r+1,2r]}$}.

The algorithm {\sf Find-Small-Stable-Set} added $u_i$ to $U_j$ because $u_i$ had at least $\mu/2\ge 2$ subtrees containing $U_{j-1}$-vertices that are within $N^r(u_i)$.
Let $T_1$ and $T_2$ be any two such subtrees. For each $k=1,2$, let $v_k$ be a $U_{j-1}$-vertex contained in both $T_k$ and $N^r(u_i)$.
Then there must be a vertex $w_k \in N^{r}(v_k)$ such that $w_k$ has been already added to $D$,
since otherwise the greedy algorithm should prefer $v_k$ over $u_i$.
Observe that $w_1$ and $w_2$ belong to separate connected components of $D_{i-1}$ {in the graph $T^{[r+1,2r]}$}, since $u_i \notin N^r(w_1) \cup N^r(w_2)$; but $w_1$, $w_2$, and $u_i$ are in the same component of $D_{i}$ {in the graph $T^{[r+1,2r]}$}, since $w_k \in N^{r}(v_k) \subseteq N^{2r}(u_i)$, for both $k=1,2$.
\end{proof}

We have proven (Lemmas~\ref{lem:stable} and \ref{lem:small})
that the algorithm {\sf Find-Small-Stable-Set} computes a set $S=L_\tau$
that is \emph{stable} and \emph{small}, in $O(\log_\mu\log n)$ time.
Lemma~\ref{lem:LLL-reduction} shows that any such algorithm can
be used to find a \emph{good} partial assignment to the variables in
any tree-structured LLL instance with $p(ed)^\lambda < 1$ and $\lambda \ge 2(4^r+8r)$.\footnote{It is possible to replace  $2(4^r+8r)$ with $2(4^r+cr)$ for some smaller $c$, but not too small. We do not attempt to optimize this coefficient.}
The \emph{stability} criterion is
used to show that the derived LLL instances satisfy $p'(ed)^{\lambda/2} < 1$ and $p' = \sqrt{p}$.
The \emph{smallness} criterion implies that the instances have size $\poly(\Delta)\log n$
and $\log n$-size, distance-$O(1)$ dominating sets.
Because $\log\mu = \Theta(\log\lambda)$,
the time to find the good partial assignment is $O(\log_\lambda\log n)$.\\

%% file: decomposition.tex
\section{Network Decomposition of Trees}\label{sect:tree-decomp}

Our interest in network decompositions stems from Lemma~\ref{lem:LLL-aux-2} due to~\cite{FischerG17}, 
which shows that they imply non-trivial \emph{deterministic} LLL algorithms.  
Most work on network decompositions~\cite{PanconesiS96}
has focussed on arbitrary graphs. 

Recall that a $(\lambda,\gamma)$-network decomposition is a 
partition of the vertices into $\lambda$ parts $V_1, \ldots, V_{\lambda}$ 
such that each $V_i$ induces connected components with diameter at most $\gamma$; and
 a $(\lambda_1,\gamma_1,\lambda_2,\gamma_2)$-network decomposition is a partition 
of the vertices into $\lambda_1+\lambda_2$ parts $V_1, \ldots, V_{\lambda_1},U_1, \ldots, U_{\lambda_2}$ 
such that each $V_i$ (resp.~$U_i$) induces connected components with diameter $\gamma_1$ (resp.~$\gamma_2$).

In this section we present two network decomposition algorithms for $T^k$ where $T=(V,E)$ is an $n$-vertex tree that contains a distance-$d$ dominating set
$S$ of size $s$.  In our application $d$ and $k$ are constants.
We assume all vertices agree on the numbers $(d,k,s)$.
We \emph{do not} need a specific dominating set $S$ to be given 
as input.

We emphasize that the network decomposition that we would like to compute is with respect to $T^k$, but the communication network is $T$.  All diameter parameters in network decompositions are measured with respect to $T^k$.

\subsection{A Simple Network Decomposition}\label{sect:tree-decomp1}

We first design a simple decomposition that partitions any tree-structured graph $T^k$ into 2 parts.

\begin{theorem}\label{thm:tree-decomp}
Let $T$ be a tree containing a distance-$d$ dominating set of size $s$.
There is a $\DetLOCAL$ algorithm $\mathcal{A}$ that computes a \underline{strong-diameter}
$(2, O(\log s + d/k))$-network decomposition of $T^k$ in $O(k\log s+d+k\log^\ast n)$ time,
i.e., $O(\log s + \log^\ast n)$ time when $d = O(1)$ and $k=O(1)$.
\end{theorem}

In what follows we prove Theorem~\ref{thm:tree-decomp}.
We assume the underlying communications network is $T$ rather than  $T^k$.
Consider the following two tree operations. They are similar to the ones described in~\cite{ChangP19}, which are inspired by Miller and Reif~\cite{MillerR89}. The second operation is parameterized by an integer $\ell \ge 2$. In our application we set $\ell = \Theta(k)$.

\medskip

\begin{description}
\item[\rake:] Remove all leaves and isolated vertices.
\item[\compress:] Remove all vertices that belong to some path $P$ such that (i) all vertices in $P$ have degree at most $2$, and (ii) the number of vertices in $P$ is at least $\ell$.
\end{description}

\medskip

Let $\mathcal{A}'$ be the algorithm on the tree $T$ defined as follows. 
(1) Do $3d+1$ \rake{} operations; 
(2) repeat the following sequence $\log s$ times: perform one \compress{} and then $\ell - 1$ \rake{} operations.

\begin{lemma}
Algorithm $\mathcal{A}'$ removes all vertices in $T$.
\end{lemma}
\begin{proof}
Let $S$ be any size-$s$ distance-$d$ dominating set of $T$.
Root $T$ at an arbitrary vertex and let $\size(v)$ be the number of vertices in the subtree $T_v$ rooted at $v$ that belong to $S$.
For any vertex $v \in V$,
we prove by induction that (i) if $\size(v) \leq 1$, then $v$ is removed in Step~(1) of $\mathcal{A}'$, and (ii) if $1 < \size(v) \leq 2^i$, then
$v$ is removed on or before the $i$th iteration of Step~(2) of $\mathcal{A}'$.

For the case $\size(v) \leq 1$, observe that the height of the subtree $T_v$ rooted at $v$ is at most $3d$.
Suppose the height of $T_v$ is at least $3d+1$, then there is a path $P$ connecting $v$ and a leaf that has at least $3d+2$ vertices. We claim that for any distance-$d$ dominating set $S$ of $T$, we need to have  $|S \cap T_v| \geq 2$. For each  $u \in S \cap T_v$,  $u$ can dominate at most $2d+1$ vertices in $P$, and so there must be at least one vertex $x$ in $P$ that is not dominated by $u$ and its distance to $v$ is at least $d$. To dominate $x$, we need another vertex in $S \cap T_v$, and so  $|S \cap T_v| \geq 2$, contradicting the assumption $\size(v) \leq 1$. 
Therefore, the entire subtree $T_v$ (including $v$) must be removed after the initial $3d+1$ \rake{} operations.

Consider the case $2^{i-1} < \size(v) \leq 2^i$. By the inductive hypothesis, all vertices $u$ with 
$\size(u) \leq  2^{i-1}$ have been removed before the $i$th iteration of Step~(2).
With respect to the vertex $v$, define $V'$ to be the set of all vertices $u$ such that (i) $\size(u) > 2^{i-1}$, and (ii) $u$ is in the subtree $T_v$ rooted at $v$. The set $V'$ induces a path with one endpoint at $v$, since otherwise $\size(v) > 2 \cdot 2^{i-1}=2^i$. Let $C$ be a connected component induced by vertices in $V'$ that are not removed yet. If $|C| \geq \ell$, then all vertices in $C$ are removed after 1 \compress. Otherwise, all vertices in $C$ are removed after $\ell - 1$ \rake{} operations.
\end{proof}

In the following discussion, the notions of connected components and degrees are  with respect to $T$.
To compute a $(2, O(\log s + d/k))$-network decomposition of $T^k$, it suffices to compute a partition $V = V_1 \cup V_2$ meeting the following two conditions. 
\begin{description}
\item (C1) For both labels $c \in \{1,2\}$, any two vertices $u$ and $v$ in two distinct connected components of $V_c$  must have $\dist_T(u,v) > k$. This guarantees that the set of connected components of $V_c$ remains unaltered if we change the underlying graph from $T$ to $T^k$.
\item (C2) For both labels $c \in \{1,2\}$, each connected component of $V_c$ has diameter at most $O(k \log s + d)$.  This implies the diameter upper bound of $O(\log s + d/k)$ when the underlying graph is $T^k$.
\end{description}

Recall that $\mathcal{A}'$ performs $L_r = (3d+1) + (\ell-1) \log s$ \rake\ and $L_c = \log s$ \compress{} operations;
let $L = L_r + L_c = (3d+1) + \ell \log s$. 
We write $U_i$ to denote the set of all vertices that are removed during the $i$th operation.
We are now in a position to present the algorithm $\mathcal{A}$.
The algorithm $\mathcal{A}$ begins by computing the decomposition $V = \bigcup_{i=1}^L U_i$ using $\mathcal{A}'$.
Then, for $i = L$ down to $1$, label all vertices $v \in U_i$ by $\{1,2\}$ as follows.

\paragraph{Case 1.}
If the $i$th operation is \rake, then label $U_i$ as follows. Let $v \in U_i$. For the case that $v$ is of degree-1 in the subgraph induced by $\bigcup_{j=i}^L U_j$, let $u$ be the unique neighbor of $v$ in $\bigcup_{j=i}^L U_j$. If $u \notin U_i$, then $v$ adopts the same label as $u$. Otherwise, $u \in U_i$ must also be of degree-1 in $\bigcup_{j=i}^L U_j$; we give both $u$ and $v$ 
the same label $c \in \{1,2\}$. 
For the case that $v$ is an isolated vertex 
of $\bigcup_{j=i}^L U_j$, we label $v$ by any $c \in \{1,2\}$.

\paragraph{Case 2.}
If the  $i$th operation is \compress, then label $U_i$ as follows. Let $P$ be a path that is a connected component of $U_i$. The number of vertices in $P$ is at least $\ell = \Theta(k)$. Compute a labeling of the vertices in $P$ meeting the following conditions: (i) each connected component induced by  vertices of the same label has size within $[k, 7k]$, (ii) if $v$ is an endpoint of $P$ that is adjacent to a vertex $u \in \bigcup_{j=i+1}^L U_j$, then the label of $v$ is the same as the label of $u$.

Such a labeling of $P$ can be computed in $O(k)$ time if we are given an independent set $I$ of $P$ such that each connected component of $P \setminus I$ has size within $[3k, 6k]$.
Suppose that we already have such a set $I$. For each $v \in I$, we find an arbitrary subpath $P_v \subseteq P$ that contains $v$ and has exactly $k$ vertices. All vertices in $\bigcup_{v\in I} P_v$ are labeled 1, and the remaining vertices in $P$ are labeled 2. At this moment,  each connected component induced by  vertices of label 1 has size $k$, and  each connected component induced by  vertices of label 2 has size within $[3k - 2(k-1), 6k] = [k+2, 6k]$. If there is a component $C$ violating Condition~(ii) of the previous paragraph, 
we flip the label of all vertices in $C$ (i.e., from 1 to 2 or from 2 to 1). 
If $\ell \geq c k$ for some large enough universal constant $c$, 
then we obtain a labeling satisfying both Condition~(i) and Condition~(ii).

The computation of the independent set $I$ can be done in $O(k \log^\ast n)$ time, as we explain below.
Suppose that we have an independent set $I'$ of $P$ such that each connected component of $P \setminus I$ has size within $[\alpha, 2\alpha]$.
We show that in $O(\alpha \log^\ast n)$ time we can compute an independent set $I''$ of $P$ such that  each connected component of $P \setminus I$ has size within $[\beta, 2 \beta]$, for any prescribed number $\beta \leq 2\alpha+1$. Let $\tilde{P}$ be the ``imaginary path'' formed by contracting all vertices in  $P \setminus I$. A maximal independent set $\tilde{I}$ of  $\tilde{P}$ can be computed in $O(\alpha \log^\ast n)$ time. At this point,  each connected component $C$ of $P \setminus \tilde{I}$ has size within $[2\alpha+1, 4\alpha + 2]$. The component size constraint $[\beta, 2 \beta]$ can be met by adding new vertices to $\tilde{I}$ to subdivide the oversized components.
The desired independent set $I$ can be computed by $\log k$ iterated applications of the above procedure, and the runtime is $\sum_{i=1}^{\log k} O(2^i \log^\ast n) = O(k \log^\ast n)$.

\paragraph{Time Complexity.}
The total running time of $\mathcal{A}$ is $ O(L_r + k L_c) + O(k\log^\ast n) = O(k\log s + d + k\log^\ast n)$, since the independent set computation of paths removed by the \compress\ operation can be computed in $O(k\log^\ast n)$ time {\em in parallel}.

\paragraph{Validity of Labeling.}
We now verify that the labeling resulting from  $\mathcal{A}$ satisfies the two conditions (C1) and (C2).
Consider two distinct connected components $C$ and $C'$ induced by $V_1$. In view of  Case~2 of algorithm  $\mathcal{A}$, any path $P'$ connecting a vertex in $C$ and a vertex in $C'$ in $T$ must contain a subpath $P''$ consisting of $k$ vertices in $V_2$. The same is true if we swap $V_1$ and $V_2$, and so (C1) holds. Consider a connected component $C$ by $V_1$ or $V_2$. Let $i^\star$ be the largest index $i$ such that $U_i \cap C \neq \emptyset$, and let $v^\star$ be any vertex in $C \cap U_{i^\star}$. We show that for any vertex $u \in C$, the unique path $P$ connecting $u$ and $v^\star$ in $T$ contains $O(L_r + k L_c) = O(k\log s+d)$ vertices, and so (C2) holds. Consider any index $i \in [1, i^\star]$. If the $i$th operation is \rake, then we have $|P \cap U_i| \leq 2$ (in view of Case 1). If the $i$th operation is \compress, then we have $|P \cap U_i| \leq 7k$ (in view of Case 2). Thus, indeed $|P| = O(L_r + k L_c)$.

\subsection{A Mixed-Diameter Network Decomposition}

In this section, we compute a mixed-diameter network decomposition of $T^k$ consisting of one part with weak diameter $O(\log_{\lambda/k}s+d/k)$ and $O(\lambda^2)$ additional parts whose connected components have strong diameter at most~$1$.

\begin{theorem}\label{thm:tree-decomp2}
Let $T$ be a tree containing a distance-$d$ dominating set of size $s$.
There is a $\DetLOCAL$ algorithm $\mathcal{A}$ that computes a \underline{weak-diameter}
$(1,O(\log_{\lambda/k} s + d/k),O(\lambda^2),1)$-network decomposition of $T^k$ in $O(k\log_{\lambda/k} s+d+k\log^* n)$ time,
where $\lambda=\Omega(k)$ is sufficiently large, i.e., $\lambda \geq ck$ for some universal constant $c$.
When $k=O(1)$ and $d=O(1)$, the time bound is $O(\log_\lambda s+\log^* n)$.
\end{theorem}

In what follows we prove Theorem~\ref{thm:tree-decomp2}.
For each tree operation, let $T_i$ denote the set of vertices that remain immediately before the $i$th operation. Consider the following two operations applied to $T_i$.

\medskip

\begin{description}
\item[\rake:] Remove all leaves and isolated vertices.
\item[\compress:] Remove all vertices $v$ such that $|N^{4k}(v)\cap T_i|\leq\lambda$.
\end{description}

\medskip

Set 
$m=\left\lfloor\frac{\lambda}{4k}\right\rfloor-1$.
By choosing the universal constant $c$ sufficiently large, we may assume that $m\geq 2$ and $m=\Theta(\lambda/k)$.
Let $\mathcal{A}^*$ be the following algorithm on $T$:
(1) perform $3d+1$ \rake{} operations; and
(2) repeat $\lceil\log_m s\rceil$ times the sequence consisting of one \compress{} operation followed by $4k$ \rake{} operations.

\begin{lemma}\label{lem-aux-decomp2}
Algorithm $\mathcal{A}^*$ removes all vertices in $T$.
\end{lemma}

\begin{proof}
Let $S$ be any distance-$d$ dominating set of $T$ of size $s$.
Root $T$ at an arbitrary vertex, and let $\size(v)$ be the number of vertices of $S$ in the subtree rooted at $v$.
We prove by induction that (i) if $\size(v)\leq 1$, then $v$ is removed in Step~(1) of $\mathcal{A}^*$, and (ii) if $1<\size(v)\leq m^i$, then $v$ is removed within the first $i$ iterations of Step~(2).

If $\size(v)\leq 1$, then the subtree rooted at $v$ has height at most $3d$, and hence the entire subtree, including $v$, is removed by the first $3d+1$ \rake{} operations.

Now suppose that $m^{i-1}<\size(v)\leq m^i$, and assume inductively that every vertex $u$ with $\size(u)\leq m^{i-1}$ has already been removed during the first $i-1$ iterations of Step~(2).
Let $V'$ be the set of descendants $u$ of $v$ satisfying $\size(u)>m^{i-1}$.
All descendants of $v$ outside $V'$ have already been removed, and $V'$ induces a subtree rooted at $v$ with at most $m-1$ leaves.

Let $T'$ be the remaining tree immediately before the \compress{} operation in the $i$th iteration of Step~(2), and consider a vertex $u\in V'$ with $\dist_T(u,v)\geq 4k$.
The subgraph induced by the vertices of $V'$ within distance $4k$ of $u$ can be viewed as a tree rooted at $u$, of height at most $4k$, with at most $(m-1)+1=m$ leaves. It therefore contains at most $4km+1\leq\lambda$ vertices. The assumption $\dist_T(u,v)\geq4k$ guarantees that $N^{4k}(u)\cap T'$ contains no ancestor of $v$. Consequently,
\[
|N^{4k}(u)\cap T'|\leq 4km+1\leq\lambda,
\]
and $u$ is removed by the next \compress{} operation. The vertices of $V'$ that remain afterward are all within distance less than $4k$ of $v$, and hence are removed by the following $4k$ \rake{} operations.
\end{proof}

We now describe the network decomposition algorithm $\mathcal{A}$.
First run $\mathcal{A}^*$, and let $C$ be the set of vertices removed by \compress{} operations; we call the vertices in $C$ \emph{centers}. For every center $v$, we mark all vertices in $N^{k/2}(v)$. The set of marked vertices is
\[
\mathcal{M}=\{x\in V:\dist_T(x,C)\leq k/2\}.
\]
Each marked vertex $x\in\mathcal{M}$ is assigned to an arbitrary center $c(x)\in C$ satisfying $\dist_T(x,c(x))\leq k/2$, with $c(v)=v$ for every center $v\in C$. For each center $v\in C$, let
\[
K_v=\{x\in\mathcal{M}:c(x)=v\}
\]
be the cluster centered at $v$.

Define the center graph $H=(C,E_H)$ by setting
\[
\{u,v\}\in E_H
\quad\text{if and only if}\quad
\dist_T(u,v)\leq 2k.
\]
We compute a proper $O(\lambda^2)$-coloring of $H$ and assign every vertex in $K_v$ the color of its center $v$. All unmarked vertices receive color $0$.

We prove that (i) $\Delta(H)\leq\lambda-1$; (ii) every connected component induced by a nonzero color has strong diameter at most $1$ in $T^k$; and (iii) every connected component induced by color $0$ has weak diameter $O(\log_{\lambda/k}s+d/k)$ in $T^k$.

\paragraph{Proof of (i).}
Fix a center $v\in C$, and let
\[
X=\{x\in C:\dist_T(x,v)\leq 2k\}.
\]
Choose a vertex $u\in X$ that is removed by the earliest \compress{} operation among the vertices of $X$; ties within the same operation are broken arbitrarily. Let $T'$ be the remaining tree immediately before that \compress{} operation.
Every vertex $x\in X$ belongs to $T'$: it cannot have been raked earlier because it is eventually removed by \compress{}, and it cannot have been compressed earlier by the choice of $u$. Moreover,
\[
\dist_T(u,x)\leq\dist_T(u,v)+\dist_T(v,x)\leq4k.
\]
Hence
\[
X\subseteq N^{4k}(u)\cap T'.
\]
Since $u$ is removed by \compress{}, $|N^{4k}(u)\cap T'|\leq\lambda$, and so $|X|\leq\lambda$. Thus, $\deg_H(v)\leq\lambda-1$, so $\Delta(H)\leq\lambda-1$.

\paragraph{Proof of (ii).}
For every center $v\in C$ and every two vertices $x,y\in K_v$,
\[
\dist_T(x,y)\leq\dist_T(x,v)+\dist_T(v,y)\leq k.
\]
Thus $K_v$ is a clique in $T^k$, and hence has strong diameter at most $1$.

It remains to show that two distinct clusters of the same color are not adjacent in $T^k$.
Suppose that $x\in K_u$ and $y\in K_v$, where $u\neq v$, and that $\dist_T(x,y)\leq k$. Then
\[
\dist_T(u,v)
\leq\dist_T(u,x)+\dist_T(x,y)+\dist_T(y,v)
\leq k/2+k+k/2
=2k.
\]
Therefore $u$ and $v$ are adjacent in $H$ and receive distinct colors. Hence each monochromatic connected component of a nonzero color is exactly one cluster $K_v$ and has strong diameter at most $1$ in $T^k$.

\paragraph{Proof of (iii).}
Let $R$ be a connected component of $T-C$. Every vertex of $R$ is
removed by a \rake{} operation. Moreover, along any path in $R$, at
most two vertices are removed by each \rake{} operation. Since the
total number of \rake{} operations is
\[
3d+1+4k\lceil\log_m s\rceil
    = O(k\log_{\lambda/k}s+d),
\]
each connected component of $T-C$ has strong diameter
$O(k\log_{\lambda/k}s+d)$ in $T$.

We claim that every connected component induced by the unmarked
vertices in $T^k$ is contained in a single connected component of
$T-C$. Indeed, suppose that two unmarked vertices $x$ and $y$ belong
to distinct connected components of $T-C$. The unique $x$--$y$ path
in $T$ contains a center $w\in C$. Since $x$ and $y$ are unmarked,
\[
\dist_T(x,w)>k/2
\qquad\text{and}\qquad
\dist_T(y,w)>k/2.
\]
Consequently,
\[
\dist_T(x,y)
 = \dist_T(x,w)+\dist_T(w,y)
 > k,
\]
and hence $x$ and $y$ are not adjacent in $T^k$. Therefore, no edge
of $T^k$ between unmarked vertices joins distinct connected
components of $T-C$. It follows that every connected component induced by color $0$ has
weak diameter
\(
O\left(\log_{\lambda/k}s+\frac{d}{k}\right)
\)
in $T^k$.

\paragraph{Summary.}
Combining (i)--(iii), $\mathcal{A}$ computes a weak-diameter
$(1,O(\log_{\lambda/k}s+d/k),O(\lambda^2),1)$-network decomposition of $T^k$ in
$O(k\log_{\lambda/k}s+d+k\log^*n)$ time.  

The decomposition algorithm  $\mathcal{A}^*$ takes $O(k\log_{\lambda/k}s+d)$ time.
The $O(\lambda^2)$-coloring of $H$ can be computed using Linial's algorithm~\cite{Linial92}. One communication round in $H$ can be simulated in $O(k)$ rounds in $T$, so the coloring takes $O(k\log^*n)$ time.

%% file: edgeColorAlg.tex
\section{Deterministic Algorithms for Edge Coloring Trees}\label{sect:upperbound}

Let $T=(V,E)$ be a tree with $n$ vertices and $N^+(v) = N(v) \cup \{v\}$ be the inclusive neighborhood of $v$.
We decompose $T$ using another variation on 
Miller and Reif's~\cite{MillerR89} rake and compress operations,
the second of which is parameterized by an integer $k\ge 2$.

\medskip
\begin{description}
\item {\rake}: Remove all leaves and isolated vertices from $T$.
\item {\compress}: Remove the set $\{v\in V \;|\; \mbox{for every $u\in N^+(v)$, $\deg_T(u) \le k$}\}$
from $T$.
\end{description}

\begin{theorem}
Alternately applying \compress\ and \rake\ $1+ \log_k n$ times removes all vertices from any $n$-vertex tree $T$.
\end{theorem}

\begin{proof}
Root $T$ at an arbitrary vertex and let $\size(v)$ be the number of vertices in the subtree rooted at $v$.
We prove by induction that if $\size(v) \le k^i$, $v$ will be removed after the first $i+1$ rounds of \compress\ and \rake.  
The claim is trivially true when $i=0$.
Assume the claim is true for $i-1$.  Let $v$ be any vertex with $\size(v) \in (k^{i-1},k^i]$,
and define $V'$ to be the set of all vertices $u$ such that (i) $\size(u) \in(k^{i-1},k^i]$ and (ii) $u$ is in the subtree rooted at $v$.
Notice that each vertex $u \in V'$ has $\deg_{V'}(u) \leq k$, since otherwise $\size(u) > k^i$.
By the inductive hypothesis, all descendants of $v$ that are not in $V'$ have been removed after $i$ rounds of
\compress\ and \rake.  The $(i+1)$th \compress\ will remove all remaining vertices in $V' - \{v\}$.
Hence all descendants of $v$ have been removed after the  $(i+1)$th \compress. 
However, the degree of the parent
of $v$ is unbounded, so $v$ may not be removed.
If $v$ still remains, the $(i+1)$th \rake\ will remove it.
\end{proof}

\begin{theorem}
There is an $O(\log_\Delta n)$-time $\DetLOCAL$ algorithm for $\Delta$-edge coloring a tree $T$ with maximum degree $\Delta\ge 3$.
\end{theorem}

\begin{proof}
Let $\beta$ be the constant such that Linial's algorithm~\cite{Linial92} finds a
$\beta\Delta^2$-edge coloring in $O(\log^* n - \log^* \Delta + 1)$ time.
We begin by decomposing $T$ with \compress\ and \rake\ steps, using parameter $k = \max\{2,\, \floor{(\Delta/\beta)^{1/3}}\}$.
Define $T_i = (V_i,E_i)$ to be the forest before the $i$th round of \compress\ and \rake, and let $V_i^c$ and $V_i^r$
be those vertices removed by the $i$th \compress\ and \rake, respectively.

We edge color the trees $T_{1 + \log_k n},\ldots,T_1=T$ in this order.
Given a coloring of $T_{i+1}$, we need to color the remaining uncolored edges in $T_i$.
Let $u\in T_{i+1}$ be a vertex, and let $v_1,\ldots,v_x \in V_i^r$ be the vertices adjacent to $u$
removed by the $i$th \rake.  At this point $u$ is incident to at most $\Delta-x$ colored edges.  We assign
to $\{u,v_1\},\ldots,\{u,v_x\}$ any distinct available colors from their palettes.

We now turn to the vertices removed by the $i$th \compress. First, suppose that $\Delta$ is large enough such
that $k = \floor{(\Delta/\beta)^{1/3}}$.
Let $\phi$ be a $\beta k^2$-edge coloring of the (as yet uncolored) subgraph of $T_i$ (i.e., the edges that are incident to some vertices in $V_i^c$).
We argue that this subgraph has maximum degree at most $k$, and so we are able to apply Linial's algorithm~\cite{Linial92} to find a
$\beta k^2$-edge coloring. 
Suppose $e = \{u,v\}$ is in this subgraph, but either $\deg_{T_i}(u) > k$ or $\deg_{T_i}(v) > k$.  
If this were true, neither $u$ nor $v$ could have 
been removed by the $i$th \compress, contradicting the fact that $e$ is incident to some vertices in $V_i^c$.

Partition the palette $\{1,\ldots,\Delta\}$ into $\beta k^2$ parts $P_1,\ldots,P_{\beta k^2}$.
Each part has size $\Delta/(\beta k^2) \ge k$.
Each $v\in V_i^c$ colors each edge $\{v,u\}$ any available color in $P_{\phi(\{v,u\})}$.  Since $\deg_{T_i}(u) \le k$,
at most $k-1$ of its incident edges may already be colored, 
and so there must be at least one available color in $P_{\phi(\{v,u\})}$ for $\{v,u\}$ to use.   All calls to Linial's $\beta k^2$-edge coloring algorithm
can be executed in parallel, so the overall time is $O(\log_k n + \log^* n - \log^* k) = O(\log_\Delta n)$.

When $k=2$, the subgraph induced by $V_{1}^c \cup \cdots \cup V_{1+\log_k n}^c$ consists of a set of paths.
In $O(\log^* n)$ time, we find an \emph{initial} 3-edge coloring of these paths.  We now color $T_{1 + \log_k n},\ldots, T_1$
in this order.  Coloring the edges removed during a \rake\ is done as before.  The set $V_i^c$ removed in
one \compress\ induces some paths, each end-edge of which may be adjacent to one (previously colored) edge in $T_{i+1}$.
If the initial color of an end-edge conflicts with the coloring of $T_{i+1}$, we recolor it any available color.  When $k=2$
this procedure takes $O(\log^* n + \log_k n) = O(\log_\Delta n)$ time.
\end{proof}

An  \emph{oriented} tree is a rooted tree where each vertex that is not the root knows its parent. 
We show that a $(\Delta+1)$-edge coloring of an oriented tree can be found in  $O(\log^* n)$ time, 
but $\Delta$-edge coloring  takes $\Omega(\log_\Delta n)$ time.

\begin{theorem}
Any {oriented} tree $T$ can be $(\Delta+1)$-edge colored in $O(\log^* n)$ time.
\end{theorem}

\begin{proof}
Initially pick color $\phi_0(\{u,\parent(u)\}) = i$ if $\ID(u)$ is the $i$th largest ID among its siblings.
Observe that for any $i$, $\phi_0^{-1}(i)$ is a subgraph consisting of oriented paths, and that
$\phi_0^{-1}(\Delta)$ is at most one edge, attached to the root.  For each $i\in\{1,\ldots,\Delta-1\}$,
in parallel, recolor $\phi_0^{-1}(i)$ using the color set $\{i,\Delta,\Delta+1\}$ in such a way that the
most ancestral edge in each path remains colored $i$.
This takes $O(\log^* n)$ time~\cite{ColeV86,Linial92}.

The result is a legal $(\Delta+1)$-edge coloring.
It is clear that for each $i\in\{1,\ldots,\Delta-1\}$, no two edges with color $i$ are adjacent. Now consider $i \in \{\Delta, \Delta + 1\}$. Suppose there exist two adjacent edges $e = \{u,v\}$ and $e' = \{v,w\}$ that are both colored $i$. Let $j \in \{1,\ldots,\Delta\}$ be the original color   of $e$   before recoloring, and let $P$ be the $j$-color (before recoloring) oriented path containing $e$.  Similarly, let $j' \in \{1,\ldots,\Delta\}$ be the original color   of $e'$   before recoloring, and let $P'$ be the $j'$-color (before recoloring) oriented path containing $e'$.
Then the two paths $P$ and $P'$ intersect only at $v$, and so at least one of $e$ and $e'$ is the most ancestral edge of the corresponding path. This contradicts the assumption that they are colored by $i \in \{\Delta, \Delta + 1\}$ (after recoloring). Thus, all edges colored $i \in \{\Delta, \Delta + 1\}$ are not adjacent to each other.
\end{proof}

\begin{theorem}
Any $\Delta$-edge coloring algorithm for oriented trees takes $\Omega(\log_\Delta n)$ time in $\RandLOCAL$.
\end{theorem}

\begin{proof}
Let $T$ be an oriented $\Delta$-regular tree with height $h = \Theta(\log_\Delta n)$ and $\mathcal{A}$ be an edge coloring
algorithm running in $h/3$ time.  The color of $\{u,\parent(u)\}$ is uniquely determined by the colors
of the edges incident to leaf-descendants of $u$. Let $V'$ denote the set of leaf-descendants of $u$. In general, $N^{h/3}(u)$ and $\bigcup_{v \in V'} N^{h/3}(v)$ do not intersect. In this case, $u$ only has a $1/\Delta$ chance of guessing the correct edge color; if it guesses incorrectly,
there must be a violation somewhere in the subtree rooted at $u$.
\end{proof}

%% file: conclusion.tex
\section{Concluding Remarks}\label{sect:conclusion}

The focus of this paper has been on the complexity of distributed \emph{edge-coloring}, on general graphs and trees, with and without randomization.  Nonetheless, we took several extended detours into apparently unrelated 
topics such as the distributed \Lovasz{} local lemma
(Section~\ref{sect:tree-LLL}) 
and network decompositions (Section~\ref{sect:tree-decomp}).
A recent line of work on developing a \emph{complexity theory} for the $\LOCAL$ 
model~\cite{NaorS95,BrandtEtal16,ChangKP19,ChangP19,FischerG17,GhaffariHK17,BalliuHOS19,BrandtHKLJPRSU17,BalliuHKLOS18,BalliuBOS18,BalliuBCORS19,GhaffariKM17} 
explains why these particular detours are \emph{natural} and perhaps unavoidable in the pursuit of optimal $\LOCAL$ algorithms.

The appearance of the distributed \Lovasz{} local lemma (LLL)
is no surprise at all, given that it generalizes a problem
related to $(1+\epsilon)\Delta$-edge coloring, namely
sinkless orientation (Theorem~\ref{thm:so-ec-reduction}),
is complete for sublogarithmic time~\cite[Thm.~4.1]{ChangP19},
and is a generally useful tool for finding objects
that cannot be generated by a 
greedy algorithm~\cite{MolloyR01,PettieS15,ChungPS17,ElkinPS15,FischerG17}. 
The \emph{structure} of the LLL algorithm in 
Sections~\ref{sect:tree-LLL} and~\ref{sect:tree-decomp} also turns out to be quite natural.  
Chang, Kopelowitz, and Pettie's derandomization~\cite[Thm~3.1]{ChangKP19}
justifies why we \emph{must} apply the graph shattering method to solve the LLL in 
randomized $O(\log\log n)$ time, and that any such algorithm must contain within it
a deterministic $O(\log n)$-time algorithm.  Our $O(\log n)$-time deterministic
LLL algorithm for trees (Theorem~\ref{thm:detLLL-tree})
follows 
Fischer and Ghaffari~\cite{FischerG17},
who showed how to solve LLL instances using
network decompositions.  
Ghaffari, Kuhn, and Maus~\cite{GhaffariKM17}
show that this choice also turns out to be \emph{natural}, in the sense
that you cannot solve the LLL deterministically \emph{without} computing good network decompositions deterministically.\footnote{In particular, the distributed LLL is $\PSLOCAL$-hard
as it generalizes the $\PSLOCAL$-complete problem of \emph{Weak Local Splitting}~\cite[Thm.~1.4]{GhaffariKM17}.  As a consequence, 
any deterministic  $\poly(\log n)$ LLL algorithm can also be used to compute 
$(\poly(\log n),\poly(\log n))$-network decompositions deterministically.}

%% file: concentration.tex
\section{Proof of Lemma~\ref{thm:concentration}}\label{sect:concentration}

In this section we prove the concentration bounds of Lemma~\ref{thm:concentration}.
For notational simplicity, we ignore all subscripts $i$, i.e.,
$p,d,t$ are the palette size, degree, and $c$-degree
before the $i$th round of coloring, all of which satisfy invariant $\mathcal{H}_i$.
Recall that we introduce imaginary edges, if necessary, to
ensure that the entire graph has uniform $c$-degree $t$ and uniform palette size $p$.
$S(v)$ is the set of real edges incident to $v$, $|S(v)|\le d$,
and $N_c(v)$ the set of real and imaginary edges incident to $v$ with $c$ in their palettes.
The arguments of this section do not differentiate between real and imaginary edges.
From Lemma~\ref{lem:estimate} we use the fact that $t = \Theta(p)$, i.e., $t$ and $p$
are interchangeable in those parts of the proof that are not sensitive to the leading constant.

We make extensive use of Theorem~\ref{lem:concentration} and Lemma~\ref{lem:aux}
to prove Lemma~\ref{thm:concentration}.  Theorem~\ref{lem:concentration} is
from Dubhashi and Panconesi's book~\cite{DubhashiPanconesi09} on the concentration of measure,
where it is called the \emph{method of bounded variances}.  Ignoring the leading constant in the exponent,
Theorem~\ref{lem:concentration} is strictly more powerful than Chernoff-Hoeffding
and Azuma-type inequalities, and is best suited in applications that have the following two features:
\begin{itemize}
\item We are interested in deviations of $f(\mathbf{X}_n)$ from its expectation (up to $\pm s$)
that are significantly smaller than the number of underlying random variables ($n$) times the
Lipschitz bound satisfied by the martingale ($M$).  This feature renders Azuma's inequality
too weak to be of any use.\footnote{A vector $(X_1, \ldots, X_i)$ of random variables is written $\mathbf{X}_i$.}

\item The Lipschitz bound is pessimistic: although $D_i = \Expect[f|\mathbf{X}_i] - \Expect[f|\mathbf{X}_{i-1}]$
can be as large as $M$, its variance ($\sigma_i^2$) conditioned on \underline{any} $\mathbf{X}_{i-1}$ is substantially smaller.
\end{itemize}

For example, in the first round of coloring, the $c$-degree of a vertex $v$
depends on $\Theta(\Delta^3)$ random variables (colors chosen by edges in the 3-neighborhood)
but we are interested in deviations from the expected $c$-degree that are $s=O(\Delta)$.
Any single edge could have a significant effect on $v$'s $c$-degree ($M=\Theta(1)$),
but the \emph{variances} of these effects are substantially smaller.
In particular, the sum of variances $\sum_i \sigma_i^2$ will be $O(\Delta)$.

\begin{theorem}[{\cite[Equation (8.5)]{DubhashiPanconesi09}}] \label{lem:concentration}
Let $X_1, \ldots, X_n$ be an arbitrary set of random variables.
Let $f(X_1, \ldots, X_n)$ be such that $\Expect[f]$ is finite.
We write
$D_i \bydef \Expect[f|\mathbf{X}_i] - \Expect[f|\mathbf{X}_{i-1}]$.
Suppose that there exist $M$ and values $\{\sigma_i^2\}_{1\le i\le n}$
meeting the following conditions.
\begin{itemize}
\item For any assignment to the random variables $\mathbf{X}_{i-1}$, $\Var[D_i | \mathbf{X}_{i-1}] \leq \sigma_i^2$.
\item For any assignment to the random variables $\mathbf{X}_i$, $|D_i| \leq M$.
\end{itemize}
Then $\Prob[f > \Expect[f] + s] \leq \exp\left(-\frac{s^2}{2\left(\sum_{i=1}^n \sigma_i^2 + Ms/3\right)}\right)$.
\end{theorem}

Lemma~\ref{lem:aux} follows from straightforward calculation.

\begin{lemma}\label{lem:aux}
Let $X$ be a random variable such that (i) $\Expect[X] = 0$, (ii)  $\Prob[X = a] = \alpha$ and $\Prob[X = b] = 1 - \alpha$, and
(iii) $|a-b| \leq k$.  Then we have the following.
\begin{itemize}
\item $\Var[X] \leq \alpha(1-\alpha) k^2 \leq \alpha k^2$.
\item $|b| \leq \alpha k$.
\item $|a| \leq (1 - \alpha)k \leq k$.
\end{itemize}
\end{lemma}

Throughout this section, we use the following notation. For each edge $e$ and each color $c$, define $z_{e,c}$ as the indicator random variable that $e$ {\em successfully} colors itself $c$, thus $z_{e,c} = 0$ if $c \notin \Psi(e)$.

\subsection{Concentration of Vertex Degree}\label{sect:concentration-deg}

Let $v^\bullet$ be a vertex.
We claim that
$\Expect[|S^\diamond(v^\bullet)|] \leq d^\diamond$.
An edge $e$ successfully colors itself with probability $(1 - 1/p)^{2(t - 1)}$, since there are $2(t-1)$ edges competing with $e$ for $\Color(e)$, and each of these $2(t-1)$ edges selects $\Color(e)$ with probability $1/p$.
Thus, by linearity of expectation,
$$\Expect[|S^\diamond(v^\bullet)|] =  (1 - (1 - 1/p)^{2(t - 1)})|S(v^\bullet)| \leq (1 - (1 - 1/p)^{2(t - 1)})d =  d^\diamond.$$
For brevity, we write $S \bydef S(v^\bullet)$, $S^\diamond \bydef S^\diamond(v^\bullet)$, and $z \bydef |S| - |S^\diamond|$.
The goal of this section is to show that $\Prob[z < \Expect[z] - s] = \exp\left(-\Omega(s^2 / |S|)\right)$, which implies the desired concentration bound $\Prob\left[|S^\diamond(v^\bullet)| > (1+\delta) d^\diamond \right]  =  \exp\left(-\Omega(\delta^2 d)\right)$, by setting $s = \delta d^\diamond$.

\paragraph{Notations.}
We write $z_e  \bydef \sum_{c \in \Psi(e)} z_{e,c}$  and $z_c  \bydef \sum_{e \in S} z_{e,c}$.
In other words,  $z_e$ is the indicator random variable that $e$ successfully colors itself;
$z_c$ is the indicator random variable that some edge in $S$ successfully colors itself by $c$.
We can express $z$ as $z = \sum_{e \in S} z_e$ or $z = \sum_{c} z_c$,
where the summation is over all colors $c \in \bigcup_{e \in S}\Psi(e)$.

Let $S'$ denote the set of edges such that $e' \in S'$
if there exists $e =\{v^\bullet, u\}\in S$ such that (i) $\Psi(e) \cap \Psi(e') \neq \emptyset$, and (ii) $e'$ is incident to $e$.
For each edge $e' \in S'$ and for each color $c \in \Psi(e')$, we define
$R(e',c)$ as the subset of $S$ such that $e \in R(e',c)$ if (i) $e$ is incident to $e'$, and (ii) $c \in \Psi(e)$.
We write $w(e',c) = |R(e',c)|$ and $w(e')=\sum_{c \in \Psi(e')} w(e',c)$.
Notice that the value $w(e',c)$ may exceed 2 when $e' \notin S$ is an imaginary edge incident to $v^\bullet$.
Intuitively, $w(e')$ measures the influence of $\Color(e')$ on $z$.
Notice that $\sum_{e' \in S'}{w(e')} \leq 2|S|pt$. 

We consider the sequence of random variables $(X_1, \ldots, X_{|S|+|S'|})$, where the initial $|S'|$  variables are the colors selected by the edges in $S'$, in arbitrary order, and the remaining $|S|$  variables are the colors selected by the edges in $S$, in arbitrary order.
We let $z = f(X_1, \ldots, X_{|S|+|S'|})$ in Theorem~\ref{lem:concentration}.
To prove the desired concentration bound,
it suffices to show that we can set $M = O(1)$ and $\sigma_i^2$ to achieve $\sum_{i=1}^{|S|+|S'|} \sigma_i^2  = O(|S|)$.
In what follows, we analyze the effect of exposing the value of the random variable $X_i$, given that all  variables in $\mathbf{X}_{i-1}$ have been fixed.

\paragraph{Exposing an Edge in $S'$.}
Consider the case where $X_i = \Color(e^\star)$ is the color selected by the edge $e^\star \in S'$.
Recall $D_i = \Expect[z|\mathbf{X}_i] - \Expect[z|\mathbf{X}_{i-1}]$.
Our goal is to show that $\Var[D_i | \mathbf{X}_{i-1}] = O(w(e)/(pt))$ and $|D_i| = O(1)$.
Hence we set $\sigma_i^2 = O(w(e)/(pt))$, which implies $\sum_{1\leq i \leq |S'|} \sigma_i^2 = O(|S|)$, as desired.

By linearity of expectation,  $D_i = \sum_{c} (\Expect[z_c | \mathbf{X}_i] - \Expect[z_c | \mathbf{X}_{i-1}])$,
where the summation ranges over all colors $c$ that appear in $\bigcup_{e \in S}\Psi(e)$.
We write $D_{i,c} = \Expect[z_c | \mathbf{X}_i] - \Expect[z_c | \mathbf{X}_{i-1}]$, and  make the following observations:
\begin{itemize}
\item $D_{i,c} \neq 0$ only if $c \in \Psi(e^\star)$. For each $c \in \Psi(e^\star)$, $D_{i,c}$ depends only on whether $e^\star$ selects the color $c$, which occurs with probability $1/p$. In particular, $D_{i,c} < 0$ only if $e^\star$ selects $c$, and $D_{i,c} > 0$ only if $e^\star$ does not select $c$. Thus, $\Cov[D_{i,c},D_{i,c'} | \mathbf{X}_{i-1}] \leq 0$ for all color pairs $\{c, c'\}$.
\item For each $e \in S$, both $\Expect[z_{e,c}|\mathbf{X}_i]$ and $\Expect[z_{e,c}|\mathbf{X}_{i-1}]$ are within $[0, 1/p]$, since $z_{e,c}=1$ only if $c \in \Psi(e)$ and $e$ selects $c$, which occurs with probability $1/p$. Thus, $\max_{X_i} D_{i,c} - \min_{X_i} D_{i,c} \leq w(e^\star, c)/p$.
\end{itemize}
By Lemma~\ref{lem:aux} (with $k \leq w(e^\star, c)/p$ and $\alpha = 1/p$), we have $\Var[D_{i,c} | \mathbf{X}_{i-1}] \leq (1/p)(w(e^\star,c)/p)^2$. We bound the variance $\Var[D_i | \mathbf{X}_{i-1}]$ as follows.
\begin{align*}
\Var[D_i | \mathbf{X}_{i-1}] &= \sum_c \Var[D_{i,c} | \mathbf{X}_{i-1}] + \sum_{c, c'} \Cov[D_{i,c},D_{i,c'} | \mathbf{X}_{i-1}] \\
&= \sum_c O((w(e^\star,c)/p)^2 / p) &\Cov[D_{i,c},D_{i,c'} | \mathbf{X}_{i-1}]\leq 0 \\
&= \sum_c O(w(e^\star,c)/p^2) &w(e^\star,c)<t=\Theta(p) \\
&= O(w(e^\star)/p^2) \\
&= O(w(e^\star)/(pt)).
\end{align*}
We  bound  $|D_i|$ as follows.
Consider $c \in \Psi(e^\star)$.
Recall that we already have the bound $|D_{i,c}| \leq w(e^\star,c)/p \leq (t-1)/p$.
If $c$ is not selected by $e^\star$, which occurs with probability $1 - 1/p$,  we have a tighter bound $|D_{i,c}|\leq w(e^\star,c)/p^2 \leq (t-1)/p^2$ by Lemma~\ref{lem:aux} with $k \leq w(e^\star, c)/p$ and $\alpha = 1/p$.
Therefore,
$$|D_i| \leq \sum_c|D_{i,c}| \leq 1 \cdot \frac{t-1}{p} + (p-1)\cdot \frac{t-1}{p^2}  = O(1).$$

\paragraph{Exposing an Edge in $S$.}
Consider the case where $X_i = \Color(e^\star)$ is the color selected by the edge $e^\star \in S$.
Suppose that $X_i = c^\star$.
Recall $D_i = \sum_{c} D_{i,c}$.
It is straightforward to see that (i) $|D_{i,c}| \leq 1$ if $c = c^\star$, (ii) $|D_{i,c}| \leq 1/p$ if $c \in \Psi(e^\star) - \{c^\star\}$, and (iii) $|D_{i,c}| = 0$ otherwise.
Thus, $|D_i| = O(1)$, and  $\Var[D_i | \mathbf{X}_{i-1}] = O(1)$.
We set $\sigma_i^2 = O(1)$, and so $\sum_{|S'| < i \leq |S| + |S'|} \sigma_i^2 = O(|S|)$.

\subsection{Concentration of Palette Size}
Let $e^\bullet=\{u,v\}$ be an edge, and let $c^\bullet =\Color(e^\bullet)$ be the color selected by $e^\bullet$.
We do not consider $c^\bullet$ as a random variable in the analysis (i.e., we expose the color selected by $e^\bullet$ first).
Let $\mathcal{E}$ be the event that $e^\bullet$ does not successfully color itself.
Since $e^\bullet$ remains uncolored with at least a constant probability,
we are allowed to ignore the condition ``$e^\bullet$ remains uncolored'' in Lemma~\ref{thm:concentration} in the subsequent calculation.
To prove the desired concentration bound regarding palette size $\Prob\left[|\Psi^\diamond(e)| < (1-\delta) p^\diamond \ | \ e \text{ remains uncolored } \right]   =  \exp\left(-\Omega(\delta^2 p)\right)$, it suffices to show that
(i) $|\Expect[|\Psi^\diamond(e^\bullet)|] - p^\diamond| = O(1)$, and (ii) $\Prob[|\Psi^\diamond(e^\bullet)| < (1-\delta)\Expect[|\Psi^\diamond(e^\bullet)|]] = \exp(-\Omega(\delta^2 \Expect[|\Psi^\diamond(e^\bullet)|]))$.

\paragraph{Notations.}
We write $S_u$ (resp., $S_v$) to denote the set of edges $e$ incident to $e^\bullet$ on $u$ (resp., $v$) such that $\Psi(e) \cap \Psi(e^\bullet) - \{c^\bullet\} \neq \emptyset$.
We write $S'$ to denote the set of edges such that $e' \in S'$ if there exists $e \in S_u \cup S_v$ meeting the following conditions: (i) $e'$ is incident to $e$, (ii) $e' \notin S_u \cup S_v \cup \{e^\bullet\}$, and (iii) $\Psi(e) \cap \Psi(e') \cap \Psi(e^\bullet) - \{c^\bullet\} \neq \emptyset$.
Notice that $\Psi^\diamond(e^\bullet)$ is determined by the colors selected by the edges in $S_u \cup S_v \cup S'$.
We have $|S_u| \leq (p-1)(t-1) < pt$, $|S_v| \leq (p-1)(t-1) < pt$, and $|S'| \leq 2(p-1)(t-1)^2 < 2pt^2$.

\paragraph{Expected Value.}
In what follows, consider a color $c \in \Psi(e^\bullet) -\{c^\bullet\}$.
\begin{itemize}
\item Let $e \in S_u \cup S_v$ such that $c \in \Psi(e)$.  We have $\Expect[z_{e,c}] = \frac{1}{p}(1-\frac{1}{p})^{2t-3}$. Notice that $e^\bullet$ selects $c^\bullet \neq c$, so there are $2t-3$ (rather than $2t-2$) edges competing with $e$ for the color $c$.
\item Let $e'=\{u,x\} \in S_u$ and $e''=\{v,y\} \in S_v$ such that $c \in \Psi(e') \cap \Psi(e'')$.
 We define $z_{e',e'',c} \bydef z_{e',c} \cdot z_{e'', c}$.
 If $x=y$, then $z_{e', e'',c} = 0$.
 Otherwise, $x \neq y$ and $\Expect[z_{e', e'',c}] = \frac{1}{p^2}(1-\frac{1}{p})^{4t-6-b(e',e'')}$, where
 $b(e',e'') \leq 3$ is the number of edges $e$ such that (i) $e \neq e^\bullet$, and (ii) $e$ is incident to both $e'$ and $e''$.
\end{itemize}
Let $z_c$  be the indicator random variable that some edge incident to $e^\bullet$ successfully colors itself by $c$, that is,
$$z_c \bydef \sum_{e \ : \ e \in S_u \cup S_v, \ c \in \Psi(e)} z_{e, c} - \sum_{e', e'' \ : \ e' \in S_u, \ e'' \in S_v, \ c \in \Psi(e') \cap \Psi(e'')} z_{e',e'',c}.$$
The number of edges $e \in S_u \cup S_v$ such that $c \in \Psi(e)$ is exactly $2t-2$.
The number of pairs $(e'=\{u,x\} \in S_u, e''=\{v,y\} \in S_v)$ such that $c \in \Psi(e') \cap \Psi(e'')$ and $x \neq y$ is at least $(t-1)^2 - (t-1)$ and at most $(t-1)^2$. By linearity of expectation (recall $t = \Theta(p)$),
$$\Expect[z_{c}] = \frac{2t}{p}(1 - 1/p)^{2t} - \frac{t^2}{p^2}(1 - 1/p)^{4t} \pm O(1/p).$$
Define $z \bydef \sum_{c \in \Psi(e^\bullet) -\{c^\bullet\}} z_c$. Then, we have:
\begin{align*}
\Expect[|\Psi^\diamond(e^\bullet)|]
&= |\Psi(e^\bullet)| - \Expect[z]  &|\Psi^\diamond(e^\bullet)| = |\Psi(e^\bullet)| - z\\
&= p \cdot \left(1 - \frac{2t}{p}(1 - 1/p)^{2t} + \frac{t^2}{p^2}(1 - 1/p)^{4t} \pm O(1/p) \right) \\
&= p \cdot \left(1 - \frac{2t}{p}(1 - 1/p)^{2t} + \frac{t^2}{p^2}(1 - 1/p)^{4t}\right)  \pm O(1)  \\
&= p^\diamond \pm O(1). &\text{Definition of $p^\diamond$}
\end{align*}
Hence $|\Expect[|\Psi^\diamond(e^\bullet)|] - p^\diamond| = O(1)$.

\paragraph{Concentration Bound.}
Consider the sequence of random variables $(X_1, \ldots, X_{|S_u|+|S_v|+|S'|})$, where the initial $|S'|$  variables are the colors selected by the edges in $S'$, in arbitrary order, and the remaining $|S_u|+|S_v|$  variables are the colors selected by the edges in $S_u \cup S_v$, in arbitrary order.
Let $z = f(X_1, \ldots, X_{|S_u|+|S_v|+|S'|})$ in Theorem~\ref{lem:concentration}.
To prove the desired concentration bound
$\Prob[|\Psi^\diamond(e^\bullet)| < (1-\delta)\Expect[|\Psi^\diamond(e^\bullet)|]] = \exp(-\Omega(\delta^2 \Expect[|\Psi^\diamond(e^\bullet)|]))$,
it suffices to show that $\Prob[z > \Expect[z] + s] = \exp\left(-\Omega(s^2 / p)\right)$, by setting $s = \delta \Expect[|\Psi^\diamond(e^\bullet)|]$, and recall that $\Expect[|\Psi^\diamond(e^\bullet)|] =  p^\diamond \pm O(1) = \Theta(p)$.
In view of Theorem~\ref{lem:concentration}, we only need to show that we can set $M = O(1)$ and $\sigma_i^2$ such that $\sum_{i=1}^{|S_u|+|S_v|+|S'|} \sigma_i^2 = O(p)$.

\paragraph{Exposing an Edge in $S'$.}
Consider the case where $X_i = \Color(e^\star)$ is the color selected by the edge $e^\star \in S'$.
Our goal is to show that $|D_i| = O(1/t)$. This implies $\Var[D_i|\mathbf{X}_{i-1}] = O(1/t^2)$, and so we may set $\sigma_i^2 = O(1/t^2)$.
Since $|S'| = O(pt^2)$, we have $\sum_{i=1}^{|S'|} \sigma_i^2 = O(p)$.

Let $R$ denote the set of edges in $S_u \cup S_v$ that are incident to $e^\star$. Notice that $1 \leq |R| \leq 2$.
We define:
$$z_c^{(i)} \bydef \sum_{e'\ : \ e' \in  R, \ c \in \Psi(e')} z_{e', c} - \sum_{e', e''  \ : \ e' \in S_u, \ e'' \in S_v, \ c \in \Psi(e') \cap \Psi(e''), \  \{e, e''\} \cap R \neq \emptyset} z_{e',e'',c}.$$
Intuitively, $z_c^{(i)}$ is the result of subtracting all terms from the definition of $z_c$ not involving edges in $R$.
We now argue that
$\Expect[z_c|\mathbf{X}_i] - \Expect[z_c|\mathbf{X}_{i-1}] = \Expect[z_c^{(i)}|\mathbf{X}_i] - \Expect[z_c^{(i)}|\mathbf{X}_{i-1}]$.
This is due to the two observations:
(i) If $e \notin R$, then $\Expect[z_{e,c}|\mathbf{X}_i] = \Expect[z_{e,c}|\mathbf{X}_{i-1}]$.
(ii) If $\{e', e''\} \cap R = \emptyset$, then $\Expect[z_{e', e'', c}|\mathbf{X}_i] = \Expect[z_{e', e'', c}|\mathbf{X}_{i-1}]$.

Consider a color $c \in \Psi(e^\star) \cap \Psi(e^\bullet) - \{c^\bullet\}$.
The probability that some edge in $R$ selects $c$ is at most $|R|/p \leq 2/p$. Thus, the conditional expectations $\Expect[z_c^{(i)}|\mathbf{X}_i]$ and $\Expect[z_c^{(i)}|\mathbf{X}_{i-1}]$ must be within $[0, 2/p]$, and so $|\Expect[z_c^{(i)}|\mathbf{X}_i] - \Expect[z_c^{(i)}|\mathbf{X}_{i-1}]| \leq 2/p$.
For the case of $c \neq X_i$, which occurs with probability $1 - 1/p$, we have a tighter bound
$|\Expect[z_c^{(i)}|\mathbf{X}_i] - \Expect[z_c^{(i)}|\mathbf{X}_{i-1}]| \leq 2/p^2$
by Lemma~\ref{lem:aux} with $k \leq 2/p$ and $\alpha = 1/p$. We bound $|D_i|$ as follows.
\begin{align*}
|D_i| &\leq \sum_{c \in \Psi(e^\bullet) - \{c^\bullet\}} |\Expect[z_c|\mathbf{X}_i] - \Expect[z_c|\mathbf{X}_{i-1}]| \\
&= \sum_{c \in \Psi(e^\star) \cap \Psi(e^\bullet) - \{c^\bullet\}} |\Expect[z_c^{(i)}|\mathbf{X}_i] - \Expect[z_c^{(i)}|\mathbf{X}_{i-1}]| \\
&\leq (2/p) + (2/p^2)(|\Psi(e^\star) \cap \Psi(e^\bullet) - \{c^\bullet\}|-1)\\
&= O(1/p) = O(1/t).
\end{align*}

\paragraph{Exposing an Edge in $S_u \cup S_v$.}
Consider the case where $X_i = \Color(e^\star)$ is the color selected by the edge $e^\star \in S_u \cup S_v$.
We define $w(e^\star) \bydef |\Psi(e^\star) \cap \Psi(e^\bullet) - \{c^\bullet\}|$.
The goal is to show that (i) $|D_i| = O(1)$ and (ii) $\Var[D_i|\mathbf{X}_{i-1}] = O(w(e^\star) /p)$.
By setting $\sigma_i^2 = O(w(e^\star) /p)$, we achieve
$$\sum_{i=|S'|+1}^{|S'|+|S_u|+|S_v|} \sigma_i^2 = \sum_{e \in S_u \cup S_v} O(w(e) /p) = O(pt / p) = O(t) = O(p).$$

By the linearity of expectation,  $D_i = \sum_{c \in \Psi(e^\star) \cap \Psi(e^\bullet) - \{c^\bullet\}} D_{i,c}$, where $D_{i,c} = \Expect[z_c | \mathbf{X}_i] - \Expect[z_c | \mathbf{X}_{i-1}]$.
Since both $\Expect[z_c | \mathbf{X}_i]$ and $\Expect[z_c | \mathbf{X}_{i-1}]$ are within $[0,1]$, we have  $|D_{i,c}| \leq 1$. We have a tighter bound $|D_{i,c}| \leq 1/p$ in the event that $\Color(e^\star) \neq c$ (by Lemma~\ref{lem:aux} with $k \leq 1$ and $\alpha = 1/p$). Thus, $|D_i| \leq 1 + (w(e^\star)-1)/p = O(1)$.

In order to prove that $\Var[D_i|\mathbf{X}_{i-1}]  = O(w(e^\star) /p)$, we need the following two observations.
\begin{itemize}
\item Consider a color $c \in  \Psi(e^\star) \cap \Psi(e^\bullet) - \{c^\bullet\}$.
Recall that $|D_{i,c}| \leq 1/p$ for the case $c$ is not selected by $e^\star$, which occurs with probability $1-1/p$.
Thus, $\Expect[D_{i,c} \cdot D_{i,c} | \mathbf{X}_{i-1}] \leq (1/p) \cdot 1 + (1-1/p) \cdot 1/p^2 = O(1/p)$.
\item Consider two distinct colors $c$ and $c'$ in $\Psi(e^\star) \cap \Psi(e^\bullet) - \{c^\bullet\}$.
If $e^\star$ selects $c$ or $c'$ (which occurs with probability $2/p$), $D_{i,c} \cdot D_{i,c'}\leq 1 \cdot (1/p)$.
Otherwise $D_{i,c} \cdot D_{i,c'}\leq (1/p) \cdot (1/p)$.
Therefore, $\Expect[D_{i,c} \cdot D_{i,c'} | \mathbf{X}_{i-1}] \leq (2/p) \cdot 1/p + (1-2/p) \cdot 1/p^2 = O(1/p^2)$.
\end{itemize}
We now bound $\Var[D_i|\mathbf{X}_{i-1}]$ as follows.
\begin{align*}
\Var[D_i|\mathbf{X}_{i-1}] &\leq \sum_{c \in \Psi(e^\star) \cap \Psi(e^\bullet) - \{c^\bullet\}} \; \sum_{c' \in \Psi(e^\star) \cap \Psi(e^\bullet) - \{c^\bullet\}} \Expect[D_{i,c} \cdot D_{i,c'} | \mathbf{X}_{i-1}]\\
&\leq w(e^\star)\cdot O(1/p) + w(e^\star)(w(e^\star)-1)\cdot O(1/p^2)\\
&=O(w(e^\star) /p).
\end{align*}

\subsection{Concentration of Color Degree}

For the remainder of this section, fix a vertex
$v^\bullet$ and a color $c^\bullet$ in the palette $\Psi(e)$ for some $e$ incident to $v^\bullet$.
For convenience, we write $R \bydef N_{c^\bullet}(v^\bullet)$.
Define $R^\diamond$ as the subset of $R$ such that $e=\{v^\bullet,u\} \in R^\diamond$ if (i) $e$ is not successfully colored by a color in $\Psi(e) - \{c^\bullet\}$, and
(ii) no edge incident to $e$ on $u$ successfully colors itself  $c^\bullet$. We write $z \bydef |R  \setminus R^\diamond|$.
Let $\mathcal{E}'$ be the event that $N_{c^\bullet}^\diamond(v^\bullet) \neq \emptyset$.
Observe that if $\mathcal{E}'$ occurs, then no edge incident to $v^\bullet$ successfully colors itself  $c^\bullet$. Thus, conditioning on $\mathcal{E}'$ happening,  $R \setminus R^\diamond$ equals $N_{c^\bullet}^\diamond(v^\bullet)$.

Our goal is to show that (i) $\Prob[ z < \Expect[z] - s ] = \exp(-\Omega(s^2/t))$, and (ii) $\Expect[|R^\diamond|] = |R| - \Expect[z] = t^\diamond \pm O(1)$.
Since $\mathcal{E}'$ occurs with constant probability, the above (i) and (ii) together imply the desired concentration bound
$\Prob[|N_{c^\bullet}^\diamond(v^\bullet)| > (1+\delta)t^\diamond \ | \ \mathcal{E}'] = \exp(-\Omega(\delta^2 t))$, by setting $s = \delta t^\diamond \pm O(1)$.  Recall that $t^\diamond = \Theta(t)$.

\paragraph{Expected Value.}
With respect to an edge $e=\{v^\bullet,u\} \in R$, we define the following notations based on parts (i) and (ii) of
the definition of $R^\diamond$.
\begin{itemize}
\item Define $z_{e}^a$ as the indicator random variable that some edge incident to $e$ on $u$ successfully colors itself $c^\bullet$.
We have $\Expect[z_{e}^a] = (t-1)\cdot \frac{1}{p}(1-\frac{1}{p})^{2t-2} = \frac{t}{p}(1-\frac{1}{p})^{2t} \pm O(1/p)$.
\item Define $z_{e}^b$ as the indicator random variable that $e$ is successfully colored by a color in $\Psi(e) - \{c^\bullet\}$.
We have $\Expect[z_{e}^b] = (p-1)\cdot \frac{1}{p}(1-\frac{1}{p})^{2t-2} = (1-\frac{1}{p})^{2t} \pm O(1/p)$.
\end{itemize}
Let $z_{e}^{a,b} \bydef z_e^a \cdot z_e^b$.
Notice that $z_e^a$ and $z_e^b$ are nearly independent but not independent.
Let $z_e \bydef z_e^a + z_e^b - z_e^{a,b}$, and so we have $z = |R \setminus R^\diamond| = \sum_{e \in R} z_e$.
We calculate $\Expect[z_{e}^{a,b}]$ as follows.
Let $e'$ be any edge incident to $e$ such that $c^\bullet \in \Psi(e')$,
and let $c$ be any color in $\Psi(e) - \{c^\bullet\}$.
With respect to $(e, e', c)$, we define the following two sets:
\begin{itemize}
\item $S_a$ is the set of all edges $e''$ such that (i) $e'' \neq e, e'$, (ii) $e''$ is incident to $e'$, and (iii) $c^\bullet \in \Psi(e'')$. Intuitively, $S_a$ is the set of all edges other than $e$ that contend with $e'$ for the color $c^\bullet$.
Notice that $|S_a| = 2t-3$, since $\Psi(e)$ must contain $c^\bullet$.
\item $S_b$ is the set of all edges $e''$ such that $e'' \in S_b$ if (i) $e'' \neq e, e'$, (ii) $e''$ is incident to $e$, and (iii) $c \in \Psi(e'')$. Intuitively, $S_b$ is the set of all edges other than $e'$ that contend with $e$ for the color $c$.
Notice that $2t-3 \leq |S_b| \leq 2t-2$, since $\Psi(e')$ may or may not contain $c$.
The extent to which $S_a$ and $S_b$ intersect is unknown.
\end{itemize}
Fixing the edge $e$ incident to $v^\bullet$, let $x(c,e')$ denote the probability that
(i) $e'$ successfully colors itself  $c^\bullet$ and
(ii) $e$ successfully colors itself  $c$.
In view of the definition of $S_a$ and $S_b$, we have:
\begin{align*}
x(c,e') & = \frac{1}{p^2}
\prod_{e'' \in S_a \setminus S_b}(1-1/p)
\prod_{e'' \in S_b \setminus S_a}(1-1/p)
\prod_{e'' \in S_a \cap S_b}(1-2/p)\\
& = \frac{1}{p^2} (1-1/p)^{|S_a \setminus S_b|}
(1-1/p)^{|S_b \setminus S_a|}
(1-2/p)^{|S_a \cap S_b|}\\
& = \frac{1}{p^2} (1-1/p)^{|S_a \setminus S_b|}
(1-1/p)^{|S_b \setminus S_a|}
(1-1/p)^{2|S_a \cap S_b|} \left(1 - O\left(\frac{|S_a \cap S_b|}{p^2}\right)\right)\\
&= \frac{1}{p^2} (1-1/p)^{|S_a| + |S_b|} (1 - O(1/p))  \hcm \mbox{(Notice that $|S_a \cap S_b| < t = \Theta(p)$.)}\\
&= \frac{1}{p^2} (1-1/p)^{4t - O(1)} (1 - O(1/p))\\
& = \frac{1}{p^2}  (1-{1/p})^{4t} \pm O(1/p^3).
\end{align*}
We now calculate $\Expect[z_{e}^{a,b}]$ and show that $\Expect[|R^\diamond|]  = |R| - \Expect[z] = t^\diamond \pm O(1)$.
\begin{align*}
\Expect[z_{e}^{a,b}]
&= \sum_{\substack{(c,e') \ : \ \text{$e'$ incident to $e$,}\\ \text{$c^\bullet \in \Psi(e')$, $c \in \Psi(e) - \{c^\bullet\}$}}} x(c,e') & \mbox{(union of disj.~events)}\\
&=(t-1)(p-1)\cdot \left( \frac{1}{p^2}  (1-1/p)^{4t} \pm O(1/p^3) \right)\\
&=\frac{t}{p}(1-1/p)^{4t} \pm O(1/p). \\
\displaybreak[0]\\
\Expect[|R^\diamond|] &= |R| - \Expect[z] \\
& = t - \sum_{e \in R} \left( \Expect[z_e^a] + \Expect[z_e^b] - \Expect[z_e^{a,b}] \right)\\
& = t \cdot \left(1 - \frac{t}{p}(1-1/p)^{2t} - (1-1/p)^{2t} + \frac{t}{p}(1-1/p)^{4t} \pm O(1/p)\right)\\
& = t \cdot \left(1 - \frac{t}{p}(1-1/p)^{2t} - (1-1/p)^{2t} + \frac{t}{p}(1-1/p)^{4t}\right) \pm O(1)\\
& = t^\diamond \pm O(1). &\text{Definition of $t^\diamond$}
\end{align*}

\paragraph{Concentration Bound.}
We have established that $|R^\diamond|$ has the correct expectation and now need to prove that it
has sufficiently good concentration around that expectation.  The analysis here
becomes more complicated because we have to consider the colors selected in some 3-neighborhood.
The \emph{palette size} and \emph{degree} analyses focussed only on 2-neighborhoods.

Based on the definition of $z_e^a$ and $z_e^b$, we define the following sets.
\begin{itemize}
\item Recall that $R=N_{c^\bullet}(v^\bullet)$.  Let $R_1$ be the set of  all edges $e$
such that (i) $e \notin R$, (ii) $c^\bullet \in \Psi(e)$, and (iii) $e$ is incident to some edge in $R$.
Similarly, let $R_2$ be the set of  all edges $e$
such that (i) $e \notin R \cup R_1$, (ii) $c^\bullet \in \Psi(e)$, and (iii) $e$ is incident to some edge in $R_1$.
Notice that the value $z_e^a$, for any $e \in R$, is determined by the information about which edges in $R \cup R_1 \cup R_2$ select $c^\bullet$.
We write $\alpha = |R \cup R_1 \cup R_2|$.
\item
Let $R'$ be the set of  all edges $e'$
such that (i) $e' \notin R$  and (ii) there exists $e \in R$ such that $\Psi(e) \cap \Psi(e') - \{c^\bullet\} \neq \emptyset$.
Notice that the the value $z_e^b$, for any $e \in R$, is determined by the colors selected by the edges in $R \cup R'$.
We write $\beta = |R \cup R'|$.
\end{itemize}

For each $e \in R$,  $z_e^a$ is simply the summation of $z_{e',c^\bullet}$ over all edges $e' \in R_1$ incident to $e$.
For each $e'' \in R_2$, we write $w(e'')$ to denote the number of edges in $R_1$ incident to $e''$.
Intuitively, $w(e'')$ measures the influence of $\Color(e'')$ on $\sum_{e \in R} z_{e}^a$.

We consider the sequence of random variables $(X_1, \ldots, X_{\alpha+\beta})$, where the initial $\alpha$ random variables reveal which edges in $R \cup R_1 \cup R_2$ select the color $c^\bullet$ according to the ordering $R_2, R_1, R$, and the remaining $\beta$ random variables reveal the colors selected by the edges in $R \cup R'$ according to the ordering $R', R$.
We let $z = f(X_1, \ldots, X_{\alpha+\beta})$ in Theorem~\ref{lem:concentration}.
To prove the desired concentration bound $\Prob[ z < \Expect[z] - s ] = \exp(-\Omega(s^2/t))$,
it suffices to show that we can set $M = O(1)$ and $\sigma_i^2$ such that $\sum_{i=1}^{\alpha+\beta} \sigma_i^2 = O(t)$.
In what follows, we analyze the effect of exposing the value of $X_i$, given that all variables in $\mathbf{X}_{i-1}$ have been fixed.

\paragraph{Revealing whether $c^\bullet$ is Selected by an Edge in $R \cup R_1 \cup R_2$.}
Consider the case where $X_i$ reveals whether $c^\bullet$ is selected by the edge $e^\star \in R \cup R_1 \cup R_2$.
Notice that $X_i$ is binary, and recall that $D_i = \Expect[z|\mathbf{X}_i] - \Expect[z|\mathbf{X}_{i-1}]$.
There are at most two distinct outcomes of $D_i|\mathbf{X}_{i-1}$, in which one occurs with probability $1/p$.
Thus, by Lemma~\ref{lem:aux} we have:
\[
\Var[D_i|\mathbf{X}_{i-1}] \leq \left(\max_{X_i}D_i|\mathbf{X}_{i-1} - \min_{X_i}D_i|\mathbf{X}_{i-1}\right)^2 / p =  O( \max_{\mathbf{X}_i} |D_i|^2 / p).
\]
Thus, to achieve $\sum_{i=1}^{\alpha} \sigma_i^2 = O(t)$ and $M=O(1)$ it suffices to show the following.
\begin{itemize}
\item For the case $e^\star \in R_2$, we must prove $|D_i| = O(w(e^\star)/p)$.\footnote{Intuitively, if $e^\star$ chooses
color $c^\bullet$, it prevents $w(e^\star)$ edges in $R_1$ from successfully coloring themselves $c^\bullet$,
but the prior probability of these edges coloring themselves $c^\bullet$ was only $O(1/p)$, hence the total influence
on the expectation of $z$ should be $O(w(e^\star)/p)$.}
Since $w(e^\star) < t = \Theta(p)$,
$\Var[D_i|\mathbf{X}_{i-1}] = O((w(e^\star)/p)^2 / p) = O(w(e^\star)/p^2)$,
so we can set $\sigma_i^2 = O(w(e^\star)/p^2)$.
\item For  the case $e^\star \in R \cup R_1$, we must prove $|D_i| = O(1)$.
Hence we may set $\sigma_i^2 = \Var[D_i|\mathbf{X}_{i-1}] = O(1/p)$.

\end{itemize}
Notice that $\sum_{e^\star \in R_2} w(e^\star) < t^3$, $|R_1|<t^2$, and $|R|=t$. Thus, $\sum_{i=1}^{\alpha} \sigma_i^2 = O(t)$.
With respect to the edge $e^\star \in R \cup R_1 \cup R_2$, we make the following definitions.
\begin{align*}
Y^a &\bydef \{e' \in R_1 \ : \ e' = e^\star \text{ or $e'$ is incident to $e^\star$}\}
& D_i^{a} &\bydef \sum_{e' \in Y^a} \left(\Expect[z_{e',c^\bullet}|\mathbf{X}_i] + \Expect[z_{e',c^\bullet}|\mathbf{X}_{i-1}]\right) \\
Y^b &\bydef \{e \in R \ : \ e = e^\star \text{ or $e$ is incident to $e^\star$}\}
& D_i^{b} &\bydef \sum_{e \in Y^b} |\Expect[z_{e}^b|\mathbf{X}_i] - \Expect[z_{e}^b|\mathbf{X}_{i-1}]|
\end{align*}
Intuitively, $Y^a$ and $Y^b$ are the subsets of $R_1$ and $R$ that are ``relevant'' to $D_i$ in the following sense:
\begin{align*}
\Expect[z_{e'',c^\bullet}|\mathbf{X}_i] 	&= \Expect[z_{e'',c^\bullet}|\mathbf{X}_{i-1}]	& \mbox{for all $e'' \in R_1 \setminus Y^a$,}\\
\Expect[z_{e'}^b|\mathbf{X}_i] 			&= \Expect[z_{e'}^b|\mathbf{X}_{i-1}]			& \mbox{for all $e' \in R \setminus Y^b$.}
\end{align*}
Our plan of bounding $|D_i|$ is as follows. First we show that $|D_i|  \leq 4 D_i^{a} + D_i^b$ in Claim~\ref{clm-1}, and then we bound $D_i^{a}$ and $D_i^b$ separately in Claims~\ref{clm-2} and~\ref{clm-3}. The three claims together establish a desired bound on $|D_i|$.

\begin{claim}\label{clm-1}
$|D_i|  \leq 4 D_i^{a} + D_i^b$.
\end{claim}
\begin{proof}
We define the following notations.
\[
\begin{matrix*}[l]
&P_1 \bydef \{(e, e') \ : e \in R \setminus Y^b, e' \in Y^a, \text{ $e$ is incident to $e'$}\} \\
&P_2 \bydef \{(e, e') \ : e \in Y^b, e' \in R_1 \setminus  Y^a, \text{ $e$ is incident to $e'$}\} \\
&P_3 \bydef \{(e, e') \ : e \in Y^b, e' \in Y^a, \text{ $e$ is incident to $e'$}\} \\
\\
& Q_j \bydef -\sum_{(e, e') \in P_j}
\left(\Expect[z_{e',c^\bullet} \cdot z_{e}^b|\mathbf{X}_i] - \Expect[z_{e',c^\bullet} \cdot z_{e}^b|\mathbf{X}_{i-1}]\right)
& \text{(for each $j = 1,2,3$)} \\
& F_j \bydef \sum_{e \in R} \left(\Expect[z_e^j|\mathbf{X}_i] - \Expect[z_e^j|\mathbf{X}_{i-1}]\right)
& \text{(for each $j = a,b$)}
\end{matrix*}
\]
The definitions of $P_1$, $P_2$, and $P_3$ depend on $Y^a$ and $Y^b$, which depend on the edge $e^\star$.
For instance, if $e^\star \in R$, then $Y^b = R$, which implies that $P_1 = \emptyset$.
Recall that the edge $e^\star$ can be any edge in $R \cup R_1 \cup R_2$, and the proof of this claim applies to all choices of $e^\star \in R \cup R_1 \cup R_2$.

Notice that for any pair $(e\in R, e'\in R_1)$ such that $e$ is incident to $e'$ but $(e, e') \notin P_1 \cup P_2 \cup P_3$, we must have
$\Expect[z_{e',c^\bullet} \cdot z_{e}^b|\mathbf{X}_i] = \Expect[z_{e',c^\bullet} \cdot z_{e}^b|\mathbf{X}_{i-1}]$ due to the definition of $Y^a$ and $Y^b$.
We rewrite the term $D_i$ as follows.
\begin{align*}
D_i
&=  \Expect[z|\mathbf{X}_i] - \Expect[z|\mathbf{X}_{i-1}] \\
&= \sum_{e \in R}\left( \Expect[z_e|\mathbf{X}_i] - \Expect[z_e|\mathbf{X}_{i-1}] \right)\\
&= \sum_{e \in R}\left( \Big(\Expect[z_e^a|\mathbf{X}_i] - \Expect[z_e^a|\mathbf{X}_{i-1}]\Big)
+  \left(\Expect[z_e^b|\mathbf{X}_i] - \Expect[z_e^b|\mathbf{X}_{i-1}]\right)  - \left(\Expect[z_e^a \cdot z_{e}^b|\mathbf{X}_i] - \Expect[z_e^a \cdot z_{e}^b|\mathbf{X}_{i-1}]\right)  \right)\\
\intertext{(Recall that $z_e^a$ is the summation of $z_{e',c^\bullet}$ over all edges $e' \in R_1$ incident to $e$.)}
&= F_a + F_b -\sum_{(e, e') \ : \ e \in R, \ e' \in R_1, \ e' \text{ incident to }e}
\left(\Expect[z_{e',c^\bullet} \cdot z_{e}^b|\mathbf{X}_i] - \Expect[z_{e',c^\bullet} \cdot z_{e}^b|\mathbf{X}_{i-1}]\right)\\
\intertext{(Any pair $(e, e') \notin P_1 \cup P_2 \cup P_3$ contributes zero to this summation.)}
&= F_a + F_b + Q_1 + Q_2 + Q_3.
\end{align*}
To prove this claim it suffices to show that
(i) $|F_a + Q_1| \leq 2 D_i^a$, (ii) $|F_b + Q_2| \leq D_i^b$, and (iii) $|Q_3| \leq 2 D_i^a$.
We expand $F_a$ using the fact that $z_e^a$ is the summation of $z_{e',c^\bullet}$ over all edges $e' \in R_1$ incident to $e$.
\begin{align*}
|F_a + Q_1|
&\leq \left|Q_1   + \sum_{(e, e') \ : \ e \in R, \ e' \in R_1, \ e' \text{ incident to }e}
\left(\Expect[z_{e',c^\bullet}|\mathbf{X}_i] - \Expect[z_{e',c^\bullet}|\mathbf{X}_{i-1}]\right)  \right| \\
\intertext{Since any pair $(e, e') \notin P_1 \cup P_3$ contributes 0 in the summation,}
&\leq \left|Q_1  + \sum_{(e, e') \in P_1 \cup P_3}
\left(\Expect[z_{e',c^\bullet}|\mathbf{X}_i] - \Expect[z_{e',c^\bullet}|\mathbf{X}_{i-1}]\right)\right| \\
\intertext{and by definition of $Q_1$,}
&\leq  \sum_{(e,e') \in P_1} \left|\Expect[z_{e',c^\bullet}(1-z_{e}^b)|\mathbf{X}_i] - \Expect[z_{e',c^\bullet}(1-z_{e}^b)|\mathbf{X}_{i-1}]\right|\\
& \hcm +\sum_{(e,e') \in P_3} \left|\Expect[z_{e',c^\bullet}|\mathbf{X}_i] - \Expect[z_{e',c^\bullet}|\mathbf{X}_{i-1}]\right|
\\
\intertext{When $e \notin R \setminus Y^b$, $\Expect[z_{e}^b|\mathbf{X}_{i-1}] = \Expect[z_{e}^b|\mathbf{X}_{i}]$, so}
&\leq
\sum_{(e,e') \in P_1} (1-\Expect[z_{e}^b|\mathbf{X}_{i-1}]) \left|\Expect[z_{e',c^\bullet}|\mathbf{X}_i] - \Expect[z_{e',c^\bullet}|\mathbf{X}_{i-1}]\right|\\
& \hcm +\sum_{(e,e') \in P_3} \left|\Expect[z_{e',c^\bullet}|\mathbf{X}_i] - \Expect[z_{e',c^\bullet}|\mathbf{X}_{i-1}]\right|
\\
\intertext{and since $0 \leq 1-\Expect[z_{e}^b|\mathbf{X}_{i-1}] \leq 1$,}
&\leq \sum_{(e,e') \in P_1 \cup P_3} \left|\Expect[z_{e',c^\bullet}|\mathbf{X}_i] - \Expect[z_{e',c^\bullet}|\mathbf{X}_{i-1}]\right|\\
\intertext{Finally, any edge $e' \in R_1$ is incident to at most 2 edges in $R$, so}
&\leq 2\sum_{e' \in Y^a} \left|\Expect[z_{e',c^\bullet}|\mathbf{X}_i] - \Expect[z_{e',c^\bullet}|\mathbf{X}_{i-1}]\right| \\
&\leq 2 D_i^a.\\
\intertext{For each $e \in Y^b$, we write $B(e)$ to denote the set of all edges $e' \in R_1 \setminus Y^a$ that are incident to $e$,
i.e., $\{e\}\times B(e) \subseteq P_2$.
    Notice that $0 \leq \Expect[\sum_{e' \in B(e)} z_{e',c^\bullet}|\mathbf{X}_{i-1}] = \Expect[\sum_{e' \in B(e')} z_{e',c^\bullet}|\mathbf{X}_{i}] \leq 1$, since $e = \{v^\bullet, u\}$ and all edges in $B(e)$ share the vertex $u$, and so at most one could
    be successfully colored $c^\bullet$.  By definition, none are incident to $e^\star$.  We can now bound $|F_b+Q_2|$ as follows.}
|F_b + Q_2|
&\leq \left| Q_2 + \sum_{e \in Y^b}\Expect[z_{e}^b|\mathbf{X}_i] - \Expect[z_{e}^b|\mathbf{X}_{i-1}] \right|\\
\intertext{According to the definition of $B(e)$ and $Q_2$,}
&\leq \sum_{e \in Y^b} \left|\Expect\left[z_{e}^b\left(1-\sum_{e' \in B(e)}z_{e',c^\bullet}\right)\middle|\mathbf{X}_i\right] - \Expect\left[z_{e}^b\left(1-\sum_{e' \in B(e)}z_{e',c^\bullet}\right) \middle| \mathbf{X}_{i-1}\right]\right| \\
\intertext{For every $e' \in R_1 \setminus Y^a$, we have $\Expect\left[z_{e',c^\bullet} | \mathbf{X}_{i}\right] = \Expect\left[z_{e',c^\bullet} | \mathbf{X}_{i-1}\right]$, which implies}
&\leq \sum_{e \in Y^b} \left(1-\Expect\left[\sum_{e' \in B(e)}z_{e',c^\bullet}\middle|\mathbf{X}_{i-1}\right]\right)
\cdot \left|\Expect[z_{e}^b|\mathbf{X}_i] - \Expect[z_{e}^b|\mathbf{X}_{i-1}]\right|\\
&\leq \sum_{e \in Y^b} \left|\Expect[z_{e}^b|\mathbf{X}_i] - \Expect[z_{e}^b|\mathbf{X}_{i-1}]\right| \\
& = D_i^b.\\
\intertext{Our last task is to bound the absolute value of $Q_3$.}
|Q_3|
&\leq \sum_{(e,e') \in P_3} \left( \Expect[z_{e',c^\bullet} \cdot z_{e}^b|\mathbf{X}_i] + \Expect[z_{e',c^\bullet} \cdot z_{e}^b|\mathbf{X}_{i-1}] \right) \\
&\leq \sum_{(e,e') \in P_3}\left( \Expect[z_{e',c^\bullet}|\mathbf{X}_i] + \Expect[z_{e',c^\bullet}|\mathbf{X}_{i-1}] \right)\\
\intertext{Since any edge $e' \in R_1$ is incident to at most 2 edges in $R$,}
&\leq 2\sum_{e' \in Y^a} \left(\Expect[z_{e',c^\bullet}|\mathbf{X}_i] + \Expect[z_{e',c^\bullet}|\mathbf{X}_{i-1}] \right)\\
&\leq 2 D_i^a. \qedhere
\end{align*}
\end{proof}

\begin{claim}\label{clm-2}
If $e^\star \in R_2$, then $D_i^{a} = O(w(e^\star)/p)$.
If $e^\star \in R \cup R_1$, then $D_i^{a} = O(1)$.
\end{claim}
\begin{proof}
We first consider the case that $e^\star \in R_2$.
In this case $|Y^a| = w(e^\star)$.
Recall that $Y^a \subseteq R_1$, and so all $e \in Y^a$ have not yet decided whether to select $c^\bullet$ when
$X_i$ is revealed.
Therefore, both $\Expect[z_{e,c^\bullet}|\mathbf{X}_i]$ and $\Expect[z_{e,c^\bullet}|\mathbf{X}_{i-1}]$ are within the range $[0, 1/p]$, and so $D_i^{a} = O(w(e^\star)/p)$.
Next, consider the case that $e^\star \in R \cup R_1$.
All edges in $Y^a$ must share a vertex with $e^\star$, and so at most two edges in $Y^a$ can successfully color themselves by
$c^\bullet$.
Hence
\[
D_i^{a} \leq \sum_{e \in Y^a} \left(\Expect[z_{e,c^\bullet}|\mathbf{X}_i]  + \Expect[z_{e,c^\bullet}|\mathbf{X}_{i-1}]\right) \leq 2+2 = 4 =O(1).
\]
\end{proof}

\begin{claim}\label{clm-3}
If $e^\star \in R_1 \cup R_2$, then $D_i^{b} = O(1/p)$.
If $e^\star \in R$, then $D_i^{b} = O(1)$.
\end{claim}
\begin{proof}
Recall that $z_{e}^b = \sum_{c \in \Psi(e)-\{c^\bullet\}} z_{e,c}$ for any edge $e \in Y^b$, and so
$$D_i^b \leq \sum_{e \in Y^b} \  \sum_{c \in \Psi(e)-\{c^\bullet\}} |\Expect[z_{e,c}|\mathbf{X}_i] - \Expect[z_{e,c}|\mathbf{X}_{i-1}]|.$$
We first show that $|\Expect[z_{e,c}|\mathbf{X}_i] - \Expect[z_{e,c}|\mathbf{X}_{i-1}]| = O(1/p^2)$ if $e^\star \neq e$.
We write $k_1$ (resp., $k_2$) to denote the number of  edges incident to $e$ that have decided to select $c^\bullet$ (resp., have decided to not select $c^\bullet$) by the time $X_i$ is revealed.
\[
\Expect[z_{e,c}|\mathbf{X}_{i-1}] =
\begin{cases}
0 &\text{($e$ has decided to select $c^\bullet$)}\\
\frac{1}{p-1} \cdot (1-1/p)^{2t - 1 - k_1 - k_2} \cdot (1-1/(p-1))^{k_2} &\text{($e$ has decided to not select $c^\bullet$)}\\
\frac{1}{p} \cdot (1-1/p)^{2t - 1 - k_1 - k_2} \cdot (1-1/(p-1))^{k_2} &\text{($e$ has not made any decision)}
\end{cases}
\]
In any case, $\Expect[z_{e,c}|\mathbf{X}_{i-1}] = O(1/p)$.
There are two possibilities of $\Expect[z_{e,c}|\mathbf{X}_{i}]$ based on $X_i$, i.e., whether $e^\star$ selects $c^\bullet$.
\[
\Expect[z_{e,c}|\mathbf{X}_{i}] =
\begin{cases}
\Expect[z_{e,c}|\mathbf{X}_{i-1}] / (1-1/p) &\text{($e^\star$ selects $c^\bullet$)}\\
\Expect[z_{e,c}|\mathbf{X}_{i-1}] \cdot (1-1/(p-1)) / (1-1/p)  &\text{($e^\star$ does not select $c^\bullet$)}
\end{cases}
\]
In any case, $|\Expect[z_{e,c}|\mathbf{X}_i] - \Expect[z_{e,c}|\mathbf{X}_{i-1}]| = O(1/p^2)$.
We are now in a position to bound $D_i^{b}$.
For the case that $e^\star \in R_1 \cup  R_2$, we have $|Y^b| \leq 2$ and $e^\star \notin Y^b$, and so  $D_i^{b} \leq 2 \cdot (p-1) \cdot O(1/p^2) = O(1/p)$.
For the case that $e^\star \in R$, we have $|Y^b| = |R| = t$  and $e^\star \in Y^b$, and so  $D_i^{b} \leq 1 + (t-1) \cdot (p-1) \cdot O(1/p^2) = O(1)$.
\end{proof}

\paragraph{Revealing the Color Selected by an Edge in $R \cup R'$.}
Next, we analyze the effect of exposing the value of $X_i$, where $\alpha < i \leq \alpha+\beta$, given that all variables in $\mathbf{X}_{i-1}$ have been fixed.

Observe that $z_{e}^a$, for all $e \in R$, are already determined by $\{ X_j \ : \ j \in [\alpha] \}$.
If $z_{e}^a = 1$, then $z_e = 1$ regardless of the value of $z_{e}^b$; if $z_{e}^a = 0$, then $z_e = z_e^b$.
For those edges $e \in R$ such that $z_e$ is not determined by $\{ X_j \ : \ j \in [\alpha] \}$, the random variable $z_e = z_{e}^b$ behaves the same as $z_e$ in the analysis of concentration of vertex degree,
 so the  analysis in Appendix~\ref{sect:concentration-deg} can be applied here (think of $S = R$ and $S' = R'$).

In more detail, for each edge $e' \in R'$, we define
$w'(e')$ as $\sum_{e \in R, \ \text{$e'$ incident to $e$} } \left| \Psi(e') \cap \Psi(e) - \{c^\bullet\} \right|$.
We have $\sum_{e' \in R'}{w'(e')} \leq |R|(p-1)(t-1) < pt^2$.
Now consider the color $X_i = \Color(e^\star)$ selected by the edge $e^\star \in R \cup R'$.
From the analysis  in Appendix~\ref{sect:concentration-deg}, we infer the following.
\begin{itemize}
\item If $e^\star \in R'$, then $|D_i| = O(1)$ and $\Var[D_i | \mathbf{X}_{i-1}] = O(w'(e^\star)/(pt))$. Hence we can set $\sigma_i^2 = O(w'(e^\star)/(pt))$.
\item If $e^\star \in R$,  then $|D_i| = O(1)$ and $\Var[D_i | \mathbf{X}_{i-1}] = O(1)$. Hence we can set $\sigma_i^2 = O(1)$.
\end{itemize}
Thus, $\sum_{j=\alpha+1}^{\alpha+\beta} \sigma_i^2 = O(t)$, as desired.